\documentclass[10pt,reqno]{amsart}
\usepackage[foot]{amsaddr}
\usepackage{amsmath, amsfonts}
\usepackage{amsthm}
\usepackage{float}
 \usepackage{amssymb}
 \usepackage{amssymb}
\usepackage{bbm}

\DeclareFontFamily{U}{mathx}{\hyphenchar\font45}
\DeclareFontShape{U}{mathx}{m}{n}{
      <5> <6> <7> <8> <9> <10>
      <10.95> <12> <14.4> <17.28> <20.74> <24.88>
      mathx10
      }{}
\DeclareSymbolFont{mathx}{U}{mathx}{m}{n}
\DeclareFontSubstitution{U}{mathx}{m}{n}
\DeclareMathAccent{\widecheck}{0}{mathx}{"71}
\DeclareMathAccent{\wideparen}{0}{mathx}{"75}

\usepackage{hyperref,url}

\usepackage{thmtools}
\usepackage{thm-restate}
\usepackage{amssymb}
\usepackage{amscd}
\usepackage{amsthm}
\usepackage{setspace}
\usepackage{enumerate}
\usepackage{graphicx}
\usepackage{mathtools}
\usepackage{tcolorbox} 
\usepackage{subfigure} 
\usepackage{siunitx}
\usepackage{tikz-cd}
\usepackage{bm,blkarray}
\numberwithin{equation}{section}

\usepackage{mathtools}
\usepackage[tableposition=top]{caption}
\usepackage{booktabs,dcolumn}

\newcommand\eps{\varepsilon}
\DeclareMathOperator{\Vol}{Vol}

\DeclareMathOperator{\Err}{Err}
\DeclareMathOperator{\Ern}{Nerr}

\DeclareMathOperator{\sgn}{sgn}

\newcommand{\quash}[1]{}

\newcommand{\GL}{\mathop{\rm GL}}

\title[The Gauss Circle Problem and Fourier Quasicrystals]{The Gauss Circle Problem and Fourier Quasicrystals}
\author{Roni A. Edwin}
\address{Department of Mathematics, Massachusetts Institute of Technology, \newline \indent Cambridge, MA, 02139 USA.}
\email{raedwin@mit.edu}
\author{Allen Lin}
\address{Department of Mathematics, Massachusetts Institute of Technology, \newline \indent Cambridge, MA, 02139 USA.}
\email{allenees@mit.edu}

\newtheorem{theorem}{Theorem}[section]
\newtheorem{definition}[theorem]{Definition}
\newtheorem{lemma}[theorem]{Lemma}

\newtheorem{proposition}[theorem]{Proposition}

\newtheorem{corollary}[theorem]{Corollary}

\theoremstyle{definition}
\newtheorem{example}{Example}

\begin{document}

\begin{abstract}
    The Gauss circle problem asks for an approximation to the number of lattice points of $\mathbb{Z}^2$ contained in $B_r$, the disk of radius $r$ centered at the origin. Upper, lower, and average bounds have been established for this number-theoretic problem and have been generalized to any lattice in any dimension. We extend this problem to a more general class of structures known as Fourier quasicrystals. Recent work from Alon, Kummer, Kurasov, and Vinzant provides an upper bound $\#(\Lambda \cap B_r) = c_0\mathopen{}\Vol_d\mathopen{}\left(B_r\right)\mathclose{}\mathclose{}+ O\left(r^{d-1}\right)$ for any Fourier quasicrystal $\Lambda \subset \mathbb{R}^d$ of density $c_0$, where $B_r$ is the $d$-dimensional ball of radius $r$. In this paper, we improve the upper bound for any uniformly discrete Fourier quasicrystal, by showing we can write $\#\left(\Lambda \cap B_r\right) = c_0\mathopen{}\Vol_d\mathopen{}\left(B_r\right)\mathclose{}\mathclose{} + O\left(r^{\theta(\Lambda)}\right)$, where $\frac{d-1}{2} < \theta(\Lambda) < d-1$ is some exponent depending on $\Lambda$. In the special case $d = 2$, we also prove lower and upper bounds for the average of the error.
\end{abstract}

\maketitle

\tableofcontents
\section{Introduction}

% Introduce Gauss Circle Problem
% State the current results of the problem using Poisson summation formula

The Gauss circle problem is a classic problem in number theory. It asks for the number of integer lattice points that lie in $B_r$, the disk of radius $r$ centered at the origin. Note $\pi r^2$ is a reasonable approximation because each lattice point is contained in one unit square, and so the real task involves analyzing the error in this approximation. A rigorous formulation is as follows. For a discrete set $X \subset \mathbb{R}^2$ with density $\rho(X)$, we define the corresponding error term
\begin{equation*}
    \Err\left(r,X\right) \coloneqq \#\left(X\cap B_r\right) - \rho(X)\pi r^2
\end{equation*} that measures how good of an approximation $\rho(X)\pi r^2$ is to $\#\left(X\cap B_r\right)$. Here the density $\rho(X)$ is given by
\begin{align*}
    \rho(X)=\lim_{R \to \infty}\frac{\#\left(X\cap B_r\right)}{\pi r^2}\textrm{.}
\end{align*}
With $ \mathbb{Z}^2$, Gauss \cite{gauss} proved that $\Err\left(r,\mathbb{Z}^2\right) = O(r)$. This upper bound was improved to $\Err\left(r,\mathbb{Z}^2\right) = O\left(r^{\frac2{3}}\right)$ by Voronoi \cite{voronoi}, Sierpiński \cite{sierpinski}, and van der Corput \cite{corput}. In 2003, Huxley \cite{huxley2003exponential} proved the current (peer-reviewed) best bound of $\Err\left(r,\mathbb{Z}^2\right) = O\left(r^{\frac{131}{208}}\right)$, though a recent preprint \cite{li2023improvement} reports a marginal improvement of $\Err\left(r,\mathbb{Z}^2\right)=O\left(r^{0.628966351948...+\eps}\right)$ for all $\eps>0$. On the other hand, Hardy \cite{hardy1915expression} showed that $\Err\left(r,\mathbb{Z}^2\right) \neq O\left(r^{\frac1{2}}\right)$ by proving
\begin{equation*}
    \limsup_{r \to \infty}{\frac{|\Err\left(r,\mathbb{Z}^2\right)|}{r^{\frac1{2}}(\log r)^{\frac1{4}}}} > 0\textrm{.}
\end{equation*}
The conjecture on the bound of this problem due to Hardy is $\Err\left(r,\mathbb{Z}^2\right) = O\left(r^{\frac1{2}+\varepsilon}\right)$ for all $\varepsilon > 0$, which he proved \cite{hardy1917average} to be true in an average sense; that is,
\begin{equation*}
    \frac{1}{R}\int_{1}^{R}{|\Err\left(r,\mathbb{Z}^2\right)|\,dr} = O\left(R^{\frac1{2}+\varepsilon}\right)
\end{equation*}
for all $\varepsilon > 0$, though this upper bound is not sharp: Bleher's \cite{10.1215/S0012-7094-92-06718-4} results imply that
\begin{equation*}
    \frac{1}{R}\int_{1}^{R}{|\Err\left(r,\Gamma\right)|\,dr} = O\left(R^{\frac1{2}}\right)
\end{equation*}
for any full-rank lattice $\Gamma \subset \mathbb{R}^2$. The analysis of the error term $\Err\left(r,\mathbb{Z}^2\right)$ to a significant extent relies on the Poisson summation formula, which says for a full rank lattice $\Gamma \subset \mathbb{R}^d$ and any Schwartz function $f$, we have 
\begin{align*}
    \sum_{x \in \Gamma}\widehat{f}(x)=\frac1{\Vol\left(\mathbb{R}^d/\Gamma\right)}\sum_{\omega \in \Gamma^*}f(\omega)\textrm{,}
\end{align*}
where $\Gamma^*=\left\{\omega \in \mathbb{R}^d\;\vert\;\left\langle x,\omega\right\rangle \in \mathbb{Z} \;\textrm{for all}\; x \in \Gamma\right\}$ denotes the dual lattice. The versatile nature of the Poisson summation formulas leads to the natural question of whether these bounds hold for more general structures that admit a similar summation formula. The particular class of structures we examine are known as Fourier quasicrystals.

% Introduce and motivate the definition of Fourier quasicrystals

\subsection{Fourier Quasicrystals}\label{FQdef}
We follow the definition of Fourier quasicrystals given in Definition 1.1 in \cite{alon2024higherdimensionalfourierquasicrystals}. To that end, a discrete set $\Lambda \subset \mathbb{R}^d$ is a Fourier quasicrystal if there exists another discrete set $S \subset \mathbb{R}^d$ and coefficients $\left(c_s\right)_{s \in S}$ indexed by $S$, called the \emph{spectrum} and \emph{Fourier coefficients} of $\Lambda$ respectively, such that for any Schwartz function $f \in \mathcal{S}\left(\mathbb{R}^d\right)$,
\begin{align}
    \sum_{\lambda \in \Lambda}\widehat{f}(\lambda)=\sum_{s \in S}{c_sf(s)}\textrm{,} \label{summationformulaunifdis}
\end{align}
and the complex coefficients $\left(c_s\right)_{s \in S}$ satisfy the polynomial growth condition
\begin{align*}
    \#\left(\Lambda \cap B_r\right)+\sum_{s \in S \cap B_r}\left|c_s\right|=O\left(r^{P}\right)
\end{align*} 
for some $P > 0$. In that case, we say $\Lambda$ has coefficient growth rate $P$. In a similar vein, we call $N \in [0,\infty]$ a growth rate of the spectrum $S$ provided $\#\left(S\cap B_r\right)=O\left(r^N\right)$.
% consider the following parameters of $\Lambda$: 
% \begin{definition}[Parameters of $\Lambda$]
% Let $\Lambda \subset \mathbb{R}^d$ be a Fourier quasicrystal with summation formula  \begin{align}
% \sum_{\lambda \in \Lambda}\widehat{f}(\lambda)=\sum_{s \in S}c_sf(s)\textrm{,}
%     %\label{summationformulaunifdis}
% \end{align} valid for all $f \in \mathcal{S}\left(\mathbb{R}^d\right)$. We define
% \begin{itemize}
%     \item the \emph{temperedness exponent of $\Lambda$} $P > 0$ for which 
%     \begin{equation*}
%         \sum_{s \in S\cap B_r}\left|c_s\right|=O\left(r^P\right)\textrm{;}
%     \end{equation*}
%     \item the \emph{growth rate of the spectrum} $N$ by $\#(S \cap B_r) = O\left(r^N\right)$;
%     \item the \emph{separation of the spectrum} $\delta$ around $0$ for which $S \cap B_\delta = \{0\}$;
%     \item the \emph{separation of the Fourier quasicrystal} $\eta$ by $\eta \coloneqq \inf_{\substack{x,y \in \Lambda \\ x \neq y}}{\lVert x-y\rVert}$. 
% \end{itemize}
% We take $N = \infty$ if the growth rate is faster than any polynomial. If $\eta > 0$, we say $\Lambda$ is \emph{uniformly discrete}.
% \label{parameters of Lambda}
%\end{definition}

It is worth asking if there are any non-trivial examples of Fourier quasicrystals, by which we mean sets that are not finite unions of translates of lattices. For $d=1$, Kurasov and Sarnak \cite{kurasov2020stable} gave the first
examples of uniformly discrete Fourier quasicrystals that are not periodic structures. As a consequence,
they also obtained examples of uniformly discrete one-dimensional Fourier quasicrsytals that meet any
arithmetic progression in a finite number of points. In all dimensions, Alon, Kummer, Kurasov and
Vinzant \cite{alon2024higherdimensionalfourierquasicrystals} constructed uniformly discrete Fourier quasicrystals that intersect periodic configurations in a finite number
points. Moreover, the spectrum $S$ of the Fourier quasicrystals coming from their construction satisfy the growth rate $\#\left(S\cap B_r\right)=Cr^N+O\left(r^{N-1}\right)$, where $N$ is the dimension of the ambient space used in constructing the Fourier quasicrystal. 
Thus, there exist Fourier quasicrystals that are not finite unions of lattices, and we refer to
such Fourier quasicrystals as \emph{non-trivial Fourier quasicrystals}. 

When $N=2$ and $\Lambda$ is uniformly discrete, $\Lambda$ is essentially a periodic configuration. In turns out in this case that \begin{align}
    \sum_{s \in S\cap B_r}\left|c_s\right|=O\left(r^2\right)\textrm{,}
    \label{movie}
\end{align} which can be seen by noting that $\left|c_s\right|\le c_0$ from the summation formula \eqref{summationformulaunifdis}, or from Theorem ~\ref{thm:pointwiseerrorbound} to be introduced. 
Under the condition in \eqref{movie}, Theorem $5$ in \cite{7ea18832-207d-3dde-8cab-263c2057d2e7} implies that $\Lambda$ is a finite union of translates of several full-rank lattices. 
\begin{example}
    Consider the Fourier quasicrystal given by
    \begin{align}\label{eqn:fq-voronoi-formula}
        \begin{split}
            \Lambda &= \left\{(x,y) \in \mathbb{R}^2\;\vert\;p(x,y) = 0, \ q(x,y)=0\right\}\textrm{,} \\
            p(x,y) &= \sin\left(\pi^2\left(-\frac{1}{34}x-\frac{\sqrt{11}}{56}y\right)\right)\cos\left(3\pi\left(\frac{\sqrt{7}}{65}x-\frac{1}{10}y\right)\right)\\
            &\quad+\sin\left(3\pi\left(\frac{\sqrt{7}}{65}x-\frac{1}{10}y\right)\right)\left(\sin\left(\pi^2\left(-\frac{1}{34}x-\frac{\sqrt{11}}{56}y\right)\right) +\cos\left(\pi^2\left(-\frac{1}{34}x-\frac{\sqrt{11}}{56}y\right)\right)\right)\textrm{,} \\
            q(x,y) &= \sin\left(\pi^2\left(-\frac{1}{34}x-\frac{\sqrt{11}}{56}y\right)\right)\cos\left(2\pi \left(\frac{\sqrt{5}}{18}x+\frac{\sqrt{3}}{50}y\right)\right) -\sin\left(2\pi \left(\frac{\sqrt{5}}{18}x+\frac{\sqrt{3}}{50}y\right)\right) \\
            &\quad \times \left(2\sin\left(\pi^2\left(-\frac{1}{34}x-\frac{\sqrt{11}}{56}y\right)\right)+\cos\left(\pi^2\left(-\frac{1}{34}x-\frac{\sqrt{11}}{56}y\right) \right)\right)\textrm{.}
        \end{split}
    \end{align}    
    This Fourier quasicrystal comes from the constructions in \cite{alon2024higherdimensionalfourierquasicrystals}. Specifically, the set $\Lambda$ is precisely
    \begin{equation*}
        \Lambda = \left\{x \in \mathbb{R}^2\;\vert\;\exp(2\pi iLx) \in X\right\}\textrm{,}
    \end{equation*}
    where $L$ is the matrix
    \begin{equation*}
        L = \begin{pmatrix}
            -\pi/34 & -\sqrt{11}\pi/56 \\
            3\sqrt{7}/65 & -3/10 \\
            \sqrt{5}/9 & \sqrt{3}/25
        \end{pmatrix}
    \end{equation*}
    and $X \subseteq \mathbb{C}^3$ is the algebraic variety defined by the equations
    \begin{align*}
        f_1(z_1,z_2,z_3) &= (1-2i) - z_1 - z_2 + (1+2i)z_1z_2 = 0\textrm{,} \\
        f_2(z_1,z_2,z_3) &= 1 - (1+i)z_1 - (1-i)z_3 + z_1z_3 = 0\textrm{.}
    \end{align*}
    The spectrum $S$ is a subset of the set
    \begin{equation*}
        \left\{L^tk\;\vert\;k\in\mathbb{Z}^3\;\textrm{with at most}\;1\;\text{sign change in}\;k\;\text{discounting zeros}\right\}\textrm{.}
    \end{equation*}
    Figure \ref{fig:nontrivial-FQ} shows $\Lambda$ given by \eqref{eqn:fq-voronoi-formula}.

    \begin{figure}[ht]
        \centering
        \includegraphics[width=0.4\linewidth]{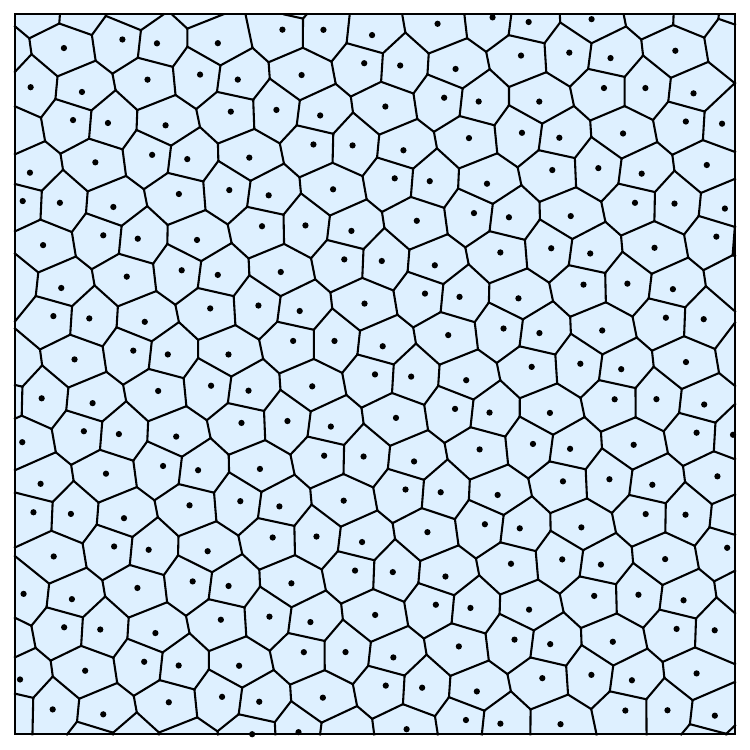}
        \includegraphics[width=0.4\linewidth]{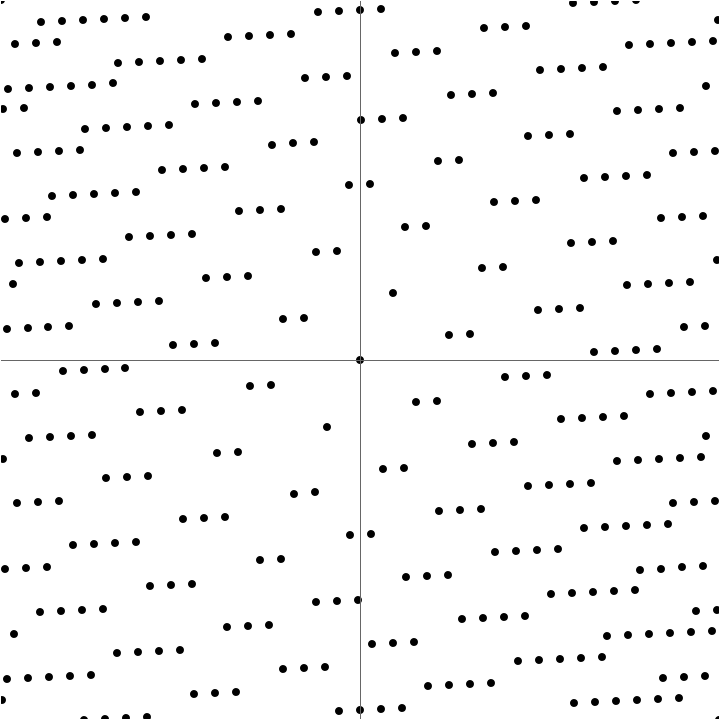}
        \caption{A non-trivial Fourier quasicrystal $\Lambda$ (left) as in \eqref{eqn:fq-voronoi-formula} with its spectrum $S$ (right)}
        \label{fig:nontrivial-FQ}
    \end{figure}
    
    \begin{figure}[H]
        \centering
        \includegraphics[width=0.5\linewidth]{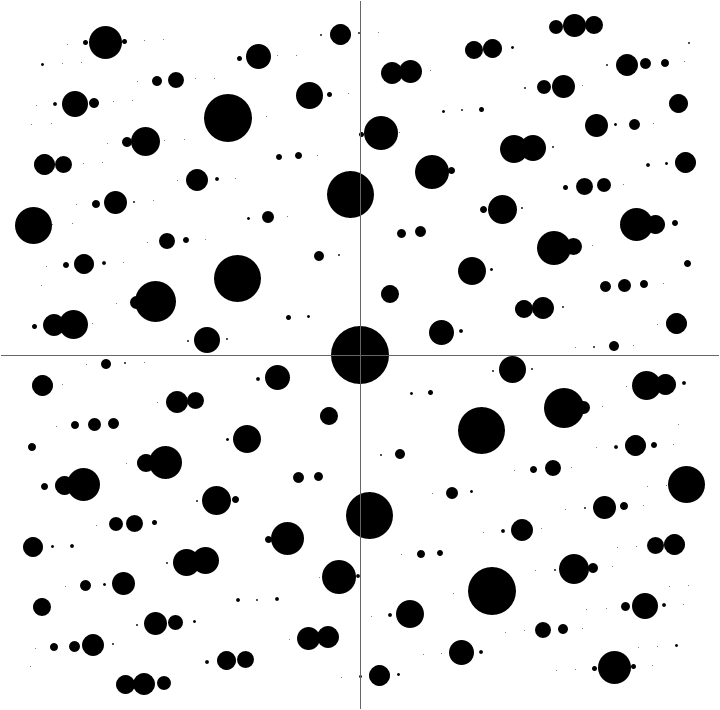}
        \caption{The spectrum $S$ of $\Lambda$ in Figure \ref{fig:nontrivial-FQ}, where each $s \in S$ is represented by a disk with radius proportional to $\left|c_s\right|$}
        \label{fig:size-coefficients}
    \end{figure}
\end{example}
The authors in \cite{alon2024higherdimensionalfourierquasicrystals} give a number of properties of Fourier quasicrystals in $\mathbb{R}^d$, among which they showed that $\#\left(\Lambda \cap B_r\right)=c_0\mathopen{}\Vol_d\mathopen{}\left(B_r\right)\mathclose{}\mathclose{}+O\left(r^{d-1}\right)$ for any Fourier quasicrystal $\Lambda \subset \mathbb{R}^d$. Here $B_r$ denotes the closed ball of radius $r$ in $\mathbb{R}^d$, and $\Vol_d$ the Lebesgue measure on $\mathbb{R}^d$. It is natural to ask how tight this bound is, and this among with the results of the Gauss circle problem in the plane motivate our generalization of the Gauss circle problem to Fourier quasicrystals. Like with the case of $\mathbb{Z}^2$, for a uniformly discrete\footnote{A set $X\subset \mathbb{R}^d$ is uniformly discrete if there exists a $\delta>0$ such that $\left\lVert x_1-x_2\right\rVert\ge \delta$ for any distinct $x_1,x_2 \in X$.} Fourier quasicrystal $\Lambda \subset \mathbb{R}^2$ with summation formula as in \eqref{summationformulaunifdis}, we define its \emph{error term} $\Err(r,\Lambda)$ for $r>0$ by \begin{align}
\Err(r,\Lambda)\coloneqq \#\left(\Lambda \cap B_r\right)-c_0\pi r^2\textrm{,}
    \label{errortermdef}
\end{align}
We also define the \emph{normalized error term} $\Ern(r,\Lambda)$ by
\begin{align}
    \Ern\left(r,\Lambda\right)\coloneqq \frac{\#\left(\Lambda \cap B_r\right)-c_0\pi r^2}{r^{\frac{1}{2}}}\textrm{,}
    \label{normalisederrortermdef}
\end{align}
so $\Ern\left(r,\Lambda\right)=r^{-\frac{1}{2}}\Err(r,\Lambda)$. It is often more convenient to work with $\Ern$ when comparing the error term to $r^{\frac1{2}}$, since that is believed to be the right order of growth for the error. Like with the case of $\mathbb{Z}^2$, we similarly concern ourselves with bounding $\Err(r,\Lambda)$. We were able to show an upper bound of the form $\Err\left(r,\Lambda\right)=r^{\Theta(\Lambda)}$ for any uniformly discrete Fourier quasicrystal $\Lambda \in \mathbb{R}^2$, with the exponent $\Theta(\Lambda) \in \left(\frac{1}{2},1\right)$. We further developed lower and upper bounds for the average of the error.

% The following questions come to mind:\newline

%  Suppose that $\Lambda \subset \mathbb{R}^2$ is any uniformly discrete Fourier quasicrystal with summation
% formula as in \eqref{summationformulaunifdis}, and let $\Err(r,\Lambda)$ as defined in \eqref{errortermdef}
% be the error in the approximation.
% \begin{enumerate}
%     \item\label{q1} Can we obtain pointwise bounds of the form $\Err\left(r,\Lambda\right)=O\left(r^{\theta}\right)$ for some exponent $\theta \in (0,1)$? In particular, can we achieve the classical exponent of $\frac2{3}$? 
%     \item\label{q2} Can we also obtain lower bounds of the form $\Err\left(r,\Lambda\right)=\Omega\left(r^{\alpha}\right)$, by which we mean \begin{align*}
%         \limsup_{r \to \infty }\frac{\left|\Err\left(r,\Lambda\right)\right|}{r^\alpha}>0?
%     \end{align*}
%     \item\label{q3} What is the behavior of $\Err(r,\Lambda)$ on average?
% \end{enumerate}

% The case where $\Lambda$ is a lattice is studied in \cite{10.1215/S0012-7094-92-06718-4}, though phrased in terms of the bounding curve. 
We first describe previous works (\cite{10.1215/S0012-7094-92-06718-4} and \cite{bleher1993distribution}) on lattices. The authors in \cite{10.1215/S0012-7094-92-06718-4} instead consider the corresponding normalized error term as in \eqref{normalisederrortermdef} where $B_r$ is replaced by an oval scaled by a factor of $r$. Note this includes the analysis of $\Err\left(r,\Gamma\right)$ for a lattice $\Gamma$, since the image of a circle under an invertible linear transformation $L\colon\mathbb{R}^2 \to \mathbb{R}^2$ is an oval (per the definition in \cite{10.1215/S0012-7094-92-06718-4}). To that end, the main results in \cite{10.1215/S0012-7094-92-06718-4} can be rephrased as follows: When $\Lambda = L\mathbb{Z}^2 + \alpha$ is a lattice shifted by some translation $\alpha \in \mathbb{R}^2$ and linear transformation $L \in \GL_2\left(\mathbb{R}\right)$, the normalized error term $\Ern\left(r,L\mathbb{Z}^2+\alpha\right)$ belongs to the Besicovitch space (see \cite{besicovitch1926generalized}) $B^2$ of almost periodic functions, which means that $\Ern\left(r,L\mathbb{Z}^2+\alpha\right)$ has a limiting probability distribution $\nu_{L,\alpha}$, in that for every bounded continuous function $g \in \mathcal{C}^0\mathopen{}\left(\mathbb{R}\right)\mathclose{}$, we have 
\begin{align}
\lim_{R \to \infty}\frac1{R}\int_0^Rg\left(\Ern\left(r,L\mathbb{Z}^2+\alpha\right)\right)\;\textup{d}r=\int_{\mathbb{R}}g(x)\;\textup{d}\nu_{L,\alpha}(x)\textrm{.}
    \label{limitingdistribution}
\end{align}
Additionally, they showed the limiting probability distribution $\nu_{L,\alpha}$ has mean $0$ and finite variance $\sigma\left(\alpha;L\right)^2$.
% normalized error term has $0$ mean on average and finite variance, so 
% \begin{align}
% \lim_{R \to \infty}\int_0^R\frac{1}{R}\rho\left(\frac{r}{R}\right)\Ern\left(r,L\mathbb{Z}^2+\alpha\right)\,dr=\int_{\mathbb{R}}x\,d\nu_{L,\alpha}(x)=0\textrm{,}
%     \label{zeromean}
% \end{align}
% and
% \begin{align}
%     \lim_{R \to \infty}\frac{1}{R}\int_0^R\rho\left(\frac{r}{R}\right)\left|\Ern\left(r,L\mathbb{Z}^2+\alpha\right)\right|^2\,dr=\int_{\mathbb{R}}x^2\,d\nu_{L,\alpha}(x)=\sigma\left(\alpha;L\right)^2
%     \label{finitevariance}
% \end{align}
% for some finite $\sigma\left(\alpha;L\right)^2$ and $\rho$ is an arbitrary probability density function on $[0,1]$.
In the case where $\Lambda$ is simply $\mathbb{Z}^2$ shifted by some translation $\alpha \in [0,1)^2$, the authors in \cite{bleher1993distribution} additionally show the distribution $\nu_{I_2,\alpha}$ is unbounded, so
\begin{align}
    \limsup_{r \to \infty}\frac{\left|\Err\left(r,\mathbb{Z}^2+\alpha\right)\right|}{r^{\frac1{2}}}=\infty\textrm{.}
    \label{errortermisunbounded}
\end{align}
Our work involved trying to generalize these results to any uniformly discrete Fourier quasicrystal $\Lambda \subset \mathbb{R}^2$, with some success. Our main results are as follows.

\subsection{Main Results in \texorpdfstring{$\mathbb{R}^2$}{Lg}}
 Our first two results concern the pointwise error bound. 
% \begin{restatable}{thm}{pointwiseerrorboundgen}
%      Let $\Lambda \subset \mathbb{R}^d$ be any uniformly discrete Fourier quasicrystal with temperedness exponent $P$, and suuppose its spectrum $S$ has growth rate $N$, so $\#\left(S\cap B_r^d\right)=O\left(r^N\right)$. Then \begin{align*}
%     \#\left(\Lambda \cap B_r^d\right)=c_0\Vol_d\mathopen{}\left(B_r^d\right)\mathclose{}+O\left(r^{\frac{\left(2P-d\right)(d-1)}{2P-(d-1)}}\right)\textrm{.}
% \end{align*} Moreover, if $N$ is finite, we can take $P= \frac{N+d}{2}$.
%     \label{thm:pointwiseerrorboundgen}
% \end{restatable}
\begin{restatable}{thm}{pointwiseerrorbound}
    Let $\Lambda \subset \mathbb{R}^2$ be any Fourier quasicrystal, and suppose $\Lambda$ has coefficient growth rate $P$. Then \begin{align*}
    \#\left(\Lambda \cap B_r\right)=c_0\pi r^2+O\left(r^{1-\frac{1}{2P-1}}\right)\textrm{.}
\end{align*}
% If there exists a finite $N$ such that $\#\left(S\cap B_r\right)=O\left(r^N\right)$, then we
% Moreover, if spectrum $S$ has finite growth rate $N$, so $\#\left(S\cap B_r\right)=O\left(r^N\right)$, we can take $P= 1+\frac{N}{2}$.
    \label{thm:pointwiseerrorbound}
\end{restatable}
For example, when $\Lambda$ is a lattice, $P=2$ and so we get $\Err\left(r,\Lambda\right)=O\left(r^{\frac2{3}}\right)$.

It is interesting to note we can recover the exponent of $\frac2{3}$ on the Fourier space for all coefficient growth rates $P$ rather than the physical space.
\begin{restatable}{thm}{boundsoncsquared}
      Let $\Lambda \subset \mathbb{R}^2$ be any uniformly discrete Fourier quasicrystal with summation formula as in \eqref{summationformulaunifdis}. Then
    \begin{equation*}
        \sum_{s \in S \cap B_r}{\left|c_s\right|^2} = c_0\pi r^2 + O\left(r^{\frac2{3}}\right)\textrm{.}
    \end{equation*}
    \label{thm:boundsoncsquared}
\end{restatable}
Applying Cauchy-Schwarz, we get the following corollary:\begin{corollary}
Let $\Lambda$ be a uniformly discrete Fourier quasicrystal with summation formula as in \eqref{summationformulaunifdis}. If there exists some $N \in \mathbb{N}$ such that $\#\left(S\cap B_r\right)=O\left(r^N\right)$ (like the constructions in \cite{alon2024higherdimensionalfourierquasicrystals}), then
    \begin{align*}
        \sum_{s \in S\cap B_r}\left|c_s\right|=O\left(r^{1+\frac{N}{2}}\right)\textrm{.}
    \end{align*}
\label{lovemejeje}
\end{corollary} 
In that case, the bound in Theorem \ref{thm:pointwiseerrorbound} reduces to \begin{align*}
    \#\left(\Lambda\cap B_r\right)=c_0\pi r^2+O\left(r^{\frac{N}{N+1}}\right)\textrm{.}
\end{align*}
We can improve the bound given in Theorem \ref{thm:pointwiseerrorbound} by instead considering the error $\Err\left(r,\Lambda\right)$ in an average sense. We can note the fact that the distribution $\nu_{L,\alpha}$ in \eqref{limitingdistribution} has finite variance and mean $0$ implies when $\Lambda$ is a lattice, the average of $\left|\Err(r,\Lambda)\right|$ over the interval $[0,R]$ is upper bounded by (up to a constant factor) $R^{\frac1{2}}$. This next theorem generalizes this observation:
\begin{restatable}{thm}{averageerrorbound}
      Let $\Lambda \subset \mathbb{R}^2$ be any uniformly discrete Fourier quasicrystal with spectrum growth rate $N \in \mathbb{N}$. Then
    \begin{equation*}
        \frac{1}{R}\int_{0}^{R}{|\Err\left(r,\Lambda\right)|\;\textup{d}r} = O\left(R^{\frac1{2}} + R^{\frac{3N-4}{3N-1}}\right)\textrm{.}
    \end{equation*}
    \label{thm:averageerrorbound}
\end{restatable}
As noted earlier, when $\Lambda$ is a lattice we can take $N=2$, and the upper bound reduces to $O\left(R^{\frac1{2}}\right)$. We also obtained a lower bound on the error term $\left|\Err\left(r,\Lambda\right)\right|$, illustrated in the following theorem. The problem is radial in nature, so we introduce the following radial quantities $\ell(\gamma)$ and $S_{\textrm{rad}}$ given by

\begin{align}
     \ell\left(\gamma\right)&=\sum_{\substack{s \in S \\ \lVert s\rVert=\gamma}}c_s\textrm{,}  \label{agammadef}\\
        S_{\textrm{rad}}&=\left\{\lVert s\rVert \;\Big\vert\; s \in S\setminus\{0\}\right\}. \label{Srad}
    \end{align}
\begin{restatable}{thm}{lowerboundLtwoaverage}
     Let $\Lambda \subset \mathbb{R}^2$ be any Fourier quasicrystal, with summation formula as in \eqref{summationformulaunifdis}.  Then
    \begin{equation*}
        \liminf_{R \to \infty}{\frac{1}{R}\int_{1}^{R}{|\Ern\left(r,\Lambda\right)|^2\;\textup{d}r}} \geq \frac{1}{2\pi^2}\sum_{\gamma \in S_\textrm{rad}} \frac{\left|\ell(\gamma)\right|^2}{\gamma^3}> 0\textrm{.}
    \end{equation*} In particular, this implies $\Err\left(r,\Lambda\right)=\Omega\left(r^{\frac1{2}}\right)$.
    \label{thm:lowerboundLtwoaverage}
\end{restatable}
It is worth remarking about the sharpness of our results. It is currently unclear to us if the upper bounds in Theorem \ref{thm:averageerrorbound} are uniformly sharp. For small values of growth rate $N$, say $N=3,4$ we do not believe they are sharp,  as heuristic arguments and numerics seem to suggest the exponent $\frac1{2}$ may be the best.

\begin{example}
    Figure \ref{fig:numerics-main}, for example, shows the behavior of $\int_0^R \left|\Err\left(r,\Lambda\right)\right|\;\textup{d}r$ as a function of $R$ for the Fourier quasicrystal $\Lambda$ given by
    \begin{align}\label{eqn:fq-example}
        \begin{split}
            \Lambda &= \left\{(x,y) \in \mathbb{R}^2 \;\vert\; f(x,y) = 0, \ g(x,y) = 0\right\}\textrm{,} \\
            f(x,y) &= \sin(\pi x)\cos(\pi y) + \sin(\pi y)(\sin(\pi x)+\cos(\pi x))\textrm{,} \\
            g(x,y) &= \sin(\pi x)\cos\left(\pi\left(-\sqrt{2}x+\sqrt{3}y\right)\right) - \sin\left(\pi\left(-\sqrt{2}x+\sqrt{3}y\right)\right)(2 \sin(\pi x)+\cos(\pi x))\textrm{.}
        \end{split}
    \end{align}
    We believe in this case the error would grow slower than that of $\mathbb{Z}^2$, since this Fourier quasicrystal does not have the radial symmetry $\mathbb{Z}^2$ has.  We refer the reader to Appendix \ref{appendixA} for further numerics.
    \begin{figure}[ht]
        \centering
        \includegraphics[width=0.9\linewidth]{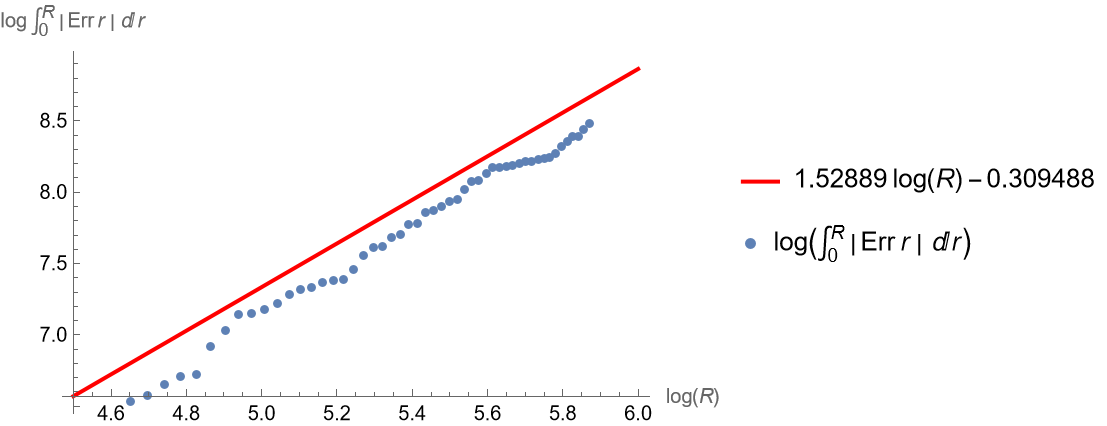}
        \caption{The $\log$-$\log$ plots of $\int_{0}^{R}{\Err(r,\Lambda)\;\textup{d}r}$ and $0.73382R^{1.52889}$ with respect to $R$ of the Fourier quasicrystal $\Lambda$ given by \eqref{eqn:fq-example}. Here $N = 3$.}
        \label{fig:numerics-main}
    \end{figure}
\end{example}

Another issue concerns the sharpness of the pointwise error bound. While the result in Theorem \ref{thm:lowerboundLtwoaverage} implies \begin{align*}
    \limsup_{r \to \infty}\frac{\left|\Err\left(r,\Lambda\right)\right|}{r^{\frac1{2}}}>0,
\end{align*} the limit superior in the left-hand side may not be finite; as stated earlier, it is infinite if $\Lambda=\mathbb{Z}^2+\alpha$, but we do not know if this holds in general. To that end, a natural question to ask is whether there exists a uniformly discrete Fourier quasicrystal $\Lambda \subset \mathbb{R}^2$ for which \begin{align*}
    \limsup_{r \to \infty}\frac{\left|\Err\left(r,\Lambda\right)\right|}{r^{\frac1{2}}}<\infty\textrm{.}
\end{align*}

% One approach we employed was to analyze the higher moments of the normalized error term, so \begin{align*}
%     \limsup_{R \to \infty}\frac1{R}\int_1^R\left|\Ern\left(r,\Lambda\right)\right|^{2k}dr,
% \end{align*} for large $k$. The authors in \cite{bleher1993distribution} remark on the existence of this quantity for $k\le 9$ when $\Lambda=\mathbb{Z}^2$, proven by Heath-Brown \cite{D1992}. We did not work out the details, but it appears Theorem \ref{thm:lowerboundLtwoaverage} could generalize to higher moments by \begin{align}
%     \liminf_{R \to \infty}\int_1^R\left|\Ern\left(r,\Lambda\right)\right|^{2k}dr\ge b_k \sum_{\gamma \in S^k_{\textrm{rad}}}\left|\sum_{\gamma_1+\cdots\gamma_k=\gamma}\prod_{j=1}^k\frac{\ell\left(\gamma_j\right)}{\gamma_j^{\frac3{2}}}\right|^2,
%     \label{icebreaker}
% \end{align} where \begin{align*}
%     S_{\textrm{rad}}^k=\left\{\sum_{j=1}^k\lVert s_j\rVert \;\Big\vert\; s_j \in S\textrm{,}  \  \lVert |s_j\rVert \;\text{distinct}\right\}\textrm{,}
% \end{align*}
% $\ell(\gamma)$ is as in \eqref{agammadef}, and $b_k$ is some constant depending only on $k$. We then tried to show the right-hand side is infinite, or at least grows faster than $C^k$ as $k \to \infty$ for any $C>0$. This approach however did not seem to pan out as we hoped; it seems such an analysis is best suited to when the coefficients $\left(c_s\right)_{s \in S}$ are non-negative.

\subsection{Results in Higher dimensions}
Before going over our results in higher dimensions, let us first go over results in higher dimensions for the lattice $\mathbb{Z}^d$ and its translates, and point out some of the subtleties. For example, for the integer lattice $\mathbb{Z}^d$, $d=3$ appears to be the most difficult; the state of the art results on this can be summarised in the following inequalities:
\begin{align*}
    \#\left(\mathbb{Z}^3\cap B_r\right)=\Vol_3\mathopen{}\left(B_r\right)\mathclose{}+O\left(r^{\frac{42}{32}+\eps}\right),
\end{align*} due to Heath-Brown \cite{heath1999lattice}, and \begin{align*}
    \#\left(\mathbb{Z}^3\cap B_r\right)=\Vol_3\mathopen{}\left(B_r\right)\mathclose{}+\Omega_{\pm}\left(r\sqrt{\log r}\right)
\end{align*} due to K.M. Tsang \cite{tsang2000counting}. Here $F(x)=\Omega_+\left(G(x)\right)$ means that  $\limsup_{x \to \infty}\frac{F(x)}{G(x)}>0$, $F(x)=\Omega_{-}\left(G(x)\right)$ means $-F(x)=\Omega_+\left(G(x)\right)$. $F(x)=\Omega_{\pm}\left(G(x)\right)$
says that both of these assertions are true, and $F(x)=\Omega\left(G(x)\right)$ means $F(x)\neq o\left(G(x)\right)$.
For larger values of $d$, the problem is essentially solved, in the sense that the error is upper bounded by (up to a constant factor) $r^{d-2}$, and least $Cr^{d-2}$ for arbitrarily large $r$. Formally, \begin{align}
    \#\left(\mathbb{Z}^d\cap B_r\right)&=\Vol_d\mathopen{}\left(B_r\right)\mathclose{}+O\left(r^{d-2}\right), \label{donna0} \\
     \#\left(\mathbb{Z}^d\cap B_r\right)&=\Vol_d\mathopen{}\left(B_r\right)\mathclose{}+\Omega_{\pm}\left(r^{d-2}\right).\label{donna1}
\end{align} 
See the survey article \cite{ivic2004lattice} for more on this. One might naturally wonder what happens if instead of $\mathbb{Z}^d$, we consider a shifted lattice $\mathbb{Z}^d+\alpha$; similar to the work in \cite{bleher1993distribution}, Bleher and Bourgain \cite{bleher1996distribution} initiate a study of the normalized error term for the shifted lattice in $\mathbb{R}^d$, which we will denote by \begin{align*}
    \Ern\left(r,\mathbb{Z}^d+\alpha\right)\coloneqq \frac{\#\left(\left(\mathbb{Z}^d+\alpha\right)\cap B_r\right)-\Vol\left(B_r\right)}{r^{\frac{d-1}{2}}}\textrm{,}
\end{align*} following their normalization. Unlike with $d=2$, the behavior of this error term is highly sensitive to choice of shift $\alpha \in \mathbb{R}^d$. For example, the results in \eqref{donna0} and \eqref{donna1} show for $\alpha=0$ and $d\ge 5$, $\Ern\left(r,\mathbb{Z}^d\right)$ grows like $r^{\frac{d-3}{2}}$. On the other hand, Bleher and Bourgain \cite{bleher1996distribution} showed that under an appropriately ``irrational" shift $\alpha$, $\Ern\left(r,\mathbb{Z}^d+\alpha\right)$ retains the Besicovitch $B^2$ almost periodicity observed in $\mathbb{R}^2$. This peculiarity in high dimensions suggests the behavior of the error term is highly sensitive to the radial symmetry of the Fourier quasicrystal; a lack of radial symmetry heuristically means the error behaves less erratically. For comparison, our Theorem \ref{thm:pointwiseerrorbound} generalizes as follows.
\begin{restatable}{thm}{pointwiseerrorboundgen}
     Let $\Lambda \subset \mathbb{R}^d$ be any Fourier quasicrystal, and suppose its coefficients have growth rate $P$. Then \begin{align*}
    \#\left(\Lambda \cap B_r\right)=c_0\mathopen{}\Vol_d\mathopen{}\left(B_r\right)\mathclose{}\mathclose{}+O\left(r^{\frac{\left(2P-d\right)(d-1)}{2P-(d-1)}}\right)\textrm{.}
\end{align*}
    \label{thm:pointwiseerrorboundgen}
\end{restatable} 
Similarly, Theorem \ref{thm:boundsoncsquared} enjoys the following generalization:
\begin{restatable}{thm}{boundsoncsquaredgen}
      Let $\Lambda \subset \mathbb{R}^d$ be any uniformly discrete Fourier quasicrystal with summation formula as in \eqref{summationformulaunifdis}. Then
    \begin{equation*}
        \sum_{s \in S \cap B_r}{\left|c_s\right|^2} = c_0\mathopen{}\Vol_d\mathopen{}\left(B_r\right)\mathclose{} + O\left(r^{\frac{d(d-1)}{d+1}}\right)\textrm{.}
    \end{equation*}
    \label{thm:boundsoncsquaredgen}
\end{restatable}
We get the following corollary by applying Cauchy-Schwarz.
\begin{corollary}
    Let $\Lambda$ be a uniformly discrete Fourier quasicrystal with summation formula as in \eqref{summationformulaunifdis}. If there exists some $N \in \mathbb{N}$ such that $\#(S\cap B_r) = O\left(r^N\right)$, then
    \begin{equation*}
        \sum_{s \in S \cap B_r}{|c_s|} = O\left(r^{\frac{N+d}{2}}\right)\textrm{.}
    \end{equation*}
\end{corollary}
In that case, the bound in Theorem \ref{thm:pointwiseerrorboundgen} reduces to
\begin{equation*}
    \#(\Lambda \cap B_r) = c_0\Vol_d(B_r) + O\left(r^{\frac{N}{N+1} \cdot (d-1)}\right)\textrm{.}
\end{equation*}
Unlike Theorems \ref{thm:pointwiseerrorbound} and \ref{thm:boundsoncsquared}, Theorems \ref{thm:averageerrorbound} and \ref{thm:lowerboundLtwoaverage} do not generalize analogously, though we conjecture that
\begin{align*}
    \#\left(\Lambda \cap B_r\right)=c_0\mathopen{}\Vol_d\mathopen{}\left(B_r\right)\mathclose{}\mathclose{}+\Omega\left(r^{\frac{d-1}{2}}\right)\textrm{.}
\end{align*}

The breakdown of the paper is as follows. In Section \ref{prel} we introduce a number useful lemmas used in proving our main results. In Section \ref{ptwise} we develop the pointwise error bounds (Theorems ~\ref{thm:pointwiseerrorboundgen} and \ref{thm:boundsoncsquaredgen}). In the proof, we introduce a smoothed out version (on a scale of some parameter $t \in (0,1)$) of the error term, which is more amenable to analytic methods. In Section \ref{expofsmooth}, we develop an expansion for this smoothed out error term, with a lower order term independent of the parameter $t$. Then in Section \ref{avg} we develop upper bounds for the average of the error (Theorem ~\ref{thm:averageerrorbound}). Finally in Section \ref{lb}, we develop the lower bound for the normalized error (Theorem ~\ref{thm:lowerboundLtwoaverage}) for $d = 2$. Both Sections \ref{avg} and \ref{lb} rely on the expansions developed in Section \ref{expofsmooth}.

\section{Preliminaries}\label{prel}
\subsection{Notations and Conventions}\label{notationsandconventions}
Throughout the rest of the paper $\Lambda \subset \mathbb{R}^d$ (we will largely focus on $d=2$) denotes any Fourier quasicrystal with the summation formula \begin{align*}
    \sum_{\lambda \in \Lambda}\widehat{f}(\lambda)=\sum_{s \in S}c_sf(s)\textrm{,}
\end{align*}
and spectrum growth rate $N \in \mathbb{N} \cup \{\infty\}$ as defined in Subsection ~\ref{FQdef}. We also use the notation $f(x)\lesssim g(x)$ to mean $f(x)\le Cg(x)$ for some constant $C>0$ independent of the variable $x$, and $f(x)\lesssim_{a_1, \ldots,a_n } g(x)$ to mean $f(x)\le Cg(x)$  where the constant $C$ depends on the parameters $a_1,\ldots,a_n$.
% We denote by $\mathcal{M}\mathopen{}\left(\mathbb{R}^d\right)\mathclose{}$ the space of rapidly decreasing functions on $\mathbb{R}^d$, so \begin{equation}
% \mathcal{M}\mathopen{}\left(\mathbb{R}^d\right)\mathclose{}=\left\{f:\mathbb{R}^d \to \mathbb{C} \ \Big\vert \ \sup_{x \in \mathbb{R}^d}|x|^{\beta}\left|f(x)\right|<\infty\;\textrm{for all}\; \beta\ge 0\right\}\textrm{.}
%     \label{rapidlydecreasingfunctions}
% \end{equation} 
For any $x \in \mathbb{R}^d$, we define the \emph{Japanese bracket} as $\langle x \rangle = \left(1+\lVert x\rVert ^2\right)^{\frac1{2}}$, where $\lVert \cdot \rVert$ denotes the Euclidean norm in $\mathbb{R}^d$. We let $B_r(x)=\left\{y \in \mathbb{R}^d:\lVert y-x\rVert\le r\right\}$ denote the closed ball of radius $r$ in $\mathbb{R}^d$ centered at $x$, and for $x=0$, we drop the $0$ so that $B_r=B_r(0)$.  Finally, we denote by $\mathbbm{1}_A$ the indicator function for a set $A \subseteq \mathbb{R}^d$. We also introduce the following family of mollifiers:
\begin{definition}[Radial Mollifiers]\label{familyofmollifiers}
Let $\varphi\colon\mathbb{R}^d \to [0,\infty)$ be a fixed smooth radial non-negative function supported on $B_1^d$, with $\int_{\mathbb{R}^d}\varphi=1$. Let $\Phi \colon \mathbb{R}_{\geq 0} \to \mathbb{R}$ be such that $\widehat{\varphi}(x) = \Phi(\|x\|)$. For any $t \in (0,1)$, define
\begin{equation*}
    \varphi_t(x) \coloneqq \frac{1}{t^d}\varphi\left(\frac{x}{t}\right) \quad \textrm{so} \quad \widehat{\varphi_t}(\omega) = \Phi(\|t\omega\|)\textrm{.}
\end{equation*}
Then we say $\left(\varphi_t\right)_{t \in (0,1)}$ is a \emph{family of mollifers}.
% For such a $\varphi$, there exists a $\Phi \colon \mathbb{R}_{\geq 0} \to \mathbb{R}$ such that $\widehat{\varphi}(x) = \Phi(\|x\|)$.
\end{definition}
We fix an arbitrary $\varphi$ that we use in the rest of the paper. For convenience, we state the following lemma. We omit the proof, as it is fairly simple.
\begin{lemma}
    Let $\varphi_t$ be the mollifier as in Definition \ref{familyofmollifiers}, and pick $r>1$. Then for any $x \in \mathbb{R}^d$,
    \begin{align*}
        \mathbbm{1}_{B_{r-t}}(x)\le \left(\mathbbm{1}_{B_r}*\varphi_t\right)(x)\le \mathbbm{1}_{B_{r+t}}(x)\textrm{.}
    \end{align*} 
\label{iwouldnever}
\end{lemma}
As a preliminary result we have the following proposition.
\begin{proposition}
Let $X \subset \mathbb{R}^d$ be a discrete set, and let $\left(b_x\right)_{x \in X}$ be complex coefficients indexed by $X$. Let $\delta>0$ be the separation of $X$ around $0$, i.e., $X\cap B_\delta \subset \{0\}$. Suppose there are constants $B_X > 0$ and $P > 0 $ such that the coefficients $b_x$ satisfy the growth rate \begin{align*}
    \sum_{\substack{x \in X\cap B_r}}\left|b_x\right|\le B_Xr^P \quad \textrm{for all}\;r\geq\delta\textrm{.}
\end{align*}
Then for any $a>0$ with $a\neq P$, there exists a constant $W_{\delta,a,P} > 0$ such that for any $t > 0$, we have
\begin{align*}
    \sum_{x \in X\setminus\{0\}}\frac{\left|b_x\widehat{\varphi}(tx)\right|}{\lVert x\rVert^a}\le B_XW_{\delta,a,P} \left(1+t^{a-P}\right)\textrm{.} %1+t^{-(P-a)}
\end{align*}
    \label{temperednesstogrowth}
\end{proposition} \begin{proof}
    We use dyadic decomposition. Because $\widehat{\varphi}$ is Schwartz, we have $\left|\widehat{\varphi}(x)\right|\le C_P\left\langle x\right\rangle^{-P-1}$ for some constant $C_P > 0$. Thus,
    \begin{align}
        \sum_{x \in X\setminus\{0\}}\frac{\left|b_x\widehat{\varphi}(tx)\right|}{\lVert x\rVert ^a}& \le C_P\sum_{x \in X\setminus\{0\}}\frac{\left|b_x\right|\left\langle tx\right\rangle^{-P-1}}{\lVert x\rVert ^a}=C_P\sum_{k=0}^\infty \sum_{\substack{x \in X \\ 2^{k}\delta\le \lVert x\rVert <2^{k+1}\delta}}\frac{\left|b_x\right|\left\langle tx\right\rangle^{-P-1}}{\lVert x\rVert ^a}\textrm{.}
        \label{andrew}
    \end{align} 
    Note $\left\langle tx\right\rangle^{-P-1} \leq \left\langle2^{k}\delta t\right\rangle^{-P-1}$ when $2^{k}\delta\le \lVert x\rVert$, so 
    \begin{align*}
    \sum_{\substack{x \in X \\ 2^k\delta\le \lVert x\rVert <2^{k+1}\delta}} \frac{\left|b_x\right|\left\langle tx\right\rangle^{-P-1}}{\lVert x\rVert ^a}\le \frac{\left\langle 2^k\delta t \right\rangle^{-P-1}}{2^{ak}\delta^a}\sum_{\substack{x \in X \\ 2^k\delta\le \lVert x\rVert <2^{k+1}\delta}}\left|b_x\right|\le \frac{\left\langle 2^k\delta t\right\rangle^{-P-1}}{\delta^a 2^{ak}}B_X\left(2^{k+1}\delta\right)^P\textrm{,}
\end{align*}
where in the last inequality we used the growth condition on the coefficients $b_x$. Summing both sides over $k \in \mathbb{N}\cup \{0\}$ and substituting this into \eqref{andrew} implies \begin{align*}
    \sum_{x \in X\setminus\{0\}}\frac{\left|b_x\widehat{\varphi}\left(tx\right)\right|}{\lVert x\rVert ^a}\le B_XC_P\delta^{P-a}2^{P}\sum_{k=0}^\infty 2^{k\left(P-a\right)}\left\langle 2^k\delta t\right\rangle^{-P-1}\textrm{.}
\end{align*}
We partition the last sum over $k$ into $2^k\delta t\le 1$ and $2^k\delta t>1$. When $\|x\| \leq 1$, we have $\left\langle x\right\rangle^{-P-1} \leq 1$; when $\|x\| > 1$, we have $\left\langle x\right\rangle^{-P-1} \leq \lVert x\rVert^{-P-1}$. Therefore, we obtain
\begin{align}
         \sum_{x \in X\setminus\{0\}}\frac{\left|b_x\widehat{\varphi}(tx)\right|}{\lVert x\rVert^a}&\le B_XC_P\delta^{P-a}2^{P}\left(\sum_{\substack{k\in \mathbb{N}\cup \{0\} \\ 2^k\delta t \leq 1}}2^{k(P-a)}+(t\delta)^{-P-1}\sum_{\substack{k\in \mathbb{N}\cup \{0\} \\ 2^k\delta t > 1}} 2^{-k(1+a)}\right)\textrm{.}
         \label{apt}
    \end{align}
    Recall $P\neq a$ by assumption. 
    % Let $k_{t,\delta}$ be the smallest non-negative integer such that $2^{k}t\delta>1$, so \begin{align*}
    %      (t\delta)^{-P-1}\sum_{\substack{k\in \mathbb{N}\cup \{0\} \\ 2^kt\delta>1}} 2^{-k(1+a)}= (t\delta)^{-P-1}2^{-k_{t,\delta}\left(1+a\right)}\sum_{h=0}^\infty 2^{-h(1+a)}\le (t\delta)^{-P-1}2^{-k_{t,\delta}(1+a)}W_a\textrm{.}
    % \end{align*} By definition, we
    To bound the second sum, we use the standard bound of the geometric series\footnote{When $2^{k_0}\delta t>1$, then $\sum_{k=k_0}^\infty 2^{-k\alpha}=2^{-k_0\alpha}\left(1-2^{-\alpha}\right)^{-1}\le (\delta t)^\alpha\left(1-2^{-\alpha}\right)^{-1}$.} \begin{align}
        (t\delta)^{-P-1}\sum_{\substack{k\in \mathbb{N}\cup \{0\} \\ 2^kt\delta>1}} 2^{-k(1+a)}\le  (t\delta)^{-P-1}\left(\delta t\right)^{1+a}\left(1-2^{-1-a}\right)^{-1}=\left(1-2^{-1-a}\right)^{-1}(\delta t)^{a-P}\textrm{.}
        \label{sothat}
    \end{align}
    For the first sum in parenthesis in \eqref{apt}, from the formula for geometric series, we have \begin{align*}
         \sum_{\substack{k \in \mathbb{N}\cup \{0\} \\ 2^k\delta t\le 1}}2^{k(P-a)}=\frac{1-2^{k_0(P-a)}}{1-2^{P-a}}\le \begin{cases}
             \left(1-2^{P-a}\right)^{-1} & \text{if}\; P<a\textrm{,} \\
              \left(2^{P-a}-1\right)^{-1}\left(\delta t\right)^{a-P} & \text{if}\; P>a\textrm{,}
         \end{cases}
    \end{align*} where $k_0=\left\lfloor -\log_2(\delta t)\right\rfloor$.
    % Note that if $P<a$, then $\sum_{k=0}^\infty 2^{k(P-a)}\le Q_{a,P}$, and if $P>a$, then \begin{align*}
    %     \sum_{\substack{k\in \mathbb{N}\cup \{0\} \\ 2^k\delta\le t^{-1}}}2^{k(P-a)}\le \left(\delta t\right)^{a-P}\sum_{k=0}^\infty 2^{a-P}\le Q_{a,P,\delta }t^{a-P}\textrm{,}
    % \end{align*} 
    Combining this bounds with that in \eqref{sothat} and plugging it into \eqref{apt}, we get
    \begin{align*}
         \sum_{x \in X\setminus\{0\}}\frac{\left|b_x\widehat{\varphi}(tx)\right|}{\lVert x\rVert ^a}&\le 
         B_XC_P\delta^{P-a}2^{P}\begin{cases}
            \left(1-2^{-1-a}\right)^{-1}t^{a-P}+\left(1-2^{P-a}\right)^{-1} & \text{if}\; P<a\textrm{,} \\
            \left(1-2^{-1-a}\right)^{-1}t^{a-P}+\left(2^{P-a}-1\right)^{-1}\left(\delta t\right)^{a-P} & \text{if}\; P>a\textrm{,}
         \end{cases}
    \end{align*}
    and the claim follows.
\end{proof} If $\Lambda \subset \mathbb{R}^d$ is a discrete Fourier quasicrystal, it is shown in \cite{alon2024higherdimensionalfourierquasicrystals} that $\#\left(\Lambda \cap B_r(x)\right)=c_0\Vol_d\mathopen{}\left(B_r\right)\mathclose{}+O\left(r^{d-1}\right)$ uniformly in $x$, and so if $\Lambda$ is furthermore uniformly discrete,
the proof above implies the following corollary:
\begin{corollary}\label{sumoverunifdiscrete}
Let $\Lambda \subset \mathbb{R}^d$ be any uniformly discrete quasicrystal and $a\neq d$. Then
    \begin{align*}
        \sup_{x\in\Lambda}\sum_{\lambda \in \Lambda\setminus\{x\}}\frac{\left|\widehat{\varphi}(t\left(\lambda-x\right))\right|}{\lVert \lambda-x\rVert ^a}\lesssim 1+t^{a-d}\textrm{.}
    \end{align*}
\end{corollary}

\section{Pointwise Error Bounds}\label{ptwise}
We now prove Theorem \ref{thm:pointwiseerrorboundgen}, restated below.
% \pointwiseerrorboundgen*
\begin{theorem}
    Let $\Lambda \subset \mathbb{R}^d$ be any uniformly discrete Fourier quasicrystal, and suppose its coefficients have growth rate $P$. Then
    \begin{equation*}
        \Err(r,\Lambda) \coloneqq \#\left(\Lambda \cap B_r\right) - c_0\Vol_d\left(B_r\right) = O\left(r^{\frac{(2P-d)(d-1)}{2P-(d-1)}}\right)\textrm{.}
    \end{equation*}
\end{theorem}
\begin{proof}
% We start with showing we can take $P= \frac{N+d}{2}$. By the Cauchy-Schwarz inequality and Theorem \ref{thm:boundsoncsquaredgen}, \begin{align*}
%     \left(\sum_{s \in S\cap B_r^d}\left|c_s\right|\right)^2\le \#\left(S\cap B_r^d\right)\sum_{s \in S\cap B_r^d}\left|c_s\right|^2=O\left(r^{N+d}\right)\textrm{,}
% \end{align*} and so this implies 
% \begin{align}\label{eqn:relating-P-to-N}
%     \sum_{s \in S\cap B_r^d}\left|c_s\right|=O\left(r^{\frac{N+d}{2}}\right)\textrm{.}
% \end{align}
To bound $\Err(r,\Lambda)$ using the summation formula, we use $\varphi_t$ to smooth the error term on a scale of $t$: To that end, we define the smoothed out error term of scale $t$, denoted $\Err_t$, by
\begin{align}
    \Err_t\left(r,\Lambda\right)&\coloneqq\sum_{\lambda \in \Lambda}\left(\mathbbm{1}_{B_r}*\varphi_t\right)(\lambda)-c_0\mathopen{}\Vol_d\mathopen{}\left(B_r\right)\mathclose{}\mathclose{}\textrm{.}
 \label{smoothdeouterrorterm}
\end{align}
We first bound the difference between $\Err_t(r,\Lambda)$ and $\Err_t(r,\Lambda)$, which is on the order of $r^{d-1}t$. Second, we use the summation formula to bound $\Err_t(r,\Lambda)$. $\mathbbm{1}_{B_r}*\varphi_t$ is Schwartz, so apply the summation formula to the sum in \eqref{smoothdeouterrorterm}, we have
\begin{equation}
   \Err_t\left(r,\Lambda\right)=r^d\sum_{s \in S\setminus\{0\}}c_s\widehat{\mathbbm{1}_{B_1}}\left(rs\right)\widehat{\varphi}(ts) =\sum_{s \in S\setminus\{0\}}c_s\widehat{\varphi}(ts)\frac{r^{\frac{d}{2}}J_{\frac{d}{2}}\left(2\pi r\lVert s\rVert\right)}{\lVert s\rVert^{\frac{d}{2}}}\textrm{,}
   \label{aminor2}
\end{equation}
where in the second equality we use the well-known formula for the Bessel function of the first kind $J_k$ of order $k$
\begin{align*}
    \widehat{\mathbbm{1}_{B_1}}(\omega)=\frac{J_{\frac{d}{2}}\left(2\pi\lVert \omega \rVert\right)}{\lVert \omega\rVert^{\frac{d}{2}}}\textrm{.}
\end{align*} Recall $\varphi$ is radial, so we can also write this as \begin{align}
    \Err_t\left(r,\Lambda\right)=\sum_{\gamma \in S_{\textrm{rad}}}\ell(\gamma)\Phi\left(t\gamma\right)\frac{r^{\frac{d}{2}}J_{\frac{d}{2}}\left(2\pi r\gamma\right)}{\gamma^{\frac{d}{2}}}\textrm{.}
    \label{aminor}
\end{align}
Recall $\Phi$ is the radial component of $\widehat{\varphi}$: $\widehat{\varphi}(x)=\Phi\left(\lVert x\rVert\right)$. We then apply the asymptotic relation $J_{\frac{d}{2}}(x)=O\left(x^{-\frac1{2}}\right)$ to \eqref{aminor2} to see that 
\begin{align*}
    \left|\Err_t\left(r,\Lambda\right)\right|\lesssim r^{\frac{d-1}{2}}\sum_{s \in S\setminus\{0\}}\frac{\left|c_s\widehat{\varphi}(ts)\right|}{\lVert s\rVert ^{\frac{d+1}{2}}}\textrm{.}
\end{align*} Let $P$ be a coefficient growth rate for $\Lambda$, so $\sum_{s \in S\cap B_r}\left|c_s\right|=O\left(r^P\right)$. Applying Proposition \ref{temperednesstogrowth} to the above expression, taking $P>\frac{d+1}{2}$, we get \begin{align}
\left|\Err_t\left(r,\Lambda\right)\right|\lesssim r^{\frac{d-1}{2}}t^{-\left(P-\frac{d+1}{2}\right)}\textrm{.}
    \label{centuries}
\end{align} 
% \begin{align}
%     \left|\Err_t\left(r,\Lambda\right)\right|\lesssim r^{\frac{d-1}{2}}\sum_{s \in S\setminus\{0\}}\frac{\left|c_s\widehat{\varphi}(ts)\right|}{\lVert s\rVert ^{\frac{d+1}{2}}}\textrm{.}\label{duhe}
% \end{align}
% \begin{align}
%     \left|\Err_t\left(r,\Lambda\right)\right|\lesssim r^{\frac{d-1}{2}}\sum_{\gamma \in S_{\textrm{rad}}}\frac{\left|\ell(\gamma)\Phi\left(t\gamma\right)\right|}{\gamma^{\frac{d+1}{2}}}\textrm{.}\label{duhe}
% \end{align}
% % We then optimize $t$ so that the resulting bound on $\Err$ is minimal.
Let $t \in \left(0,\frac1{2}\right)$, $r>1$. We know from Lemma \ref{iwouldnever} that $\mathbbm{1}_{B_{r-t}}\le \mathbbm{1}_{B_r}*\varphi_t\le \mathbbm{1}_{B_{r+t}}$, or alternatively $\mathbbm{1}_{B_{r-t}}*\varphi_t\le \mathbbm{1}_{B_r}\le \mathbbm{1}_{B_{r+t}}*\varphi_t$. Summing over $\lambda \in \Lambda$, we get
\begin{equation}
   \sum_{\lambda \in \Lambda}\left(\mathbbm{1}_{B_{r-t}}*\varphi_t\right)(\lambda)\le \#\left(\Lambda \cap B_r\right)\le \sum_{\lambda \in \Lambda}\left(\mathbbm{1}_{B_{r+t}}*\varphi_t\right)(\lambda)\textrm{.} \label{winter}
\end{equation}
% Define the smoothed out error term of scale $t$ by 
Note $\Vol_d\mathopen{}\left(B_r\right)\mathclose{}=\Vol_d\mathopen{}\left(B_1\right)\mathclose{}r^d$. Subtracting $c_0\Vol_d\mathopen{}\left(B_1\right)\mathclose{}r^d$ from \eqref{winter} implies \begin{align*}
     &\Err_t\left(r-t,\Lambda\right)+c_0\Vol_d\left(B_1\right) \left(r-t\right)^d-c_0\Vol_d\left(B_1\right) r^d \\
     &\le \Err\left(r,\Lambda\right)\le \Err_t\left(r+t,\Lambda\right)+c_0\Vol_d\left(B_1\right)\left(r+t\right)^d-c_0\Vol_d\left(B_1\right)r^d\textrm{.}
 \end{align*} Note that for $t \in \left(0,\frac1{2}\right)$, and $r>1$, we have $r^d-\left(r-t\right)^d\lesssim r^{d-t}t$, $\left(r+t\right)^d-r^d\lesssim r^{d-1}t$. Taking the absolute value of both sides of the above inequality gives
 \begin{equation}
    \left|\Err\left(r,\Lambda\right)\right|\lesssim  r^{d-1}t+\max_{\tau \in \{\pm 1\}}\left|\Err_t\left(r+\tau t,\Lambda\right)\right|\textrm{,}
    \label{bounderrorterm3}
\end{equation} and combining this with the estimate in \eqref{centuries}, we get
% Note $\mathbbm{1}_r*\varphi_t$ is Schwartz, so applying the summation formula \eqref{summationformulaunifdis} gives
% \begin{equation*}
%     \Err_t\left(r,\Lambda\right)=\sum_{\lambda \in \Lambda}\left(\mathbbm{1}_r*\varphi_t\right)(\lambda)-c_0\mathopen{}\Vol_d\mathopen{}\left(B_r^d\right)\mathclose{}\mathclose{}=\sum_{s \in S}c_s\widehat{\mathbbm{1}_r}(s)\widehat{\varphi_t}(s)-c_0\mathopen{}\Vol_d\mathopen{}\left(B_r^d\right)\mathclose{}\mathclose{}=\sum_{s \in S\setminus\{0\}}c_s\widehat{\mathbbm{1}_r}(s)\widehat{\varphi_t}(s)\textrm{,}
% \end{equation*} since \begin{align*}
%     \widehat{1_r}(0)\widehat{\varphi_t}(0)=\int_{B_r^d}1dx\cdot \int_{\mathbb{R}^d}\varphi_t=\Vol_d\mathopen{}\left(B_r^d\right)\mathclose{}\textrm{.}
% \end{align*}
% Note that $\mathbbm{1}_r(s)=\mathbbm{1}_1\left(\frac{s}{r}\right)$, so $\widehat{\mathbbm{1}_r}(\omega)=r^d\widehat{\mathbbm{1}_1}\left(r\omega\right)$. 
% Suppose $\Lambda$ has coefficient growth rate  $P>\frac{d+1}{2}$, so $\sum_{s \in S\cap B_T^d}\left|c_s\right|=O\left(T^P\right)$ as $T \to \infty$. Applying Proposition \ref{temperednesstogrowth} to the right-hand side of \eqref{duhe}  implies $\left|\Err_t\left(r,\Lambda\right)\right|\lesssim r^{\frac{d-1}{2}}t^{-\left(P-\frac{d+1}{2}\right)}$.
% Substituting this into \eqref{bounderrorterm3} we get 

\begin{equation*}
    \left|\Err\left(r,\Lambda\right)\right|\lesssim r^{d-1}t+\frac{r^{\frac{d-1}{2}}}{t^{P-\frac{d+1}{2}}}.
\end{equation*} Choosing $t=r^{-\frac{d-1}{2P-d+1}}$ then gives
\begin{equation*}
    \Err\left(r,\Lambda\right)=O\left(r^{\frac{(d-1)(2P-d)}{2P-d+1}}\right)\textrm{,}
\end{equation*}
 as desired.
\end{proof}

Finally, we prove the generalization of Theorem \ref{thm:boundsoncsquared} given in Theorem ~\ref{thm:boundsoncsquaredgen}. 
We restate it for convenience.
\boundsoncsquaredgen*
In the proof, we use 
the following result of Alon, Kummer, Kurasov, and Vinzant \cite{alon2024higherdimensionalfourierquasicrystals}.
\begin{theorem}[Theorem 10.1, \cite{alon2024higherdimensionalfourierquasicrystals}]\label{autcorr}
    Let $\Lambda \subset \mathbb{R}^d$ be any uniformly discrete Fourier quasicrystal with summation formula as in \eqref{summationformulaunifdis}. Let $f \in \mathcal{S}\left(\mathbb{R}^d\right)$. Then
    \begin{align*}
        \sum_{s \in S}\left|c_s\right|^2f(s)=\lim_{R \to \infty}\frac1{\Vol_d\mathopen{}\left(B_R\right)\mathclose{}}\sum_{x \in \Lambda\cap B_R}\sum_{y \in \Lambda}\widehat{f}(y-x)\textrm{.}
    \end{align*}
\end{theorem}
\begin{proof}[Proof of Theorem \ref{thm:boundsoncsquaredgen}]
We take $t \in \left(0,\frac1{2}\right)$, $r>1$. Consider an arbitrary function $f \in \mathcal{S}\left(\mathbb{R}^d\right)$. We write the equality in Theorem \ref{autcorr} as
\begin{align*}
    \sum_{s \in S}\left|c_s\right|^2f(s)&=\lim_{T \to \infty}\frac1{\Vol_d\mathopen{}\left(B_T\right)\mathclose{}}\sum_{x \in \Lambda \cap B_T}\sum_{y \in \Lambda}\widehat{f}\left(x-y\right)\textrm{,}
\end{align*} and introduce a factor of $\#\left(\Lambda \cap B_T\right)$ to get \begin{align}
    \sum_{s \in S}\left|c_s\right|^2f(s)&=\lim_{T \to \infty}\frac{\#\left(\Lambda \cap B_T\right)}{\Vol_d\mathopen{}\left(B_T\right)\mathclose{}}\frac1{\#\left(\Lambda \cap B_T\right)}\sum_{x \in \Lambda \cap B_T}\sum_{y \in \Lambda}\widehat{f}\left(y-x\right)\textrm{.}
    \label{tsongbird}
\end{align}
Recall the density of $\Lambda$ is $c_0$, so \begin{align*}
    \lim_{T \to \infty}\frac{\#\left(\Lambda \cap B_T\right)}{\Vol_d\mathopen{}\left(B_T\right)\mathclose{}}=c_0\textrm{.}
\end{align*} \eqref{tsongbird} then becomes
\begin{align*}
     \sum_{s \in S}\left|c_s\right|^2f(s)&=c_0\lim_{T \to \infty}\frac1{\#\left(\Lambda \cap B_T\right)}\sum_{x \in \Lambda \cap B_T}\sum_{y \in \Lambda}\widehat{f}\left(y-x\right)\textrm{.}
\end{align*}
Pulling out the terms in the summation where $y=x$ yields
\begin{align}
    \sum_{s \in S}\left|c_s\right|^2f(s)&=c_0\widehat{f}(0)+\lim_{T \to \infty}\frac{c_0}{\#\left(\Lambda\cap B_T\right)}\sum_{x \in \Lambda \cap B_T}\sum_{y \in \Lambda\setminus\{x\}}\widehat{f}(y-x)\textrm{,} \ \ \forall f \in \mathcal{S}\left(\mathbb{R}^d\right)\textrm{.}
    \label{struck}
\end{align} 
From Lemma ~\ref{iwouldnever} we have $\mathbbm{1}_{B_{r-t}}\le \mathbbm{1}_{B_r}*\varphi_t\le \mathbbm{1}_{B_{r+t}}$. Thus,
\begin{equation*}
    \sum_{s \in S\cap B_{r-t}}\left|c_s\right|^2\le \sum_{s \in S}\left|c_s\right|^2\left(\mathbbm{1}_{B_r}*\varphi_t\right)(s)\le \sum_{s \in S\cap B_{r+t}}\left|c_s\right|^2\textrm{.}
\end{equation*} Let $f_{r,t}=\mathbbm{1}_{B_r}*\varphi_t$, so the above inequality can be rewritten as
\begin{equation}
    \sum_{s \in S}\left|c_s\right|^2f_{r-t,t}(s)\le \sum_{s \in S\cap B_r}\left|c_s\right|^2\le \sum_{s \in S}\left|c_s\right|^2f_{r+t,t}(s)\textrm{.} \label{maria}
\end{equation}
The Fourier transform of $f_{r,t}$ is given by \begin{align*}
    \widehat{f_{r,t}}(s)=\widehat{\mathbbm{1}_{B_r}}(s)\widehat{\varphi}(ts)=\frac{r^{\frac{d}{2}}J_{\frac{d}{2}}\left(2\pi r\lVert s\rVert \right)\widehat{\varphi}(ts)}{\lVert s\rVert^{\frac{d}{2}} }\textrm{,}
\end{align*}
where $J_{k}$ is the Bessel function of the first kind of order $k$. Note $\widehat{f_{t,r}}(0)=\widehat{\mathbbm{1}_{B_r}}(0)=\Vol_d\mathopen{}\left(B_r\right)\mathclose{}$. Using the asymptotic relation $J_{\frac{d}{2}}(x)=O\left(x^{-\frac1{2}}\right)$, we get \begin{align}
    \left|\widehat{f_{r,t}}(s)\right|\lesssim \frac{r^{\frac{d-1}{2}}\left|\widehat{\varphi}\left(ts\right)\right|}{\lVert s\rVert ^{\frac{d+1}{2}}}. \label{forgames}
\end{align} Substituting $f=f_{r,t}$ into \eqref{struck} implies
\begin{align}
      \sum_{s \in S}\left|c_s\right|^2f_{r,t}(s)&=c_0\mathopen{}\Vol_d\mathopen{}\left(B_r\right)\mathclose{}\mathclose{}+c_0\lim_{T \to \infty}\frac{1}{\#\left(\Lambda\cap B_T\right)}\sum_{x \in \Lambda \cap B_T}\sum_{y \in \Lambda\setminus\{x\}}\widehat{f_{r,t}}(y-x)\textrm{.}
      \label{givemeasign}
\end{align} Using the estimate \eqref{forgames}, for $x \in \Lambda$, \begin{align*}
    \sum_{y \in \Lambda\setminus\{x\}}\left|\widehat{f_{r,t}}(y-x)\right|\lesssim r^{\frac{d-1}{2}}\sum_{y \in \Lambda\setminus\{x\}}\frac{\left|\widehat{\varphi}\left(t(y-x)\right)\right|}{\lVert y-x\rVert^{\frac{d+1}{2}}}.
\end{align*} Recall $\Lambda$ is uniformly discrete, so from Corollary \ref{sumoverunifdiscrete}, we may write \begin{align*}
    \sum_{y \in \Lambda\setminus\{x\}}\frac{\left|\widehat{\varphi}\left(t(y-x)\right)\right|}{\lVert y-x\rVert^{\frac{d+1}{2}}}\lesssim t^{-\frac{d-1}{2}}
\end{align*} independently of $x$. Averaging over $x \in \Lambda \cap B_T$ then gives \begin{align*}
    \frac{1}{\#\left(\Lambda\cap B_T\right)}\sum_{x \in \Lambda \cap B_T}\sum_{y \in \Lambda\setminus\{x\}}\left|\widehat{f_{r,t}}(y-x)\right|\lesssim r^{\frac{d-1}{2}}\frac{1}{\#\left(\Lambda\cap B_T\right)}\sum_{x \in \Lambda \cap B_T} \sum_{y \in \Lambda\setminus\{x\}}\frac{\left|\widehat{\varphi}\left(t(y-x)\right)\right|}{\lVert y-x\rVert^{\frac{d+1}{2}}}\lesssim r^{\frac{d-1}{2}}t^{-\frac{d-1}{2}}\textrm{,}
\end{align*}
  and substituting this into \eqref{givemeasign} implies
\begin{align*}
    \sum_{s \in S}\left|c_s\right|^2f_{r,t}(s)=c_0\mathopen{}\Vol_d\mathopen{}\left(B_r\right)\mathclose{}\mathclose{}+O\left(r^{\frac{d-1}{2}}t^{-\frac{d-1}{2}}\right)\textrm{.}
\end{align*}  Substituting this into \eqref{maria} implies \begin{align*}
    c_0\mathopen{}\Vol_d\mathopen{}\left(B_1\right) \mathclose{}\mathclose{}\left(r-t\right)^d+O\left(r^{\frac{d-1}{2}}t^{-\frac{d-1}{2}}\right)\le \sum_{s \in S\cap B_r}\left|c_s\right|^2\le c_0\mathopen{}\Vol_d\mathopen{}\left(B_1\right) \mathclose{}\mathclose{}\left(r+t\right)^d+O\left(r^{\frac{d-1}{2}}t^{-\frac{d-1}{2}}\right).
\end{align*} We then subtract $c_0\mathopen{}\Vol_d\mathopen{}\left(B_1\right)\mathclose{}\mathclose{}r^d$ from both sides and take the absolute value to get \begin{align*}
    \left|\sum_{s \in S\cap B_r}\left|c_s\right|^2-c_0\mathopen{}\Vol_d\mathopen{}\left(B_1\right)\mathclose{}\mathclose{}r^d\right|\lesssim r^{d-1}t+\frac{r^{\frac{d-1}{2}}}{t^{\frac{d-1}{2}}}.
\end{align*} Choosing $t=r^{-\frac{d-1}{d+1}}$ implies \begin{align*}
    \left|\sum_{s \in S\cap B_r}\left|c_s\right|^2-c_0\mathopen{}\Vol_d\mathopen{}\left(B_1\right)\mathclose{}\mathclose{}r^d\right|\lesssim r^{\frac{d(d-1)}{d+1}},
\end{align*} as desired. This completes the proof. 
\end{proof}

% \subsection{Testing}
% \begin{proof}[Proof of second part]
%     Substitute $f=f_{r,t}$, and write it as \begin{align*}
%         \sum_{s \in S}\left|c_s\right|^2f_{r,t}(s)=c_0\sum_{x \in \Lambda\cap B_T}\sum_{y \in \Lambda\cap B_1(x)}\frac1{\#\left(\Lambda \cap B_T\right)}\widehat{f_{r,t}}(y-x)+c_0\sum_{x \in \Lambda\cap B_T}\sum_{\substack{y \in \Lambda \\ |y-x|\ge 1}}\frac1{\#\left(\Lambda \cap B_T\right)}\widehat{f_{r,t}}(y-x)\textrm{.}
%     \end{align*} Note $\widehat{f_{r,t}(s)}=\widehat{\mathbbm{1}_{B_r}}(s)\widehat{\phi}(ts)=r^d\widehat{\mathbbm{1}_{B_1}}(rs)\widehat{\varphi}(ts)$, so $\left|\widehat{f_{r,t}}(s)\right|\lesssim \frac{r^{\frac{d-1}{2}}\left|\widehat{\varphi}(ts)\right|}{\lVert s\rVert}$. So from Corollary, we have \begin{align*}
%         \left|\sum_{x \in \Lambda\cap B_T}\sum_{\substack{y \in \Lambda \\ |y-x|\ge 1}}\frac1{\#\left(\Lambda \cap B_T\right)}\widehat{f_{r,t}}(y-x)\right|\lesssim \frac{r^{\frac{d-1}{2}}}{t^{\frac{d-1}{2}}}\textrm{.}
%     \end{align*}
% \end{proof}

\section{Expansion of the Smooth Error Term}\label{expofsmooth}
Next, we perform a more detailed analysis on $\Err_t$ in $\mathbb{R}^2$. 
%     a_k(\nu)=\frac{\left(4\nu^2-1\right)\left(4\nu^2-3\right)\cdots \left(4\nu^2-\left(2k-1\right)^2\right)}{k!\cdot 8^k},
% \end{equation*} for integral $k\ge 1$, with $a_0(\nu)=1$. Then
% \begin{align}
% \begin{split}
%     J_{\nu}(x)&=\sqrt{\frac2{\pi x}}\left(\cos\left(\omega_\nu(x)\right)\sum_{k=0}^{M_1}\frac{\left(-1\right)^ka_{2k}(\nu)}{x^{2k}}-\sin\left(\omega_\nu(x)\right)\sum_{k=0}^{M_2}\frac{(-1)^ka_{2k+1}(\nu)}{x^{2k+1}}\right) \\
%     &\quad +O\left(\frac1{x^{\min\left(2M_1,2M_2+1\right)+\frac3{2}}}\right)
% \end{split}
% \label{eqn:hankelexpansion}
% \end{align}
% for each $M_1,M_2 \in \mathbb{N}$ as $x \to \infty$. We use the identities \begin{align*}
%     \sin x=\frac{\left(e^{ix}-e^{-ix}\right)}{2i}, \ \cos x=\frac{\left(e^{ix}+e^{-ix}\right)}{2}
% \end{align*}
We use Hankel's asymptotic expansion of the Bessel functions (see Chapter 9, page 364 in ~\cite{abramowitz1968handbook}) 
to write
\begin{align}\label{certified}
    J_\nu(x)=\sum_{\tau \in \{\pm 1\}}\sum_{k=0}^M\frac{a_{k,\tau}(\nu)e^{i \tau x }}{x^{k+\frac1{2}}}+O\left(\frac1{x^{M+\frac3{2}}}\right)
\end{align}
for fixed $\nu$, as $x \to \infty$. The constants $a_{k,\tau}(\nu)$ are explicitly given in ~\cite{abramowitz1968handbook}. In particular, when $\nu = 1$, we have the expansion
\begin{align}
    J_1(x)=\sqrt{\frac2{\pi x}}\cos\left(x-\frac{3\pi}{4}\right)+O\left(x^{-\frac3{2}}\right)
    \label{J1asym}\textrm{.}
\end{align}
We use this formula for the Bessel function to prove the following statement about the smoothed out error function $\Err_t$.
\begin{restatable}{prop}{pseudoperiod}
    Let $\Lambda \subset \mathbb{R}^2$ be any Fourier quasicrystal. Let $\Err_t\left(r,\Lambda\right)$ denote the smoothed out error term of scale $t$ defined in \eqref{smoothdeouterrorterm}. Then  
    we have
    % \begin{align}
    %     \Err_t\left(r,\Lambda\right)=r^{\frac1{2}}\sum_{\tau \in \{\pm 1\}}\sum_{s \in S\setminus\{0\}}\frac{a_\tau c_s\widehat{\varphi}(ts)e^{2\pi i \tau r \lVert s\rVert }}{\lVert s\rVert^{\frac3{2}}}+O(1)\textrm{,}
    %     \label{nicerformerrorcomplexexponentials}
    % \end{align}
    \begin{align}
        \Err_t\left(r,\Lambda\right)=r^{\frac1{2}}\sum_{\tau \in \{\pm 1\}}\sum_{\gamma \in S_{\textrm{rad}}}\frac{a_\tau \ell(\gamma)\Phi(t\gamma )e^{2\pi i \tau r \gamma }}{\gamma^{\frac3{2}}}+O(1)\textrm{,}
        \label{nicerformerrorcomplexexponentials}
    \end{align}
    where $\ell(\gamma)$ and $S_{\textrm{rad}}$ are as in \eqref{agammadef} and \eqref{Srad}, and  coefficients $a_\tau = a_{0,\tau}(1)=\frac{1}{2\pi} \cdot {e^{-\frac{3\pi \tau i }{4}}}$.
    % \begin{align*}
    %     a_\tau\coloneqq a_{0,\tau}(1)=\frac{1}{2\pi} \cdot {\exp\left(-\frac{3\pi \tau i }{4}\right)}\textrm{,}
    % \end{align*}
    The constant in the big $O$ is independent of $t$.
    \label{prop:pseudoperiod}
\end{restatable}
% This expression is a generalization of Hardy's formula \cite{hardy1915expression} for the remainder term. 
% Following Hardy, let $r_2$ denote the sum of squares function
% \begin{align*}
%     r_2(n)=\#\left\{(m,k) \in \mathbb{Z}^2\;\vert\;m^2+k^2=n\right\}\textrm{.}
% \end{align*}
% Hardy showed
% \begin{align*}
%     \sum_{n\le x}r_2(n)=\pi x-1+\sqrt{x}\sum_{n=1}^\infty \frac{r_2(n)}{\sqrt{n}}J_1\left(2\pi\sqrt{n x}\right)\textrm{,}
% \end{align*} for $x \notin \mathbb{N}$. Adapted to the circle problem, this is equivalent to \begin{align*}
%     \Err\left(r,\mathbb{Z}^2\right)=r\sum_{n=1}^\infty \frac{r_2(n)}{\sqrt{n}}J_1\left(2\pi r \sqrt{n}\right)-1
% \end{align*} when $r^2 \notin \mathbb{N}$, i.e., no point of $\mathbb{Z}^2$ lies on the circle of radius $r$. One can then substitute in the asymptotic expression
% \begin{align}
%     J_1(x)=\sqrt{\frac2{\pi x}}\cos\left(x-\frac{3\pi}{4}\right)+O\left(x^{-\frac3{2}}\right)
%     \label{J1asym}
% \end{align} to write the error as a sum 
% \begin{align*}
%     \Err\left(r,\mathbb{Z}^2\right)=\frac{r^{\frac1{2}}}{\pi}\sum_{n=1}^\infty \frac{r_2(n)}{n^{\frac3{4}}}\cos\left(2\pi  r\sqrt{n}-\frac{3\pi}{4}\right)+O\left(r^{-\frac1{2}}\right)\textrm{.}
% \end{align*}
% The authors in \cite{bleher1993distribution} use a generalization of this relation when analyzing the almost-periodicity of the normalized error $\Ern\left(r,\mathbb{Z}^2+\alpha\right)$, and we adapt this to any uniformly discrete Fourier quasicrytal $\Lambda \subset \mathbb{R}^2$.\newline

The proof of Proposition \ref{prop:pseudoperiod} relies on two supporting lemmas. We first assume the truth of these lemmas and prove Proposition \ref{prop:pseudoperiod}. Then we prove the lemmas. The first lemma is as follows. \begin{lemma}
    Let $t \in (0,1)$, $\tau \in \{\pm 1\}$ and $k\ge 1$.  Then \begin{align*}
        \sum_{\gamma \in S_{\textrm{rad}}}\frac{\ell(\gamma)\Phi(t\gamma)e^{2\pi i \tau r \gamma}}{\gamma^{\frac3{2}+k}}=O\left(r^{\frac1{2}}\right)
    \end{align*} independently of $t$. Here  $\tau$ and $k$ are considered to be fixed, with $t \in \left(0,\frac1{2}\right)$ and $r>1$ variables.
    \label{intertoshow}
\end{lemma}
This lemma relies on the following lemma which provides some useful estimates.
\begin{lemma}
Let $\delta>0$ be such that $S\cap B_{\delta}=\{0\}$. Let $\psi\colon\mathbb{R} \to [0,1]$ be a fixed smooth function such that $\psi(x)\equiv 0$ if $x\le \frac{\delta}{2}$ and $\psi(x)\equiv 1$ if $x\ge \delta$. Fix $k\ge 1$ and define $G_{r}^t\colon\mathbb{R}^2 \to \mathbb{C}$ for $t \in \left(0,\frac1{2}\right)$, $r>1$ by
    \begin{align}
        G_{r}^t(x)=\frac{\Phi\left(t\lVert x\rVert \right)\mathclose{}\psi\left(\lVert x\rVert\right)e^{2\pi i \tau r \lVert x\rVert}}{\lVert x\rVert^{\frac3{2}+k}}\textrm{.} \label{capGdef}
\end{align}
   Then \begin{equation*}
         \left|\widehat{G_{r}^t}(\omega)\right|\lesssim_{p,k} \frac{\left\langle r-\lVert \omega\rVert\right\rangle^{-p}}{\lVert \omega\rVert^{\frac1{2}}}+\frac1{\lVert \omega\rVert^{M+\frac3{2}}}\textrm{,}
    \end{equation*} for all $\lVert \omega \rVert \ge 1$, for some $M>1$.
    \label{mainequationtoshow}
\end{lemma} Once we have these two lemmas, we can readily prove Proposition \ref{prop:pseudoperiod}:
\begin{proof}[Proof of Proposition \ref{prop:pseudoperiod}]
     We start with the formula for $\Err_t$ (for $d=2$) given in \eqref{aminor},  \begin{equation*}
    \Err_t\left(r,\Lambda\right)=\sum_{\gamma \in S_{\textrm{rad}}}\frac{\ell(\gamma)\Phi(t\gamma)rJ_1\left(2\pi r\gamma\right)}{\gamma}.
\end{equation*} We use the asymptotic formula for $J_1$ given in \eqref{certified} to get \begin{equation}
    J_1\left(2\pi r\gamma\right)=\sum_{\tau \in \{\pm 1\}}\sum_{k=0}^M\frac{a_{k,\tau}(1)e^{2\pi i\tau r\gamma}}{\left(2\pi r\gamma\right)^{k+\frac1{2}}}+O\left(\frac1{r^{M+\frac3{2}}\gamma^{M+\frac3{2}}}\right)\textrm{.}
    \label{suitandtie}
\end{equation} 
% For example, comparing this with the expansion in \eqref{J1asym},
% we see that \begin{align}
%     a_{0,\tau}=a_\tau=\frac{\exp\left(-\frac{3\pi \tau i}{4}\right)}{\sqrt{2\pi}}\textrm{.}
%     \label{aotau}
% \end{align}
Multiplying \eqref{suitandtie} by $\frac{r}{\gamma}$, we get
\begin{equation}
    \frac{rJ_1\left(2\pi r\gamma\right)}{\gamma}=r^{\frac1{2}}\sum_{\tau \in \{\pm 1\}}\sum_{k=0}^M\frac{a_{k,\tau}(1)e^{2\pi i\tau r\gamma}}{\sqrt{2\pi}\gamma^{\frac3{2}}\left(2\pi r\gamma\right)^k}+O\left(\frac1{r^{M+\frac1{2}}\gamma^{M+\frac5{2}}}\right)\textrm{.}\label{meetthegrahams}
\end{equation}
By Proposition \ref{temperednesstogrowth}, if $M$ is sufficiently large,
\begin{equation*}
    \sum_{\ell \in S_{\textrm{rad}}}\frac{\left|\ell(\gamma)\right|}{\gamma^{M+\frac5{2}}}\le \sum_{s \in S\setminus\{0\}}\frac{\left|c_s\right|}{\lVert s\rVert^{M+\frac5{2}}}<\infty\textrm{.}
\end{equation*}
We choose a sufficiently large $M$, multiply both sides of \eqref{meetthegrahams} by $\ell(\gamma)\Phi(t\gamma)$, and sum over $S_{\textrm{rad}}$ to obtain
\begin{equation*}
    \Err_t\left(r,\Lambda\right)=r^{\frac1{2}}\sum_{\tau \in \{\pm 1\}}\sum_{k=0}^Ma_{k,\tau}(1)\sum_{\gamma \in S_{\textrm{rad}}}\frac{\ell(\gamma)\Phi(t\gamma)e^{2\pi i\tau r\gamma}}{\sqrt{2\pi}\gamma^{\frac3{2}}\left(2\pi r\gamma\right)^{k}}+O\left(r^{-M-\frac1{2}}\right)\textrm{.}
\end{equation*}
Equivalently, 
\begin{align}
    \Err_t\left(r,\Lambda\right)=\sum_{\tau \in \{\pm 1\}}\sum_{k=0}^M\frac{a_{k,\tau}(1)}{\left(2\pi \right)^{k+\frac1{2}}}\left(r^{\frac1{2}-k}\sum_{\gamma \in S_{\textrm{rad}}}\frac{\ell(\gamma)\Phi(t\gamma)e^{2\pi i \tau r \gamma }}{\gamma^{\frac3{2}+k}}\right)+O\left(r^{-M-\frac1{2}}\right)\textrm{.} \label{yellow}
\end{align} We now use Lemma \ref{intertoshow} to bound the inner sums. Note that if $k\ge 1$ then $r^{\frac1{2}-k}\le r^{-\frac1{2}}$. Invoking Lemma \ref{intertoshow} on the sum below implies for $k\ge 1$, \begin{align*}
    r^{\frac1{2}-k}\sum_{\gamma \in S_{\textrm{rad}}}\frac{\ell(\gamma)\Phi(t\gamma)^{2\pi i \tau r \gamma}}{\gamma^{\frac3{2}+k}}=O\left(1\right).
\end{align*} Separating the sum in \eqref{yellow} into $k=0$ and $k\ge 1$, we get \begin{align*}
    \Err_t\left(r,\Lambda\right)=r^{\frac1{2}}\sum_{\tau \in \{\pm 1\}}\frac{a_{0,\tau}(1)}{\sqrt{2\pi}}\sum_{\gamma \in S_{\textrm{rad}}}\frac{\ell(\gamma)\Phi(t\gamma)e^{2\pi i \tau r \gamma }}{\gamma^{\frac3{2}}}+O(1)\textrm{,}
\end{align*}
proving Proposition \ref{prop:pseudoperiod}
% Comparing the expansion in \eqref{certified} with the known expansion \begin{align*}
%     J_1(x)=\sqrt{\frac2{\pi}}\cos\left(x-\frac{3\pi}{4}\right)+O\left(x^{-\frac3{2}}\right)\textrm{,}
% \end{align*} we see that $a_{0,\tau}=\frac{\exp\left(-\frac{3\pi \tau i}{4}\right)}{2\pi}$
% \begin{align*}
%     \Err_t(r)=r^{\frac1{2}}\sum_{\tau \in \{\pm 1\}}\frac{\exp\left(-\frac{3\pi \tau i}{4}\right)}{2\pi}\sum_{\gamma \in S_{\textrm{rad}}}\frac{\ell(\gamma)\Phi(t\gamma)e^{2\pi i \tau r \gamma }}{\gamma^{\frac3{2}}}+O(1)\textrm{,}
% \end{align*}
% proving
\end{proof}
We now prove Lemma \ref{mainequationtoshow}.
\begin{proof}[Proof of Lemma \ref{mainequationtoshow}]
Let us recall the statement of the lemma. We defined $G_r^t\colon\mathbb{R}^2 \to \mathbb{C}$ in \eqref{capGdef} by (without loss of generality we took $\tau=1$) \begin{align*}
   G_r^t(x)= \frac{\Phi\left(t\lVert x\rVert\right)\psi\left(\lVert x\rVert\right)e^{2\pi i r \lVert x\rVert}}{\lVert x\rVert^{\frac3{2}+k}},
\end{align*}
where $k\ge 1$, and $\psi:\mathbb{R} \to [0,1]$ is a smooth function satisfying $\psi(x)=0$ if $x\le \frac{\delta}{2}$ and $\psi(x)=1$ if $x\ge \delta$ for some $\delta>0$. $t \in \left(0,\frac1{2}\right)$ and $r>1$ are variables, and we want to show under these hypotheses that \begin{equation*}
         \left|\widehat{G_{r}^t}(\omega)\right|\lesssim_{p,k} \frac{\left\langle r-\lVert \omega\rVert\right\rangle^{-p}}{\lVert \omega\rVert^{\frac1{2}}}+\frac1{\lVert \omega\rVert^{M+\frac3{2}}}\textrm{,}
    \end{equation*} where $M>1$ is some large integer.  Note $G_{r}^t$ is a radial function, so we can write its Fourier transform as a Hankel transform of order $0$,  \begin{align*}
    \widehat{G_{r}^t}(\omega)=2\pi \int_0^\infty \frac{\Phi(tq)\psi(q)e^{2\pi i r q}}{q^{\frac1{2}+k}}J_0\left(2\pi q \lVert \omega\rVert\right)\textup{d}q\textrm{,}
\end{align*}
where $\Phi$ denotes the radial component of $\widehat{\varphi}$, i.e., $\widehat{\varphi}(x)=\Phi\mathopen{}\left(\lVert x\rVert\right)\mathclose{}$. We now use the expansion for $J_0$ introduced in \eqref{certified}, writing 
\begin{align*}
    J_0(x)=\sum_{\sigma \in \{\pm 1\}}\sum_{m=0}^M\frac{h_{m,\sigma}e^{i \sigma x}}{x^{m+\frac1{2}}}+O\left(\frac1{x^{M+\frac3{2}}}\right)
\end{align*} 
for complex coefficients $h_{m,\sigma}$. Thus
\begin{align*}
    \widehat{G_{r}^t}(\omega)&=2\pi \int_0^\infty \frac{\Phi(tq)\psi(q)e^{2\pi i r q}}{q^{\frac1{2}+k}}\left(\sum_{\sigma \in \{\pm 1\}}\sum_{m=0}^M\frac{h_{m,\sigma}e^{2\pi i \sigma q \lVert \omega \rVert}}{\left(2\pi q \lVert \omega\rVert\right)^{m+\frac1{2}}}+O\left(\frac1{q^{M+\frac3{2}}\lVert \omega\rVert^{M+\frac3{2}}}\right)\right)\textup{d}q \\
    &=\sum_{\sigma \in \{\pm 1\}}\sum_{m=0}^M\frac{b_{m,\sigma}}{\lVert \omega\rVert ^{m+\frac1{2}}}\int_0^\infty \frac{\Phi(tq)\psi(q)e^{2\pi i q\left( r +\sigma \lVert \omega\rVert\right)}}{q^{1+k+m}}\textup{d}q+O\left(\frac1{\lVert \omega\rVert^{M+\frac3{2}}}\right)
\end{align*}
for complex coefficients $b_{m,\sigma}$, and the constant in the big $O$ is independent of $t$ and $r$. Define $F_m^t\colon \mathbb{R} \to \mathbb{C}$ by \begin{align*}
    F_{m}^t(x)=\int_0^\infty \frac{\Phi(tq)\psi(q)e^{2\pi i q x}}{q^{1+k+m}}\textup{d}q\textrm{,}
\end{align*}
so that we can write the above identity as \begin{align}
     \widehat{G_{r}^t}(\omega)&=\sum_{\sigma \in \{\pm 1\}}\sum_{m=0}^M\frac{b_{m,\sigma}}{\lVert \omega\rVert ^{m+\frac1{2}}}F_{m}^t\mathopen{}\left( r +\sigma \lVert \omega\rVert\right)\mathclose{}+O\left(\frac1{\lVert \omega\rVert^{M+\frac3{2}}}\right). \label{proof}
\end{align}
We will now show that
\begin{align}
    \left|F_{m}^t\mathopen{}(x)\mathclose{}\right|\lesssim_{p} \left\langle x\right\rangle^{-p}
    \label{wakeup}
\end{align} 
independently of $t$. First, recall that $\psi(q) = 0$ if $q < \frac{\delta}{2}$ and $\psi(q) = 1$ if $q \geq \delta$, and that $\Phi$ is bounded. Since $k \geq 1$, we have
\begin{align}
    \left|F_{m}^t\mathopen{}(x)\mathclose{}\right|\lesssim \int_{\frac{\delta}{2}}^\infty \frac1{q^2}\textup{d}q\textrm{.} \label{sonata}
\end{align}
From the definition of $\Phi$, $\widehat{\varphi}(x)=\Phi\left(\lVert x\rVert\right)$, $\Phi(q)$ and its derivatives are rapidly decreasing as $q \to \infty$. Using integration by parts, we have
\begin{align}
    \left|F_{m}^t\mathopen{}(x)\mathclose{}\right|\le \frac1{\left(2\pi |x|\right)^p}\int_0^\infty \left|\frac{\partial^p }{\partial q ^p}\left(\frac{\Phi(tq)\psi(q)}{q^{1+k+m}}\right)\right|\textup{d}q, \label{depression}
\end{align}
for each integer $p\ge 0$.
% Note that $\Phi$ and $\psi$ are smooth with bounded derivatives, and $\psi(q)=0$ for all $q\le \frac{\delta}{2}$. So for $t \in (0,1)$, we can see via the product rule there exists some constant $C_p> 0$ such that \begin{equation*}
%       \left|\frac{\partial^p}{\partial q^p}\left(\frac{\Phi(tq)\psi(q)}{q^{1+k+m}}\right)\right| \leq C_p\frac{\mathbbm{1}_{[\frac{\delta}{2},\infty)}(q)}{q^{1+k+m}}\textrm{.}
% \end{equation*} 

% Since $\Psi \equiv 0$ for $q < \frac{\delta}{2}$,
% \begin{equation*}
%     \left|\frac{\partial^p}{\partial q^p}\left(\frac{\Phi(tq)\psi(q)}{q^{1+k+m}}\right)\right| \leq \sum_{p_1+p_2+p_3=p}{\left|\frac{}{}\right|} \leq C\frac{\mathbbm{1}_{[\frac{\delta}{2},\infty)}(q)}{q^{1+k+m}}
% \end{equation*}
Using the product rule,
\begin{align*}
    \frac{\partial^p }{\partial q ^p}\left(\frac{\Phi(tq)\psi(q)}{q^{1+k+m}}\right)&=\sum_{p_1+p_2+p_3=p}\frac{p!}{p_1!p_2!p_3!}\frac{\partial^{p_1}}{\partial q^{p_1}}\left(\frac1{q^{1+k+m}}\right)\frac{\partial^{p_2}\Phi(tq)}{\partial q^{p_2}}\frac{\partial^{p_3}\psi(q)}{\partial q^{p_3}} \\
    &=\sum_{p_1+p_2+p_3=p}\frac{p!}{p_1!p_2!p_3!}Q\left(k,m\right)\frac{t^{p_2}\Phi^{\left(p_2\right)}(tq)\psi^{\left(p_3\right)}(q)}{q^{1+k+p_1+m}}
\end{align*} for constants $Q\left(k,m\right)$. The derivatives of $\Phi$ are bounded, and similarly for $\psi$ since it is eventually constant. Thus we can take the absolute value and use the fact that $t \in (0,1)$ to obtain
\begin{align*}
    \left|\frac{\partial^p }{\partial q ^p}\left(\frac{\Phi(tq)\psi(q)}{q^{1+k+m}}\right)\right|\le B(p)\frac{\mathbbm{1}_{\left[\frac{\delta}{2},\infty\right)}(q)}{q^{1+k+m}}\lesssim B(p)\frac{\mathbbm{1}_{\left[\frac{\delta}{2},\infty\right)}(q)}{q^{2}}
\end{align*}
because $k\ge 1$ and $\psi(q)=0$ if $q\le \frac{\delta}{2}$. Substituting this into \eqref{depression} implies $\left|F_{m}^t\mathopen{}(x)\mathclose{}\right|\lesssim_p |x|^{-p}$ for each integer $p\ge 0$. This combined with \eqref{sonata} shows that $\left|F_{m}^t(x)\right|\lesssim_p \left\langle x\right\rangle^{-p}$ for each $p\ge 0$, so \eqref{wakeup} follows.
% \begin{align}
%     \left|F_{m}^t\mathopen{}(x)\mathclose{}\right|\lesssim_{p} \left\langle x\right\rangle^{-p}\textrm{.}
%     \label{wakeup}
% \end{align} 
Taking the absolute value of both sides of \eqref{proof}, we get the inequality that for $\lVert \omega\rVert\ge 1$, \begin{align*}
    \left|\widehat{G_r^t}(\omega)\right|\lesssim_{k,p} \frac{\left\langle r - \lVert \omega\rVert\right\rangle^{-p}}{\lVert \omega\rVert^{\frac1{2}}}+O\left(\frac1{\lVert \omega\rVert^{M+\frac3{2}}}\right)\textrm{,}
\end{align*}
since $\sigma \in \{\pm 1\}$ and $M$ is some large integer. This completes the proof of Lemma \ref{mainequationtoshow}.
\end{proof}

We now show Lemma \ref{intertoshow} using from Lemma \ref{mainequationtoshow}, hence completing the proof of Proposition \ref{prop:pseudoperiod}.
\begin{proof}[Proof of Lemma \ref{intertoshow}]
     We start by noting it is sufficient to show the lemma for $\tau=1$; the case of $\tau=-1$ follows by conjugation. We first expand the sum in the lemma using the relation $\ell(\gamma)=\sum_{\substack{s \in S\setminus\{0\} \\ \left\lVert s\right\rVert=\gamma}}c_s$ to obtain
     \begin{align}
     \sum_{\gamma \in S_{\textrm{rad}}}\frac{\ell(\gamma)\Phi(t\gamma)e^{2\pi i \tau r \gamma}}{\gamma^{\frac3{2}+k}}=\sum_{s \in S\setminus\{0\}}\frac{c_s\Phi\left(t\lVert s\rVert \right)e^{2\pi i \tau r \lVert s\rVert}}{\lVert s\rVert^{\frac3{2}+k}}\textrm{.}
     \label{flames}
\end{align}
As $S$ is discrete, let $\delta>0$ be such that $S\cap B_{\delta}=\{0\}$. Let $\psi\colon\mathbb{R} \to [0,1]$ and $G_r^t$ be as defined in Lemma \ref{mainequationtoshow}, and for reference,
    \begin{align}
        G_{r}^t(x)=\frac{\Phi\mathopen{}\left(t\lVert x\rVert\right)\mathclose{}\psi\left(\lVert x\rVert\right)e^{2\pi i \tau r \lVert x\rVert}}{\lVert x\rVert^{\frac3{2}+k}}\textrm{.} 
        \label{gels}
\end{align}
% Note that this is the same $G_r^t$ introduced in Lemma \ref{mainequationtoshow}. 
Therefore, we have
\begin{align*}
    \sum_{s \in S\setminus\{0\}}\frac{c_s\Phi\mathopen{}\left(t\lVert s\rVert\right)\mathclose{}e^{2\pi i \tau r \lVert s\rVert}}{\lVert s\rVert^{\frac3{2}+k}}=\sum_{s \in S}c_sG_{r}^t(s)\textrm{,}
\end{align*} so we now want to show $\sum_{s \in S}c_sG_r^t(s)=O\left(r^{\frac1{2}}\right)$ independently of $t$. We apply the summation formula to $G_{r}^t$. We first check $G_{r}^t \in \mathcal{S}\left(\mathbb{R}^2\right)$. This follows because $\widehat{\varphi}$ is Schwartz, $t \in (0,1)$, and all the derivatives of \begin{align*}
    x \mapsto \frac{\psi\left(\lVert x\rVert\right)e^{2\pi i  r \lVert x\rVert }}{\lVert x\rVert^{\frac3{2}+k}}
\end{align*} are bounded over $\mathbb{R}^2$. The summation formula gives
\begin{align}
    \sum_{s \in S\setminus\{0\}}\frac{c_s\Phi\left(t\lVert s\rVert\right)e^{2\pi i r \lVert s\rVert}}{\lVert s\rVert^{\frac3{2}+k}}=\sum_{\lambda \in \Lambda}\widehat{G_{r}^t}(\lambda)\textrm{.}
    \label{losecontrowol}
\end{align}
To prove Lemma \ref{intertoshow}, it suffices to show that $\sum_{\lambda \in \Lambda}\widehat{G_{r}^t}(\lambda)=O\left(r^{\frac1{2}}\right)$ independently of $t$. Note we can bound $\widehat{G_r^t}$ independently of $t$ and $r$: $\Phi$ and $\psi$ are bounded, and $\psi(x)=0$ if $x\le \frac{\delta}{2}$ for some $\delta>0$, so from \eqref{gels}, \begin{align*}
    \left|\widehat{G_r^t}(\omega)\right|\lesssim \int_{\mathbb{R}^2\setminus B_{\frac{\delta}{2}}}\frac1{\lVert x\rVert^{\frac3{2}+k}}\textup{d}x<\infty\textrm{,}
\end{align*} since $k\ge 1$. Thus we can write \begin{align*}
    \sum_{\lambda \in \Lambda}\left|\widehat{G_r^t}(\lambda)\right|\lesssim 1+\sum_{\substack{\lambda \in \Lambda \\ \left\lVert \lambda\right\rVert\ge 1}}\left|\widehat{G_r^t}(\lambda)\right|\textrm{,}
\end{align*} and substituting the bound from Lemma \ref{mainequationtoshow}, \begin{align*}
     \sum_{\lambda \in \Lambda}\left|\widehat{G_r^t}(\lambda)\right|\lesssim_p \sum_{\substack{\lambda \in \Lambda \\ \lVert \lambda\rVert\ge 1}}\left(\frac{\left\langle r-\lVert \lambda\rVert\right\rangle^{-p}}{\lVert \lambda\rVert^{\frac1{2}}}+\frac1{\lVert \lambda\rVert^{M+\frac3{2}}}\right)\textrm{,}
\end{align*}
for some $M>1$. Theorem \ref{thm:pointwiseerrorbound} implies $\#\left(\Lambda \cap B_r\right)=O\left(r^2\right)$, and so from Proposition \ref{temperednesstogrowth} we may deduce $\sum_{\substack{\lambda \in \Lambda \\ \lVert \lambda\rVert\ge 1}}\lVert \lambda\rVert^{-M-\frac3{2}}<\infty$, so \begin{align}
    \sum_{\lambda \in \Lambda}\left|\widehat{G_r^t}(\lambda)\right|\lesssim_p 1+\sum_{\substack{\lambda \in \Lambda \\ \left\lVert \lambda\right\rVert\ge 1}}\frac{\left\langle r-\lVert \lambda \rVert\right\rangle^{-p}}{\lVert \lambda \rVert^{\frac1{2}}}\textrm{.}
    \label{thedoorts}
\end{align}
We now bound the last sum. We write this sum as \begin{align*}
    \sum_{\substack{\lambda \in \Lambda \\ \left\lVert \lambda\right\rVert\ge 1}} \frac{\left\langle r - \lVert \lambda\rVert \right\rangle^{-p}}{\lVert \lambda\rVert^{\frac1{2}}}&= \sum_{j=1}^\infty \sum_{\substack{\lambda \in \Lambda \\ j\le \lVert \lambda\lVert < j+1}}\frac{\left\langle r-\lVert \lambda\lVert \right\rangle^{-p}}{\lVert \lambda\lVert ^{\frac1{2}}}\le \sum_{j=1}^\infty \frac1{j^{\frac1{2}}}\sum_{\substack{\lambda \in \Lambda \\ j\le \lVert \lambda\lVert< j+1}}\left\langle r-\lVert \lambda\rVert\right\rangle^{-p}\textrm{,}
\end{align*} and extract the maximum term in the inner sum to get  \begin{align*}
    \sum_{\substack{\lambda \in \Lambda \\ \lVert \lambda\rVert \ge 1}} \frac{\left\langle r - \lVert \lambda\rVert\right\rangle^{-p}}{\lVert \lambda\rVert^{\frac1{2}}}\le \sum_{j=1}^\infty \frac{\#\left\{\lambda \in \Lambda\;\vert\; j\le \lVert \lambda\rVert< j+1 \right\}}{j^{\frac1{2}}}\max\left(\left\langle  r-\lVert \lambda\rVert \right\rangle^{-p}\;\vert\;\lambda \in \Lambda, \ j\le \lVert \lambda\rVert< j+1\right)\textrm{.}
\end{align*} Theorem \ref{thm:pointwiseerrorbound} implies $\#\left(\Lambda \cap B_r\right)=c_0\pi r^2+O(r)$, and so  $\#\left\{\lambda \in \Lambda\;\vert\; j \le \lVert \lambda\rVert <j+1  \right\}\lesssim j$ for $j\ge 1$. From the above inequality, this implies  \begin{align}
    \sum_{\substack{\lambda \in \Lambda \\ \lVert \lambda\rVert \ge 1}} \frac{\left\langle r-\lVert \lambda\rVert \right\rangle^{-p}}{\lVert \lambda\rVert ^{\frac1{2}}}&\lesssim_p \sum_{j=1}^\infty j^{\frac1{2}}\max\left(\left\langle  r-\lVert \lambda\rVert \right\rangle^{-p}\;\vert\;\lambda \in \Lambda, \ j\le \lVert \lambda\rVert\le j+1\right)\textrm{.} \label{die4u}
\end{align}
We consider $r>1$, and
analyze the max expression. 
% If $j+1\le r$, then for $j\le \lVert \lambda\rVert\le j+1$ we have $\left|r-\lVert \lambda\rVert \right|+1\ge r-1-j+1=r-j$. If $j\ge r$, then $\left|r-\lVert \lambda\rVert\right|\ge j-r$. If $r\le j\le r-1$, then $\left|r-\lVert \lambda\rVert \right|+1\ge 1$.
We see that
\begin{align*}
    \max\left(\left\langle  r-\lVert \lambda\lVert \right\rangle^{-p}\;\vert\;\lambda \in \Lambda, \ j\le \lVert \lambda\rVert\le j+1\right)\lesssim_p \begin{cases}
        \left(r-j\right)^{-p} & \text{ if } j+1\le r \\
        1 & \text{ if } r-1<j<r \\
        \left(j-r+1\right)^{-p} & \text{ if } j\ge r\textrm{.}
    \end{cases}
\end{align*}
Substituting this into \eqref{die4u} gives
\begin{align*}
   \sum_{\substack{\lambda \in \Lambda \\ \left\lVert \lambda\right\rVert\ge 1}} \frac{\left\langle r - \lVert \lambda\rVert\right\rangle^{-p}}{\lVert \lambda\rVert^{\frac1{2}}}&\lesssim_{p}\sum_{\substack{j \in \mathbb{N} \\ j\le r-1}}\frac{j^{\frac1{2}}}{(r-j)^p}+r^{\frac1{2}}+\sum_{\substack{ j \in \mathbb{N} \\ j\ge r}}\frac{j^{\frac1{2}}}{(j-r+1)^p} \\
    &\lesssim_{p}\sum_{\substack{j \in \mathbb{N} \\ j\le r-1}}\frac{ r^{\frac1{2}}}{(r-j)^p}+ r^{\frac1{2}}+\sum_{\substack{j \in \mathbb{N} \\ j\ge r}}\frac{j^{\frac1{2}}}{(j-r+1)^p}\textrm{.}
\end{align*}
If $p$ is sufficiently large, then
\begin{align*}
    \sum_{\substack{j\in \mathbb{N} \\ j\le r-1}}\frac1{(r-j)^p}\le \sum_{h=1}^\infty \frac1{h^p}<\infty\textrm{,}
\end{align*}
so by taking $p$ large enough, we get \begin{align*}
     \sum_{\substack{\lambda \in \Lambda \\ \lVert \lambda\rVert\ge 1}} \frac{\left\langle  r - \lVert \lambda\rVert\right\rangle^{-p}}{\lVert \lambda\rVert^{\frac1{2}}}&\lesssim_{p} r^{\frac1{2}}+\sum_{\substack{j \in \mathbb{N} \\ j\ge r}}\frac{j^{\frac1{2}}}{(j-r+1)^p}\textrm{.}
\end{align*} We split the second sum on the right-hand side, writing \begin{align*}
    \sum_{\substack{\lambda \in \Lambda \\ \lVert \lambda\rVert\ge 1}} \frac{\left\langle r -\lVert \lambda\rVert\right\rangle^{-p}}{\lVert \lambda\rVert^{\frac1{2}}} &\lesssim_{p} r^{\frac1{2}}+\sum_{\substack{j \in \mathbb{N} \\ r\le j< 2r}}\frac{j^{\frac1{2}}}{(j-r+1)^p}+\sum_{\substack{j \in \mathbb{N} \\ j\ge 2r}}\frac{j^{\frac1{2}}}{(j-r+1)^p} \\
    &\lesssim_{p} r^{\frac1{2}}\sum_{h=1}^\infty \frac1{h^p}+\sum_{\substack{j \in \mathbb{N} \\ j\ge 2r}}\frac{j^{\frac1{2}}}{(j-r+1)^p}\textrm{.}
\end{align*}
If $j\ge 2r$, then $j-r\ge \frac{j}{2}$, and so we get \begin{align*}
    \sum_{\substack{\lambda \in \Lambda \\ \lVert \lambda\rVert\ge 1}} \frac{\left\langle  r - \lVert \lambda\rVert \right\rangle^{-p}}{|\lambda|^{\frac1{2}}}\lesssim_{p} r^{\frac1{2}}+\sum_{\substack{j \in \mathbb{N} \\ j\ge 2r}}\frac1{j^{p-\frac1{2}}}\lesssim r^{\frac1{2}}\textrm{.}
\end{align*}
Substituting this into \eqref{thedoorts} implies \begin{align*}
    \sum_{\lambda \in \Lambda}\left|\widehat{G_r^t}(\lambda)\right|\lesssim r^{\frac1{2}},
\end{align*} and so from \eqref{flames} and \eqref{losecontrowol}, we get \begin{align*}
    \sum_{\gamma \in S_{\textrm{rad}}}\frac{\ell(\gamma)\Phi\left(t\gamma\right)e^{2\pi i \tau r\gamma}}{\gamma^{\frac3{2}+k}}=O\left(r^{\frac1{2}}\right)\textrm{,}
\end{align*}completing the proof of Lemma \ref{intertoshow}.
\end{proof}

\section{Average Bounds}\label{avg}
The goal of this section is to consider the error term $\Err\left(r,\Lambda\right)$ in an average sense. Hardy considered this problem in \cite{hardy1917average}, and we attempt to do something similar by proving Theorem \ref{thm:averageerrorbound}, restated below.
\averageerrorbound*
% While the theorem as stated only considers a uniform probability density, the bound works for any continuously differentiable probabilty density scaled to fit $[0,R]$, that is \begin{align*}
%     \frac1{R}\int_0^R\rho\left(\frac{r}{R}\right)\left|\Err(r,\Lambda)\right|\;\textup{d}r\lesssim_{\rho,\Lambda} R^{\frac1{2}}+R^{\frac{3N-4}{3N-1}}\textrm{,}
% \end{align*} for any probability density $\rho \in \mathcal{C}^1\left([0,1]\right)$.
% As a matter convenience, we consider the average in the left-hand side above.  
The proof proceeds in a number steps, which we outline, before delving into the technical details. The first involves bounding the average of $\Err(r,\Lambda)$ in terms of the average of $\Err_t(r,\Lambda)$:
\begin{lemma}
    Take $R>1$ large, $t \in \left(0,\frac1{2}\right)$. Then \begin{align*}
        \frac1{R}\int_0^R\left|\Err(r,\Lambda)\right|\textup{d}r\lesssim Rt+\frac1{R}\int_0^R\left|\Err_t(r,\Lambda)\right|\textup{d}r\textrm{.}
    \end{align*}
    \label{bounderrterr}
\end{lemma}
\begin{proof}[Proof of Lemma \ref{bounderrterr}]
    Recall $\Err_t$ is the smoothed-out error term of scale $t$ introduced in \eqref{smoothdeouterrorterm}, given by $\Err_t\left(r,\Lambda\right)=\sum_{\lambda \in \Lambda}\left(\mathbbm{1}_{B_r}*\varphi_t\right)(\lambda)-c_0\pi r^2$. Note that \begin{align*}
    \Err\left(r,\Lambda\right)-\Err_t\left(r,\Lambda\right)=\sum_{\lambda \in \Lambda}\mathbbm{1}_{B_r}(\lambda)-\sum_{\lambda \in \Lambda}\left(\mathbbm{1}_{B_r}*\varphi_t\right)(\lambda)\textrm{.}
\end{align*}
From Lemma \ref{iwouldnever}, we can deduce $\left(\mathbbm{1}_{B_r}*\varphi_t\right)(x)=\mathbbm{1}_{B_r}(x)$ if $\lVert x\rVert \le r-t$ or $\lVert x \rVert\ge r+t$, and both $\left|\mathbbm{1}_{B_r}\right|,\left|\mathbbm{1}_{B_r}*\varphi_t\right|\le 1$, so 
\begin{align}
    \left|\Err\left(r,\Lambda\right)-\Err_t\left(r,\Lambda\right)\right|\le 2 \#\left\{\lambda \in \Lambda \;\vert\; r-t< \left\lVert\lambda\right\rVert\le r+t\right\}=2  \#\left(\Lambda \cap B_{r+t}\right)-2  \#\left(\Lambda \cap B_{r-t}\right)\textrm{.}
    \label{rasa}
\end{align}
Observe that \begin{align*}
    \int_t^R\left(\#\left(\Lambda \cap B_{r+t}\right)-\#\left(\Lambda \cap B_{r-t}\right)\right)\textup{d}r&=\int_{2t}^{R+t}\#\left(\Lambda \cap B_r\right)\;\textup{d}r-\int_0^{R-t}\#\left(\Lambda \cap B_r\right)\textup{d}r \\
    &=\int_{R-t}^{R+t}\#\left(\Lambda \cap B_r\right)\textup{d}r-\int_0^{2t}\#\left(\Lambda \cap B_r\right)\textup{d}r\textrm{,}
\end{align*} so from \eqref{rasa},  \begin{align*}
    \int_0^R\left|\Err\left(r,\Lambda\right)-\Err_t\left(r,\Lambda\right)\right|\textup{d}r\le 2\int_{R-t}^{R+t}\#\left(\Lambda \cap B_r\right)\textup{d}r\lesssim \int_{R-t}^{R+t}r^2\textup{d}r\lesssim R^2t\textrm{.} 
\end{align*}
% From Theorem \ref{thm:pointwiseerrorbound}, we know $\#\left(\Lambda \cap B_r\right)=c_0\pi r^2+O\left(r\right)$. Combining this with the above inequality implies
% \begin{align*}
%       \int_t^R\left(\#\left(\Lambda \cap B_{r+t}\right)-\#\left(\Lambda \cap B_{r-t}\right)\right)\,dr&\lesssim R^2t+R t\textrm{,}
% \end{align*}
% so
% \begin{align*}
%     \frac1{R}\int_0^R\left|\Err\left(r,\Lambda\right)-\Err_t\left(r,\Lambda\right)\right|\,dr\lesssim Rt\textrm{.}
% \end{align*}
Applying the triangle inequality, we get \begin{align}
    \frac1{R}\int_0^R\left|\Err\left(r,\Lambda\right)\right|\textup{d}r\lesssim Rt+\frac1{R}\int_0^R\left|\Err_t\left(r,\Lambda\right)\right|\textup{d}r\textrm {.}
    %\label{oi}
\end{align}
\end{proof}
Our next step involves bounding the average of $\Err_t$. For $\tau \in \{-1,1\}$ and $t \in \left(0,\frac1{2}\right)$, we define
\begin{equation}
    X_{t,\tau}\mathopen{}\left(r\right)\mathclose{}\coloneqq\sum_{\gamma \in S_{\textrm{rad}}}\frac{\ell(\gamma)\Phi(t\gamma)e^{2\pi i\tau r\gamma}}{\gamma^{\frac3{2}}}\textrm{,} \label{Xttaudef}
\end{equation}
so that from Proposition \ref{prop:pseudoperiod}, we can write
\begin{align}
    \Err_t\left(r,\Lambda\right)= r^{\frac1{2}}\mathopen{}\sum_{\tau \in \{\pm 1\}}a_{\tau}X_{t,\tau}\mathopen{}\left(r\right)\mathclose{}\mathclose{}+O\left(1\right)\textrm{,}
    \label{sereisXkj}
\end{align}
where $a_\tau=\frac{\exp\left(-\frac{3\pi \tau i }{4}\right)}{2\pi}$. (We omit the $\Lambda$ in $X_{t,\tau}$ since we consider the Fourier quasicrystal to be fixed throughout our analysis.) We then control the average of $\Err_t$ by bounding the averages of $\left|X_{t,\tau}\right|^2$ for $\tau= -1,1$. We start by squaring $X_{t,\tau}(r)$ and integrating to get \begin{align*}
    \frac1{R}\int_0^R\left|X_{t,\tau}\mathopen{}\left(r\right)\mathclose{}\right|^2\textup{d}r&=\frac1{R}\int_0^RX_{t,\tau}\mathopen{}\left(r\right)\mathclose{}\overline{X_{t,\tau}\mathopen{}\left(r\right)\mathclose{}}\;\textup{d}r \\
    &=\frac1{R}\int_0^R\sum_{\gamma_1,\gamma_2 \in S_{\textrm{rad}}}\frac{\ell\left(\gamma_1\right)\overline{\ell\left(\gamma_2\right)}\Phi\left(t\gamma_1\right)\overline{\Phi\left(t\gamma_2\right)}e^{2\pi i \tau r\left(\gamma_1-\gamma_2\right)}}{\gamma_1^{\frac3{2}}\gamma_2^{\frac3{2}}}\;\textup{d}r\textrm{.}
\end{align*}
Recall that $\widehat{\varphi}$ is Schwartz and $\widehat{\varphi}(x)=\Phi\left(\lVert x\rVert\right)$, so by Fubini's theorem we may switch the sum with the integral to get \begin{align*}
    \frac1{R}\int_0^R\left|X_{t,\tau}\mathopen{}\left(r\right)\mathclose{}\right|^2\textup{d}r&=\sum_{\gamma_1,\gamma_2 \in S_{\textrm{rad}}}\frac{\ell\left(\gamma_1\right)\overline{\ell\left(\gamma_2\right)}\Phi\left(t\gamma_1\right)\overline{\Phi\left(t\gamma_2\right)}}{\gamma_1^{\frac3{2}}\gamma_2^{\frac3{2}}}\frac1{R}\int_0^Re^{2\pi i \tau r \left(\gamma_1 -\gamma_2\right)}\textup{d}r\textrm{.}
\end{align*}
Evaluating the integral, we get
\begin{align*}
    \frac1{R}\int_0^Re^{2\pi i \tau r \left(\gamma_1-\gamma_2 \right)}\textup{d}r=\frac1{R}\left[\frac{e^{2\pi i \tau r\left(\gamma_1-\gamma_2\right)}}{2\pi i \tau \left(\gamma_1-\gamma_2\right)}\right]_0^R=\frac{e^{2\pi i \tau R\left(\gamma_1-\gamma_2\right)}-1}{2\pi i \tau R \left(\gamma_1-\gamma_2\right)}\textrm{,}
\end{align*}
so
\begin{align*}
    \frac1{R}\int_0^R\left|X_{t,\tau}\mathopen{}\left(r\right)\mathclose{}\right|^2\textup{d}r&=\sum_{\gamma_1,\gamma_2 \in S_{\textrm{rad}}}\frac{\ell\left(\gamma_1\right)\overline{\ell\left(\gamma_2\right)}\Phi\left(t\gamma_1\right)\overline{\Phi\left(t\gamma_2\right)}}{\gamma_1^{\frac3{2}}\gamma_2^{\frac3{2}}}\cdot \frac{e^{2\pi i \tau R\left(\gamma_1-\gamma_2\right)}-1}{2\pi i \tau R \left(\gamma_1-\gamma_2\right)}\textrm{.}
\end{align*} We split this sum based on whether or not the difference $\gamma_2-\gamma_1$ is small. To that end, introduce the following quantities of $R$ and $\eps \in \left(0,\frac1{2}\right)$: $A_{1,t}(R,\eps)$, $A_{2,t}(R,\eps)$, defined by \begin{align}
    A_{1,t}\left(R,\eps\right)&\coloneqq\sum_{\substack{\gamma_1,\gamma_2 \in S_{\textrm{rad}} \\ \left|\gamma_1-\gamma_2\right|\le \eps }}\frac{\ell\left(\gamma_1\right)\overline{\ell\left(\gamma_2\right)}\Phi\left(t\gamma_1\right)\overline{\Phi\left(t\gamma_2\right)}}{\gamma_1^{\frac3{2}}\gamma_2^{\frac3{2}}}\cdot \frac{e^{2\pi i \tau R\left(\gamma_1-\gamma_2\right)}-1}{2\pi i \tau R \left(\gamma_1-\gamma_2\right)}\textrm{,} \label{Asub1t} \\
    A_{2,t}\left(R,\eps\right)&\coloneqq\sum_{\substack{\gamma_1,\gamma_2\in S_{\textrm{rad}} \\ \left|\gamma_1-\gamma_2\right| > \eps }}\frac{\ell\left(\gamma_1\right)\overline{\ell\left(\gamma_2\right)}\Phi\left(t\gamma_1\right)\overline{\Phi\left(t\gamma_2\right)}}{\gamma_1^{\frac3{2}}\gamma_2^{\frac3{2}}}\cdot \frac{e^{2\pi i \tau R\left(\gamma_1-\gamma_2\right)}-1}{2\pi i \tau R \left(\gamma_1-\gamma_2\right)}\textrm{,}\label{Asub2t}
\end{align} so that \begin{align}
    \frac1{R}\int_0^R\left|X_{t,\tau}(r)\right|^2\textup{d}r=A_{1,t}(R,\eps)+A_{2,t}(R,\eps)\textrm{.}
    \label{havefaith}
\end{align} The most difficult task involves bounding $A_{1,t}$ and $A_{2,t}$ in terms of $t$, $\eps$ and $R$, and optimizing the choice of $\eps$ to give the best bound. The following lemma gives us the desired bounds.
\begin{lemma}
    Let $\eps, t \in \left(0,\frac1{2}\right)$. Then we can bound $A_{1,t}(R,\eps)$ and $A_{2,t}(R,\eps)$ as defined in \eqref{Asub1t}, \eqref{Asub2t} by
    \begin{equation}
    \left|A_{1,t}(R,\eps)\right|\lesssim 1-\eps \log t +\frac{\eps}{t^{N-2}}+\frac1{t^{N-\frac7{3}}}\textrm{,}
        \label{boundonA1t}
    \end{equation} and
    \begin{equation}
    \left|A_{2,t}\left(R,\eps\right)\right|\lesssim \frac{1}{R\eps t^{N-1}}\textrm{.} \label{boundA2t}
\end{equation}
%     \begin{equation}
%     \frac1{R}\int_0^R\rho\left(\frac{r}{R}\right)\left|X_{t,\tau}(r)\right|^2dr\lesssim 1+\frac1{t^{N-\frac7{3}}}.
%     \label{therealofAsandintX}
% \end{equation}
    \label{thelemmaboundingA12t}
\end{lemma}
Assuming these bounds, we prove Theorem \ref{thm:averageerrorbound}, and then prove the Lemma: 
\begin{proof}[Proof of Theorem \ref{thm:averageerrorbound}]
    The first thing to do is optimize the choice of $\eps$. We first consider the case where $N=2$. In that case we set $\eps=t$ in \eqref{boundonA1t} and \eqref{boundA2t} and add the resulting bounds in \eqref{boundonA1t}, \eqref{boundA2t} to get \begin{align*}
    \left|A_{1,t}(R,\eps)\right|+\left|A_{2,t}(R,\eps)\right|\lesssim 1+\frac1{Rt^2}\textrm{.}
\end{align*}
For $N\ge 3$, we let $\eps=t^{N-2}$, and add the resulting bounds in \eqref{boundonA1t}, \eqref{boundA2t} to get \begin{align*}
    \left|A_{1,t}(R,\eps)\right|+\left|A_{2,t}(R,\eps)\right|\lesssim 1+\frac1{t^{N-\frac7{3}}}+\frac1{Rt^{2N-3}}.
\end{align*} The last bound above is the strongest of the two, so in general, we have \begin{align*}
    \left|A_{1,t}(R,\eps)\right|+\left|A_{2,t}(R,\eps)\right|\lesssim 1+\frac1{t^{N-\frac7{3}}}+\frac1{Rt^{2N-3}},
\end{align*} and so from \eqref{havefaith}, \begin{align*}
    \frac1{R}\int_0^R\left|X_{t,\tau}(r)\right|^2\textup{d}r\lesssim 1+\frac1{t^{N-\frac7{3}}}+\frac1{Rt^{2N-3}}.
\end{align*}
Combining this with \eqref{sereisXkj} implies \begin{align*}
    \frac1{R}\int_0^R\left|\Err_t\left(r,\Lambda\right)\right|^2\textup{d}r\lesssim R+\frac{R}{t^{N-\frac7{3}}}+\frac1{t^{2N-3}}.
\end{align*} 
The Cauchy-Schwarz inequality then implies
\begin{align*}
    \left(\int_0^R\frac1{R}\left|\Err_t\left(r,\Lambda\right)\right|\textup{d}r\right)^2\le \left(\int_0^R\frac1{R}\;\textup{d}r\right)\left(\int_0^R\frac1{R}\left|\Err_t\left(r,\Lambda\right)\right|^2\textup{d}r\right)\textrm{,}
\end{align*} so \begin{align*}
   \frac1{R}\int_0^R\left|\Err_t\left(r,\Lambda\right)\right|\textup{d}r\le  \left(\frac1{R}\int_0^R\left|\Err_t\left(r,\Lambda\right)\right|^2\;\textup{d}r\right)^{\frac1{2}}\lesssim  R^{\frac1{2}}+\frac{R^{\frac1{2}}}{t^{\frac{3N-7}{6}}}+\frac1{t^{N-\frac3{2}}}\textrm{.}
\end{align*}
 We then combine this the bound in Lemma \ref{bounderrterr} to get \begin{align}
     \frac1{R}\int_0^R\left|\Err\left(r,\Lambda\right)\right|\textup{d}r\lesssim R^{\frac1{2}}+\frac{R^{\frac1{2}}}{t^{\frac{3N-7}{6}}}+\frac1{t^{N-\frac3{2}}}+Rt\textrm{.}
     \label{finalres}
\end{align}
If $N=2$, we choose $t=R^{-\frac1{2}}$, and the above inequality reduces to\begin{align*}
    \frac1{R}\int_0^R\left|\Err\left(r,\Lambda\right)\right|\;\textup{d}r\lesssim R^{\frac1{2}}.
\end{align*} For $N\ge 3$, we choose $t=R^{-\frac{3}{3N-1}}$, so that
\begin{align*}
    Rt=\frac{R^{\frac1{2}}}{t^{\frac{3N-7}{6}}}=R^{\frac{3N-4}{3N-1}}\textrm{,}
\end{align*} and \begin{align*}
    \frac1{t^{N-\frac3{2}}}=R^{\frac3{3N-1}\cdot \frac{2N-3}{2}}\le R^{\frac{3N-4}{3N-1}}
\end{align*} in which case the inequality in \eqref{finalres} reduces to \begin{align*}
    \frac1{R}\int_0^R\left|\Err\left(r,\Lambda\right)\right|\textup{d}r\lesssim R^{\frac1{2}}+R^{\frac{3N-4}{3N-1}}\textrm{,}
\end{align*}
 proving Theorem \ref{thm:averageerrorbound}. 

\end{proof}
 
We have now reduced our task to proving Lemma ~\ref{thelemmaboundingA12t}, the proof of which is given below.
\begin{proof}[Proof of Lemma \ref{thelemmaboundingA12t}]
 We start with bounding $A_{2,t}$. Recall from \eqref{Asub2t} that \begin{align*}
     A_{2,t}\left(R,\eps\right)&=\sum_{\substack{\gamma_1,\gamma_2\in S_{\textrm{rad}} \\ \left|\gamma_1-\gamma_2\right| > \eps }}\frac{\ell\left(\gamma_1\right)\overline{\ell\left(\gamma_2\right)}\Phi\left(t\gamma_1\right)\overline{\Phi\left(t\gamma_2\right)}}{\gamma_1^{\frac3{2}}\gamma_2^{\frac3{2}}}\cdot \frac{e^{2\pi i \tau R\left(\gamma_1-\gamma_2\right)}-1}{2\pi i \tau R \left(\gamma_1-\gamma_2\right)}\textrm{.}
 \end{align*} Note that \begin{align*}
     \left| \frac{e^{2\pi i \tau R\left(\gamma_1-\gamma_2\right)}-1}{2\pi i \tau R \left(\gamma_1-\gamma_2\right)}\right|\le \frac1{R\left|\gamma_1-\gamma_2\right|}\textrm{,}
 \end{align*}
 so 
    \begin{equation}
    \left|A_{2,t}\left(R,\eps\right)\right|\le \frac{1}{R}\sum_{\substack{\gamma_1,\gamma_2\in S_{\textrm{rad}} \\ \left|\gamma_1-\gamma_2\right| > \eps }}\frac{\left|\ell\left(\gamma_1\right)\ell\left(\gamma_2\right)\Phi\left(t\gamma_1\right)\Phi\left(t\gamma_2\right)\right|}{\gamma_1^{\frac3{2}}\gamma_2^{\frac3{2}}\left|\gamma_2-\gamma_1\right|}\textrm{.}
    \label{timeless}
 \end{equation} 
 We have\begin{align*}
   \sum_{\substack{\gamma_1,\gamma_2\in S_{\textrm{rad}} \\ \left|\gamma_1-\gamma_2\right| > \eps }}\frac{\left|\ell\left(\gamma_1\right)\ell\left(\gamma_2\right)\Phi\left(t\gamma_1\right)\Phi\left(t\gamma_2\right)\right|}{\gamma_1^{\frac3{2}}\gamma_2^{\frac3{2}}\left|\gamma_2-\gamma_1\right|}&\le \sum_{\substack{\gamma_1,\gamma_2\in S_{\textrm{rad}} \\ \left|\gamma_1-\gamma_2\right| > \eps }}\frac{\left|\ell\left(\gamma_1\right)\ell\left(\gamma_2\right)\Phi\left(t\gamma_1\right)\Phi\left(t\gamma_2\right)\right|}{\eps\gamma_1^{\frac3{2}}\gamma_2^{\frac3{2}}} \\
   &\le \sum_{\substack{\gamma_1,\gamma_2\in S_{\textrm{rad}} }}\frac{\left|\ell\left(\gamma_1\right)\ell\left(\gamma_2\right)\Phi\left(t\gamma_1\right)\Phi\left(t\gamma_2\right)\right|}{\eps\gamma_1^{\frac3{2}}\gamma_2^{\frac3{2}}}\textrm{,}
\end{align*} and 
\begin{align*}
     \sum_{\substack{\gamma_1,\gamma_2\in S_{\textrm{rad}} }}\frac{\left|\ell\left(\gamma_1\right)\ell\left(\gamma_2\right)\Phi\left(t\gamma_1\right)\Phi\left(t\gamma_2\right)\right|}{\eps\gamma_1^{\frac3{2}}\gamma_2^{\frac3{2}}}= \frac{1}{\eps}\left(\sum_{\gamma \in S_{\textrm{rad}}}\frac{\left|\ell(\gamma)\Phi\left(t\gamma\right)\right|}{\gamma^{\frac3{2}}}\right)^2\textrm{,}
\end{align*}
 so\begin{align}
     \sum_{\substack{\gamma_1,\gamma_2\in S_{\textrm{rad}} \\ \left|\gamma_1-\gamma_2\right| > \eps }}\frac{\left|\ell\left(\gamma_1\right)\ell\left(\gamma_2\right)\Phi\left(t\gamma_1\right)\Phi\left(t\gamma_2\right)\right|}{\gamma_1^{\frac3{2}}\gamma_2^{\frac3{2}}\left|\gamma_2-\gamma_1\right|}  \le \frac{1}{\eps}\left(\sum_{\gamma \in S_{\textrm{rad}}}\frac{\left|\ell(\gamma)\Phi\left(t\gamma\right)\right|}{\gamma^{\frac3{2}}}\right)^2\textrm{.}
     \label{norules}
\end{align}
From Corollary \ref{lovemejeje} we can take $P=1+\frac{N}{2}$ for the coefficient growth rate, so $\sum_{s \in S\cap B_r}\left|c_s\right|=O\left(r^{1+\frac{N}{2}}\right)$. Using this value of $P$ in Proposition \ref{temperednesstogrowth} for the coefficients $\left(c_s\right)_{s \in S}$ and the fact that $\ell(\gamma)=\sum_{\substack{s \in S \\ \lVert s\rVert=\gamma}}c_s$ implies
\begin{equation*}
   \sum_{\gamma \in S_{\textrm{rad}}}\frac{\left|\ell(\gamma)\Phi\left(t\gamma\right)\right|}{\gamma^{\frac3{2}}}\le  \sum_{s \in S\setminus\{0\}}\frac{\left|c_s\widehat{\varphi}\left(ts\right)\right|}{\lVert s\rVert^{\frac3{2}}}\lesssim t^{-\frac{N-1}{2}}\textrm{,}
\end{equation*} and substituting this into \eqref{norules}, we get \begin{equation*}
    \sum_{\substack{\gamma_1,\gamma_2\in S_{\textrm{rad}} \\ \left|\gamma_1-\gamma_2\right| > \eps }}\frac{\left|\ell\left(\gamma_1\right)\ell\left(\gamma_2\right)\Phi\left(t\gamma_1\right)\Phi\left(t\gamma_2\right)\right|}{\gamma_1^{\frac3{2}}\gamma_2^{\frac3{2}}\left|\gamma_2-\gamma_1\right|} \lesssim \frac1{\eps t^{N-1}} \textrm{.}
\end{equation*}
From \eqref{timeless}, this implies \begin{align*}
     \left|A_{2,t}\left(R,\eps\right)\right|\lesssim \frac{1}{R\eps t^{N-1}}\textrm{.}
\end{align*}
 To complete the proof of the lemma we now have to show
\begin{align}
\left|A_{1,t}(R,\eps)\right|\lesssim 1-\eps \log t +\frac{\eps}{t^{N-2}}+\frac1{t^{N-\frac7{3}}}\textrm{.}
    \label{thingtoshow}
\end{align} 
 Recall from \eqref{Asub1t} that \begin{align*}
     A_{1,t}\left(R,\eps\right)&=\sum_{\substack{\gamma_1,\gamma_2 \in S_{\textrm{rad}} \\ \left|\gamma_1-\gamma_2\right|\le \eps }}\frac{\ell\left(\gamma_1\right)\overline{\ell\left(\gamma_2\right)}\Phi\left(t\gamma_1\right)\overline{\Phi\left(t\gamma_2\right)}}{\gamma_1^{\frac3{2}}\gamma_2^{\frac3{2}}}\cdot \frac{e^{2\pi i \tau R\left(\gamma_1-\gamma_2\right)}-1}{2\pi i \tau R \left(\gamma_1-\gamma_2\right)}\textrm{.}
\end{align*}
By the mean value theorem, \begin{align*}
    \left| \frac{e^{2\pi i \tau R\left(\gamma_1-\gamma_2\right)}-1}{2\pi i \tau R \left(\gamma_1-\gamma_2\right)}\right|\le \sup_{x \in \mathbb{R}}\left|\partial_x\left(e^{ix}\right)\right|=1\textrm{,}
\end{align*}
 so
\begin{equation*}
    \left|A_{1,t}(R,\eps)\right|\le  \sum_{\substack{\gamma_1,\gamma_2 \in S_{\textrm{rad}} \\ \left|\gamma_1-\gamma_2\right|\le \eps }}\frac{\left|\ell\left(\gamma_1\right)\ell\left(\gamma_2\right)\Phi\left(t\gamma_1\right)\Phi\left(t\gamma_2\right)\right|}{\gamma_1^{\frac3{2}}\gamma_2^{\frac3{2}}}\textrm{.}
\end{equation*}
Because $S$ is discrete we let $\delta>0$ be such that $S\cap B_{2\delta}=\{0\}$. We partition the radii into bins $(\varepsilon j, \varepsilon(j+1)]$ of size $\eps$, so 
\begin{align*}
    \left(2\delta,\infty\right)\subset \bigcup_{\substack{j \in \mathbb{N} \\ \eps j\ge \delta}}\left(\eps j,\eps(j+1)\right]\textrm{,}
\end{align*} since $\eps \left\lceil \delta\eps^{-1}\right\rceil\le 2\delta$ for $\eps$ small enough.
We now write the bound for $A_{1,t}\left(R,\eps\right)$ as \begin{align*}
    \left|A_{1,t}\left(R,\eps\right)\right|&\le \sum_{\substack{j \in \mathbb{N} \\ j\ge \delta\eps^{-1}}}\sum_{\substack{\gamma_1 \in S_{\textrm{rad}} \\ \gamma_1  \in \left(\eps j,\eps(j+1)\right]}}\sum_{\substack{\gamma_2 \in S_{\textrm{rad}} \\ \left|\gamma_2-\gamma_1\right|\le \eps}}\frac{\left|\ell\left(\gamma_1\right)\Phi\left(t\gamma_1\right)\right|}{\gamma_1^{\frac3{2}}}\frac{\left|\ell\left(\gamma_2\right)\Phi\left(t\gamma_2\right)\right|}{\gamma_2^{\frac3{2}}}, 
\end{align*}
where $\eps$ is sufficiently small such that $\delta\eps^{-1}>2$. If $\gamma_1 \in \left(\eps j,\eps(j+1)\right]$ and $\left|\gamma_2-\gamma_1\right|\le \eps$, then \begin{equation*}
    \gamma_2 \in \bigl[\gamma_1-\eps,\gamma_1+\eps\bigr] \subset\left(\eps(j-1),\eps(j+2)\right]\textrm{,}
\end{equation*} 
so this implies \begin{align*}
    \left|A_{1,t}\left(R,\eps\right)\right|&\le\sum_{\substack{j \in \mathbb{N} \\ j\ge \delta\eps^{-1}}} \sum_{\substack{\gamma_1 \in S_{\textrm{rad}} \\ \gamma_1 \in \left(\eps j,\eps(j+1)\right]}}\sum_{\substack{\gamma_2 \in S_{\textrm{rad}}\\ \gamma_2 \in\left(\eps (j-1),\eps(j+2)\right]}}\frac{\left|\ell\left(\gamma_1\right)\Phi\left(t\gamma_1\right)\right|}{\gamma_1^{\frac3{2}}}\frac{\left|\ell\left(\gamma_2\right)\Phi\left(t\gamma_2\right)\right|}{\gamma_2^{\frac3{2}}}\textrm{.} 
\end{align*} Note $\left(\eps j, \eps(j+1)\right] \subset \left(\eps(j-1),\eps(j+2)\right]$, so by summing $\gamma_1$ over the larger range $\left(\eps(j-1),\eps(j+2)\right]$, we get \begin{align*}
    \left|A_{1,t}\left(R,\eps\right)\right|\le \sum_{\substack{j \in \mathbb{N} \\ j\ge \delta\eps^{-1}}} \sum_{\substack{\gamma_1,\gamma_2 \in S_{\textrm{rad}} \\ \gamma_1,\gamma_2 \in \left(\eps (j-1),\eps(j+2)\right]}}\frac{\left|\ell\left(\gamma_1\right)\Phi\left(t\gamma_1\right)\right|}{\gamma_1^{\frac3{2}}}\frac{\left|\ell\left(\gamma_2\right)\Phi\left(t\gamma_2\right)\right|}{\gamma_2^{\frac3{2}}}\textrm{,}
\end{align*} and evaluating the right-hand side above, we get \begin{align*}
    \left|A_{1,t}(R,\eps)\right|\le \sum_{\substack{j \in \mathbb{N} \\ j\ge \delta\eps^{-1}}}\left(\sum_{\substack{\gamma \in S_{\textrm{rad}} \\ \gamma \in \left(\eps(j-1),\eps(j+2)\right]}}\frac{\left|\ell(\gamma)\Phi(t\gamma)\right|}{\gamma^{\frac3{2}}}\right)^2\textrm{.}
\end{align*}
Let $\delta_1=\frac{\delta}{2}$, so $j-1\ge \delta\eps^{-1}-1\ge \delta_1\eps^{-1}$. Re-indexing the sum $j \to j-1$, and using the fact that $\left|\Phi(x)\right|\lesssim_p \left\langle x\right\rangle ^{-p}$, by taking $p$ sufficiently large we get
\begin{equation}
    \left|A_{1,t}(R,\eps)\right|\lesssim_p \sum_{\substack{j \in \mathbb{N} \\ j\ge \delta_1\eps^{-1}}} \left(\sum_{\substack{\gamma \in S_{\textrm{rad}} \\ \gamma \in \left(\eps j,\eps(j+3)\right]}}\frac{\left|\ell(\gamma)\right|\left\langle t\gamma\right\rangle ^{-p}}{\gamma^{\frac3{2}}}\right)^2\textrm{.}
    \label{hurricane}
\end{equation} Recall $\ell(\gamma)=\sum_{\substack{s \in S\setminus\{0\} \\ \lVert s\rVert=\gamma}}c_s$, so \begin{align*}
    \left(\sum_{\substack{\gamma \in S_{\textrm{rad}} \\ \gamma \in \left(\eps j,\eps(j+3)\right]}}\frac{\left|\ell(\gamma)\right|\left\langle t\gamma\right\rangle ^{-p}}{\gamma^{\frac3{2}}}\right)^2\le \left(\sum_{\substack{s \in S \\ \lVert s\rVert \in \left(\eps j,\eps(j+3)\right]}}\frac{\left|c_s\right|\left\langle ts\right\rangle ^{-p}}{\lVert s\rVert^{\frac3{2}}}\right)^2\textrm{,}
\end{align*} and combining this with Cauchy-Schwarz, \begin{align*}
     \left(\sum_{\substack{\gamma \in S_{\textrm{rad}} \\ \gamma \in \left(\eps j,\eps(j+3)\right]}}\frac{\left|\ell(\gamma)\right|\left\langle t\gamma\right\rangle ^{-p}}{\gamma^{\frac3{2}}}\right)^2\le \#\left\{s \in S\;\vert\; \lVert s\rVert  \in \left(\eps j,\eps(j+3)\right]\right\}\sum_{\substack{s \in S \\ \lVert s\rVert \in \left(\eps j,\eps(j+3)\right]}}\frac{\left|c_s\right|^2\left\langle ts\right\rangle^{-2p}}{\lVert s\rVert^{3}}\textrm{.}
\end{align*} Summing over $j \in \mathbb{N}$, $j\ge \delta_1\eps^{-1}$, from \eqref{hurricane}, we get
% At this point, we're going to invoke Theorem \ref{thm:boundsoncsquared}, so we 'unwrap' the $\ell(\gamma)$ term, so the bound above implies \begin{align*}
%     \left|A_{1,t}(R,\eps)\right|\le \sum_{\substack{j \in \mathbb{N} \\ j\ge \delta_1\eps^{-1}}}\left(\sum_{\substack{s \in S\setminus\{0\} \\ \lVert s\rVert \in \left(\eps j,\eps(j+3)\right]}}\frac{\left|c_s\right|\left\langle ts\right\rangle^{-p}}{\lVert s\rVert^{\frac3{2}}}\right)^2\textrm{.}
% \end{align*}
% Applying the Cauchy-Schwarz inequality to the squared term gives
\begin{align*}
    \left|A_{1,t}\left(R,\eps\right)\right|\lesssim_p & \sum_{\substack{j \in \mathbb{N} \\ j\ge \delta_1\eps^{-1}}}\left( \#\left\{s \in S\;\vert\; \lVert s\rVert  \in \left(\eps j,\eps(j+3)\right]\right\}\sum_{\substack{s \in S \\ \lVert s\rVert \in \left(\eps j,\eps(j+3)\right]}}\frac{\left|c_s\right|^2\left\langle ts\right\rangle^{-2p}}{\lVert s\rVert^{3}} \right)\textrm{,}
\end{align*} and using the fact $\left\langle x\right\rangle^{-1}$ is decreasing in $\lVert x \rVert$, we get \begin{align*}
    \left|A_{1,t}\left(R,\eps\right)\right|\lesssim_p  \sum_{\substack{j \in \mathbb{N} \\ j\ge \delta_1\eps^{-1}}}\left( \frac{\#\left\{s \in S\;\vert\; \lVert s\rVert \in \left(\eps j,\eps(j+3)\right]\right\}\left\langle t\eps j\right\rangle^{-2p} }{\eps^3j^{3}}\sum_{\substack{s \in S\\ \lVert s\rVert \in \left(\eps j,\eps(j+3)\right]}}\left|c_s \right|^2\right)\textrm{.}
\end{align*}
We now invoke Theorem \ref{thm:boundsoncsquared}, which implies
\begin{equation*}
    \sum_{\substack{s \in S\setminus\{0\} \\ \lVert s\rVert \in \left(\eps j,\eps(j+3)\right]}}\left|c_s\right|^2=\sum_{\substack{s \in S\setminus\{0\} \\ \lVert s\rVert \le \eps(j+3)}}\left|c_s\right|^2-\sum_{\substack{s \in S\setminus\{0\} \\ \lVert s\rVert \le \eps j}}\left|c_s\right|^2= c_0\pi \left(\left(\eps j +3\eps\right)^2-\eps^2j^2\right)+O\left(\eps^{\frac2{3}}j^{\frac2{3}}\right)
    \lesssim \eps^2j+\eps^{\frac2{3}}j^{\frac2{3}}\textrm{,}
\end{equation*} since $j\eps$ is bounded below. This implies 
\begin{align*}
    \left|A_{1,t}(R,\eps)\right|\lesssim_p \sum_{\substack{j \in \mathbb{N} \\ j\ge \delta_1\eps^{-1}}}  \frac{\#\left\{s \in S\;\vert\;\eps j< \lVert s\rVert \le \eps(j+3)\right\}\left(\eps^{2}j+\eps^{\frac2{3}}j^{\frac2{3}}\right)}{\eps^{3}j^{3}\left\langle t\eps j\right\rangle^{2p}}\textrm{.}
\end{align*}
We now use dyadic decomposition on this sum. Let $\delta_2=\frac{\delta_1}{2}$, and $\eps$ be small enough so that \begin{align*}
    \left[\delta_1\eps^{-1},\infty\right) \subset \bigcup_{\substack{m \in \mathbb{N}\cup\{0\} \\ 2^m\ge \delta_2\eps^{-1}}}\left[2^m,2^{m+1}\right).
\end{align*}
Then
\begin{align*}
    &\left|A_{1,t}(R,\eps)\right|\lesssim_p \sum_{\substack{m \in \mathbb{N}\cup\{0\} \\ 2^m\ge \delta_2\eps^{-1}}} \sum_{j=2^m}^{2^{m+1}-1}\frac{\#\left\{s \in S\;\vert\;\eps j< \lVert s\rVert \le \eps(j+3)\right\}\left(\eps^{2}j+\eps^{\frac2{3}}j^{\frac2{3}}\right)}{\eps^{3}j^{3}\left\langle t\eps j\right\rangle^{2p}}\textrm{.}
\end{align*}
If $j\ge 2^m$, then \begin{align*}
    \frac{\left(\eps^{2}j+\eps^{\frac2{3}}j^{\frac2{3}}\right)}{\eps^{3}j^{3}\left\langle t\eps j\right\rangle^{2p}}=\frac1{\left\langle t\eps j\right\rangle^{2p}}\left(\frac1{\eps j^2}+\frac1{\eps^{\frac7{3}}j^{\frac7{3}}}\right)\le \frac1{\left\langle 2^mt\eps \right\rangle^{2p}}\left(\frac1{2^{2m}\eps }+\frac1{2^{\frac{7m}{3}}\eps^{\frac7{3}}}\right)\textrm{,}
\end{align*}
therefore,  \begin{align}
    &\left|A_{1,t}(R,\eps)\right| \lesssim_p \sum_{\substack{m \in \mathbb{N}\cup \{0\} \\ 2^m\ge \delta_2\eps^{-1}}} \left(\frac1{\left\langle 2^mt\eps \right\rangle^{2p}}\left(\frac1{2^{2m}\eps }+\frac1{2^{\frac{7m}{3}}\eps^{\frac7{3}}}\right)\sum_{j=2^m}^{2^{m+1}-1}\#\left\{s \in S\;\vert\;\eps j< \lVert s\rVert \le \eps(j+3)\right\}\right)\textrm{.} \label{aceofbase}
\end{align} Observe that \begin{align*}
    \sum_{j=2^m}^{2^{m+1}-1}\#\left\{s \in S\;\vert\; \eps j< \lVert s\rVert \le \eps(j+3)\right\}& \le 3\#\left\{s \in S \;\vert\;       \lVert s\rVert \in \left(2^m\eps,\left(2^{m+1}+2\right)\eps\right]\right\} \\
    &\le 3\#\left(S\cap B_{\eps\left(2^{m+1}+2\right)}\right) \textrm{.}
\end{align*} 
Using the growth condition $\#\left(S\cap B_{R}\right)\lesssim R^N$ for $R\ge 1$, we get $\#\left(S\cap B_{\eps\left(2^{m+1}+2\right)}\right) \lesssim \eps^N \left(2^{m+1}+2\right)^N$, and since $2^m\eps\ge \delta_2>0$, we can weaken this to $\#\left(S\cap B_{\eps\left(2^{m+1}+2\right)}\right)\lesssim \eps^N 2^{mN}$. Substituting this into the above displayed equation implies \begin{align*}
     \sum_{j=2^m}^{2^{m+1}-1}\#\left\{s \in S\;\vert\; \eps j< \lVert s\rVert \le \eps(j+3)\right\} \lesssim \eps^N 2^{mN}\textrm{,}
\end{align*}
and substituting this into \eqref{aceofbase} implies 
\begin{align}
    \left|A_{1,t}(R,\eps)\right|\lesssim_p \sum_{\substack{m \in \mathbb{N}\cup\{0\} \\  2^m\ge \delta_2\eps^{-1}}} \left(2^{m(N-2)}\eps^{N-1}+2^{m\left(N-\frac7{3}\right)}\eps^{N-\frac7{3}}\right)\left\langle 2^m\eps t\right\rangle^{-2p} .\label{whiskey}
\end{align}
We split this sum into $2^m\le (\eps t)^{-1}$ and  $2^m>(\eps t)^{-1}$, by defining \begin{align}
    L\left(R,\eps\right)&=\sum_{\substack{m \in \mathbb{N}\cup\{0\} \\ \delta_2\eps^{-1}\le 2^m\le \left(\eps t\right)^{-1}}}\left(2^{m(N-2)}\eps^{N-1}+2^{m\left(N-\frac7{3}\right)}\eps^{N-\frac7{3}}\right)\left\langle 2^m\eps t\right\rangle^{-2p}\textrm{,} \label{lower} \\
    U\left(R,\eps\right)&= \sum_{\substack{m \in \mathbb{N}\cup\{0\} \\ 2^m>(\eps t)^{-1}}}\left(2^{m(N-2)}\eps^{N-1}+2^{m\left(N-\frac7{3}\right)}\eps^{N-\frac7{3}}\right)\left\langle 2^m\eps t\right\rangle^{-2p}\textrm{,}
\end{align} so that \begin{equation}
    \left|A_{1,t}\left(R,\eps\right)\right|\lesssim L\left(R,\eps\right)+U\left(R,\eps\right)\textrm{.}\label{sumaaas}
\end{equation} We start by bounding $U(R,\eps)$. We use the fact that $\left\langle 2^m\eps t\right\rangle^{-2p}\le 2^{-2mp}\left(\eps t\right)^{-2p}$ to get \begin{align}
    U(R,\eps)\le \left(\eps t\right)^{-2p}\sum_{\substack{m \in \mathbb{N}\cup\{0\} \\ 2^m >(\eps t)^{-1}}}\left(\eps^{N-1}2^{-m\left(2p+2-N\right)}+\eps^{N-\frac7{3}}2^{-m\left(2p+\frac7{3}-N\right)}\right)\textrm{.}
    \label{andimscreaming}
\end{align} Using the formula for the sum of a geometric series and taking $p$ large enough, we have the inequalities \begin{align*}
     \sum_{\substack{m \in \mathbb{N}\cup\{0\} \\ 2^m>(\eps t)^{-1}}}2^{-m\left(2p+2-N\right)}=2^{-\left(\left\lceil 1-\log_2(\eps t)\right\rceil\right)\left(2p+2-N\right)}\sum_{h=0}^\infty 2^{-h\left(2p+2-N\right)}\lesssim_{p,N} \left(\eps t\right)^{2p+2-N}\textrm{,}
\end{align*} and similarly, \begin{align*}
     \sum_{\substack{m \in \mathbb{N}\cup\{0\} \\ 2^m>(\eps t)^{-1}}}2^{-m\left(2p+\frac7{3}-N\right)}=2^{-\left(\left\lceil 1-\log_2(\eps t)\right\rceil\right)\left(2p+\frac7{3}-N\right)}\sum_{h=0}^\infty 2^{-h\left(2p+\frac7{3}-N\right)}\lesssim_{p,N} \left(\eps t\right)^{2p+\frac7{3}-N}\textrm{.}
\end{align*}
Substituting these bounds into \eqref{andimscreaming} , we get \begin{equation}
U\left(R,\eps\right)\lesssim \frac{\eps}{t^{N-2}}+\frac1{t^{N-\frac7{3}}}\textrm{.}
    \label{boundonU}
\end{equation} For $L\left(R,\eps\right)$, we use the fact that $\left\langle 2^m\eps t\right\rangle^{-2p}\le 1$, substituting this into \eqref{lower} to get \begin{align}
    L\left(R,\eps\right)\le \sum_{\substack{m \in \mathbb{N}\cup\{0\} \\ \delta_2\eps^{-1}\le 2^m\le \left(\eps t\right)^{-1}}}\left(2^{m(N-2)}\eps^{N-1}+2^{m\left(N-\frac7{3}\right)}\eps^{N-\frac7{3}}\right)\textrm{.}\label{Lee}
\end{align}
We now consider $2$ cases: $N=2$ and $N \geq 3$. When $N=2$, note \begin{align*}
    \sum_{\substack{m \in \mathbb{N}\cup\{0\} \\ \delta_2\eps^{-1}\le 2^m\le \left(\eps t\right)^{-1}}}\left(\eps+\eps^{-\frac1{3}}2^{-\frac{m}{3}}\right)\le \eps \#\left\{m \in \mathbb{N}\cup\{0\} \;\vert\; \delta_2\eps^{-1}\le 2^m\le (\eps t)^{-1} \right\}+\eps^{-\frac1{3}}\sum_{\substack{m \in \mathbb{N}\cup\{0\} \\ 2^{m}\ge \delta_2\eps^{-1}}}2^{-\frac{m}{3}}\textrm{,}
\end{align*}
so the inequality in \eqref{Lee} implies \begin{align*}
    L(R,\eps)\le  \eps \#\left\{m \in \mathbb{N}\cup\{0\} \;\vert\; \delta_2\eps^{-1}\le 2^m\le (\eps t)^{-1} \right\}+\eps^{-\frac1{3}}\sum_{\substack{m \in \mathbb{N}\cup\{0\} \\ 2^{m}\ge \delta_2\eps^{-1}}}2^{-\frac{m}{3}}.
\end{align*} We note $\#\left\{m \in \mathbb{N}\cup \{0\} \;\vert\; \delta_2\eps^{-1}\le 2^m\le (\eps t)^{-1} \right\}\lesssim -\log t$ when $t \in \left(0,\frac1{2}\right)$, and \begin{equation*}
    \sum_{\substack{m \in \mathbb{N}\cup\{0\} \\ 2^{m}\ge \delta_2\eps^{-1}}}2^{-\frac{m}{3}}=2^{-\frac{\left\lceil \log_2\left(\delta_2\eps^{-1}\right)\right\rceil}{3}}\sum_{k=0}^\infty 2^{-\frac{k}{3}} \le \eps^{\frac1{3}}\delta_2^{-\frac1{3}}\sum_{k=0}^\infty 2^{-\frac{k}{3}}\lesssim \eps^{\frac1{3}}\textrm{,} 
\end{equation*}
so this tells us $L(R,\eps)\lesssim 1-\eps\log t$  for $N=2$. Adding this to \eqref{boundonU} implies \begin{equation}
    L(R,\eps)+U(R,\eps)\lesssim 1-\eps \log t+\frac{\eps}{t^{N-2}}+\frac1{t^{N-\frac7{3}}} \quad\text{for }N=2.
    \label{N=2}
\end{equation}

For $N\ge 3$, by summing the relevant geometric series, we have \begin{align*}
    \sum_{\substack{m \in \mathbb{N}\cup \{0\} \\ 2^m\le (\eps t)^{-1}}}2^{m(N-2)}&=\sum_{m=0}^{\left\lfloor -\log_2(\eps t)\right\rfloor}2^{m(N-2)}=\frac{2^{\left\lfloor 1-\log_2(\eps t)\right\rfloor (N-2)}-1}{2^{N-2}-1}\lesssim \left(\eps t\right)^{2-N}\textrm{,} \\
       \sum_{\substack{m \in \mathbb{N}\cup \{0\} \\ 2^m\le (\eps t)^{-1}}}2^{m\left(N-\frac7{3}\right)}&=\sum_{m=0}^{\left\lfloor -\log_2(\eps t)\right\rfloor}2^{m\left(N-\frac7{3}\right)}=\frac{2^{\left\lfloor 1-\log_2(\eps t)\right\rfloor \left(N-\frac7{3}\right)}-1}{2^{N-\frac7{3}}-1}\lesssim \left(\eps t\right)^{\frac7{3}-N}\textrm{.}
\end{align*} Substituting these bounds into the right-hand side of \eqref{Lee}, we get
\begin{align*}
    L(R,\eps)\lesssim \eps^{N-1}(\eps t)^{2-N}+\eps^{N-\frac7{3}}(\eps t)^{\frac7{3}-N}= \frac{\eps}{t^{N-2}}+\frac1{t^{N-\frac7{3}}}\textrm{.}
\end{align*} 
Adding this to \eqref{boundonU} implies \begin{align*}
    L(R,\eps)+U(R,\eps)\lesssim\frac{\eps}{t^{N-2}}+\frac1{t^{N-\frac7{3}}},
\end{align*} for $N\ge 3$. Combining this with the result in \eqref{N=2} implies that in general,  \begin{align*}
    L(R,\eps)+U(R,\eps)\lesssim 1-\eps\log t+\frac{\eps}{t^{N-2}}+\frac1{t^{N-\frac7{3}}}\textrm{.}
\end{align*} Substituting this into \eqref{sumaaas} implies \begin{align*}
    \left|A_{1,t}(R,\eps)\right|\lesssim 1-\eps\log t +\frac{\eps}{t^{N-2}}+\frac1{t^{N-\frac7{3}}},
\end{align*} as desired. This completes the proof of Lemma \ref{thelemmaboundingA12t}.
 \end{proof}

\section{Lower Bounds for the error}\label{lb}
We consider lower bounds for the average. For $\mathbb{Z}^2$, Hardy \cite{hardy1915expression} showed that $\Err\left(r,\mathbb{Z}^2\right) = \Omega\left(r^{\frac1{2}}\left(\log r\right)^{\frac1{4}} \right)$. Here we perform a similar analysis. Define smoothed out normalized error term of scale $t$, $\Ern_t$, by $\Ern_t(r,\Lambda)\coloneqq r^{-\frac1{2}}\Err_t\left(r,\Lambda\right)$. From the definition of $\Err_t$ in \eqref{smoothdeouterrorterm}, we have the expression \begin{equation}
    \Ern_t\left(r,\Lambda\right)=r^{-\frac1{2}}\left(\sum_{\lambda \in \Lambda}\left(\mathbbm{1}_{B_r}*\varphi_t\right)(\lambda)-c_0\pi r^2\right)\textrm{.}
    \label{despacito}
\end{equation} 
Recall the definitions of $S_{\textrm{rad}}$ and $\ell(\gamma)$ from \eqref{Srad} and \eqref{agammadef},
\begin{align*}
    S_\textrm{rad}&\coloneqq \left\{\lVert s\rVert \;\vert\;s \in S\setminus\{0\}\right\}\textrm{,}\\
    \ell(\gamma)&\coloneqq\sum_{\substack{s \in S\setminus\{0\} \\ \lVert s \rVert=\gamma}}c_s\textrm{.}
\end{align*}
For convenience, we restate Theorem \ref{thm:lowerboundLtwoaverage} below.
\lowerboundLtwoaverage*

We first show the lower bound holds with respect to a smooth probability density function $\rho \in \mathcal{C}_c^\infty\left((0,1)\right)$. In other words,
\begin{align*}
    \liminf_{R \to \infty}\frac1{R}\int_0^R\rho\left(\frac{r}{R}\right)\left|\Ern(r,\Lambda)\right|^2dr\ge \frac1{2\pi^2}\sum_{\gamma \in S_\textrm{rad}}\frac{\left|\ell(\gamma)\right|^2}{\gamma^3}\textrm{,}
\end{align*} and approximate the constant function on $[0,1]$ from below. We show the above inequality by treating $\Ern$ as a Besicovitch $B^2$ almost periodic function, beginining by calculating its generalized Fourier coefficients, given by the following lemma:
\begin{lemma}
Let $\rho \in \mathcal{C}_c^\infty((0,1))$ be a non-negative function, with $\rho(x)\le 1$ for all $x$. Then for any $\xi \in \mathbb{R}$, we have \begin{align*}
    \lim_{R \to \infty}\frac1{R}\int_0^R\rho\left(\frac{r}{R}\right)\Ern\left(r,\Lambda\right)e^{-2\pi i \xi r}\textup{d}r=\frac{\widehat{\rho}(0)\exp\left(-\frac{3\pi \sgn(\xi) i}{4}\right)}{2\pi}\cdot \frac{\ell\left(|\xi|\right)}{|\xi|^{\frac3{2}}}\textrm{.}
\end{align*}
    \label{besicovitchlemmathing}
\end{lemma} Once we establish this lemma, we then show we can bound the average of $|\Ern|^2$ from below by sum of these coefficients squared. We first prove Theorem \ref{thm:lowerboundLtwoaverage} under this lemma, and then prove the lemma.
\begin{proof}[Proof of Theorem \ref{thm:lowerboundLtwoaverage}]
We first establish the second inequality in Theorem \ref{thm:lowerboundLtwoaverage} that
\begin{align*}
    \sum_{\gamma \in S_{\textrm{rad}}}\frac{\left|\ell(\gamma)\right|^2}{\gamma^3}>0\textrm{.}
\end{align*}
Suppose for the sake of contradiction that $\ell(\gamma)=0$
    for all $\gamma \in S_\textrm{rad}$. Let $f \in \mathcal{S}\left(\mathbb{R}^2\right)$ be radial, and let $F(r)$ for $r>0$ denote the radial component of $f$, so $f\left(x\right)=F\left(\lVert x\rVert\right)$. Substituting $f$ into the summation formula yields \begin{align*}
        \sum_{\lambda \in \Lambda}\widehat{f}(\lambda)=c_0f(0)+\sum_{\gamma \in S_{\textrm{rad}}}\left(\sum_{\substack{s \in S \\ \lVert s\rVert=\gamma}}c_s\right)F(\gamma)=c_0f(0)+\sum_{\gamma \in S_{\textrm{rad}}}\ell(\gamma)F(\gamma)=c_0f(0),
    \end{align*} since $\ell(\gamma)=0$ for all $\gamma \in S_{\textrm{rad}}$ by assumption. Let $g \in \mathcal{C}_c^\infty\left(B_1\right)$ be radial with $\widehat{g}(0)=1$ and $g(0)>0$. Pick $T>0$, and define $f$ by $f(x)=\widehat{g}(Tx)$ in the above equality, to get 
   \begin{align}\label{eqn:bounding-with-T}
       \sum_{\lambda \in \Lambda}\frac1{T^2}g\left(\frac{\lambda}{T}\right)=c_0\textrm{.}
    \end{align}
    We now let $T \to 0$. Recall that we chose $g$ to be supported on $B_1$. Since $\Lambda$ is discrete, let $\eta$ be the separation around $0$, so $\Lambda\cap B_\eta \subset \{0\}$. If $0 \notin \Lambda$, then the sum in \eqref{eqn:bounding-with-T} is eventually empty when $T < \eta$, so we obtain $c_0 = 0$, a contradiction. On the other hand, if $0 \in \Lambda$, then when $T < \eta$, the sum in \eqref{eqn:bounding-with-T} becomes $\frac{g(0)}{T^2}$. Hence we have $\frac{g(0)}{T^2}=c_0$, a contradiction since $\frac{g(0)}{T^2}$ is unbounded as $T \to 0$ but $c_0$ is finite. Both cases lead to a contradiction, so there exists some $\gamma \in S_\textrm{rad}$ for which $\ell(\gamma)\neq 0$, as desired.\newline
    
    We are now left with showing that \begin{align*}
         \liminf_{R \to \infty}\frac1{R}\int_1^R\left|\Ern(r,\Lambda)\right|^2\textup{d}r\ge \frac1{2\pi^2}\sum_{\gamma \in S_\textrm{rad}}\frac{\left|\ell(\gamma)\right|^2}{\gamma^3}\textrm{.}
    \end{align*} 

Let $\rho \in \mathcal{C}_c^\infty\left((0,1)\right)$ be a fixed but arbitrary non-negative function, with $\rho(x)\le 1$ for all $x$. Note in that case
 \begin{align}
         \liminf_{R \to \infty}\frac1{R}\int_1^R\left|\Ern(r,\Lambda)\right|^2\textup{d}r\ge \liminf_{R \to \infty}\frac1{R}\int_0^R\rho\left(\frac{r}{R}\right)\left|\Ern(r,\Lambda)\right|^2\textup{d}r\textrm{.}
         \label{corybooker}
    \end{align}
    Note we can use a lower limit of $1$ instead of $0$ in the integral on the left-hand side, since $\rho$ is compactly supported in $(0,1)$. For functions $f,g:[0,\infty) \to \mathbb{C}$, we define the Hermitian positive semi-definite form $\left\langle f,g\right\rangle_H$ (assuming it exists) as the following limit: \begin{align*}
    \left\langle f,g\right\rangle_\infty\coloneqq\lim_{R \to \infty}\frac1{R}\int_0^R\rho\left(\frac{r}{R}\right)f(r)\overline{g(r)}\,\textup{d}r\textrm{.}
\end{align*} We can then interpret Lemma \ref{besicovitchlemmathing} as saying that
\begin{align*}
    \left\langle \Ern,e^{2\pi i \xi r}\right\rangle_\infty=\frac{\widehat{\rho}(0)\exp\left(-\frac{3\pi \sgn(\xi) i}{4}\right)}{2\pi}\cdot \frac{\ell\left(|\xi|\right)}{|\xi|^{\frac3{2}}}\textrm{.}
\end{align*} Here we abuse notation and denote by $e^{2\pi i \xi r}$ the function $r \mapsto e^{2\pi i \xi r}$.
We can verify that for any $\xi_1,\xi_2 \in \mathbb{R}$ we have $\left\langle e^{2\pi i \xi_1 r},e^{2\pi i \xi_2 r}\right\rangle_\infty =\widehat{\rho}(0)\delta_{\xi_1,\xi_2}$, as
\begin{align*}
    \left\langle e^{2\pi i \xi_1 r},e^{2\pi i \xi_2 r}\right\rangle_\infty=\lim_{R \to \infty}\frac1{R}\int_0^R\rho\left(\frac{r}{R}\right)e^{2\pi i r\left(\xi_1-\xi_2\right)}\textup{d}r
\end{align*} and making the substitution $r=Rq$ yields
\begin{align*}
     \left\langle e^{2\pi i \xi_1 r},e^{2\pi i \xi_2 r}\right\rangle_\infty=\lim_{R \to \infty}\int_0^1\rho\left(q\right)e^{-2\pi i q R\left(\xi_2-\xi_1\right)}\textup{d}q=\lim_{R \to \infty}\widehat{\rho}\left(R\left(\xi_2-\xi_1\right)\right)\textrm{.}
\end{align*}
The limit in the right-hand side above is $\widehat{\rho}(0)$ if $\xi_1=\xi_2$, and $0$ otherwise. This implies that the set $\left\{r \mapsto e^{2\pi i \xi r}\;\vert\;\xi \in \mathbb{R}\right\}$ is an orthogonal set with respect to the hermitian form $\left\langle \cdot, \cdot\right\rangle_\infty$. We define another form $\left\langle \cdot,\cdot\right\rangle_R$ by 
\begin{align*}
    \left\langle f,g\right\rangle_R=\frac1{R}\int_0^R\rho\left(\frac{r}{R}\right)f(r)\overline{g(r)}\;\textup{d}r,
\end{align*} so $\lim_{R \to \infty}\left\langle f,g\right\rangle_R=\left\langle f,g\right\rangle_\infty$, if the limit exists. Let $\Gamma$ be a finite subset of $S_\textrm{rad}\cup -S_\textrm{rad}$, and define $G:[0,\infty) \to \mathbb{C}$ by \begin{align*}
    G(r)=\frac1{\widehat{\rho}(0)}\sum_{\gamma \in \Gamma}\left\langle \Ern,e^{2\pi i \gamma r}\right\rangle_\infty e^{2\pi i \gamma r}.
\end{align*}
We write \begin{align*}
    \left\langle \Ern,\Ern\right\rangle_R &=\left\langle \Ern-G+G,  \Ern-G+G\right\rangle_R =\left\langle \Ern-G,\Ern-G\right\rangle_R+ 2\Re\left\langle \Ern-G,G\right\rangle_R+\left\langle G,G\right\rangle_R\textrm{,} 
\end{align*}
and since $\left\langle \cdot,\cdot\right\rangle_R$ is a Hermitian positive semi-definite form, we obtain
\begin{equation*}
  \left\langle \Ern,\Ern\right\rangle_R\ge 2\Re\left\langle \Ern-G,G\right\rangle_R+\left\langle G,G\right\rangle_R\textrm{,}
\end{equation*} and taking the limit inferior of both sides as $R \to \infty$, we get \begin{align*}
    \liminf_{R \to \infty}\left\langle \Ern,\Ern\right\rangle_R\ge 2\Re\left\langle \Ern-G,G\right\rangle_\infty+\left\langle G,G\right\rangle_\infty\textrm{.}
\end{align*} Note from the definition of $G$ that \begin{align*}
    \left\langle G,G\right\rangle_\infty =\frac1{\widehat{\rho}(0)^2}\sum_{\gamma \in \Gamma}\left|\left\langle \Ern,e^{2\pi i \gamma r}\right\rangle_\infty\right|^2 \widehat{\rho}(0)=\frac1{\widehat{\rho}(0)}\sum_{\gamma \in \Gamma}\left|\left\langle \Ern,e^{2\pi i \gamma r}\right\rangle_\infty\right|^2\textrm{,}
\end{align*} and so \begin{align}
     \liminf_{R \to \infty}\left\langle \Ern,\Ern\right\rangle_R\ge 2\Re\left\langle \Ern-G,G\right\rangle_\infty+\frac1{\widehat{\rho}(0)}\sum_{\gamma \in \Gamma}\left|\left\langle \Ern,e^{2\pi i \gamma r}\right\rangle_\infty\right|^2\textrm{.}
     \label{finetune}
\end{align}
We have
\begin{align}
   \left\langle \Ern-G,G \right\rangle_\infty=\sum_{\gamma_2 \in \Gamma}\frac{\left\langle \Ern,e^{2\pi i \gamma_2 r}\right\rangle_\infty}{\widehat{\rho}(0)}\left\langle \Ern-\frac1{\widehat{\rho}(0)}\sum_{\gamma_1 \in \Gamma}\left\langle \Ern,e^{2\pi i \gamma_1 r}\right\rangle_\infty e^{2\pi i \gamma_1 r},e^{2\pi i \gamma_2 r}\right\rangle_\infty\textrm{.} \label{kursz}
\end{align}
Because $\left\{r \mapsto e^{2\pi i \xi r}\;\vert\;\xi \in \mathbb{R}\right\}$ is an orthogonal set with respect to  $\left\langle\cdot,\cdot\right\rangle_\infty$, for any $\gamma_2 \in \Gamma$, \begin{align*}
    \left\langle \sum_{\gamma_1 \in \Gamma}\left\langle \Ern,e^{2\pi i \gamma_1 r}\right\rangle_\infty e^{2\pi i \gamma_1 r},e^{2\pi i \gamma_2 r}\right\rangle_\infty&=\left\langle \Ern,e^{2\pi i \gamma_2r}\right\rangle_\infty \left\langle e^{2\pi i \gamma_2r},e^{2\pi i \gamma_2r}\right\rangle_\infty \\
    &=\widehat{\rho}(0)\left\langle \Ern,e^{2\pi i \gamma_2 r}\right\rangle_\infty\textrm{,}
\end{align*}
so
\begin{align*}
     \left\langle \Ern-\frac1{\widehat{\rho}(0)}\sum_{\gamma_1 \in \Gamma}\left\langle \Ern,e^{2\pi i \gamma_1 r}\right\rangle_\infty e^{2\pi i \gamma_1 r},e^{2\pi i \gamma_2 r}\right\rangle_\infty=0\textrm{.}
\end{align*}
From \eqref{kursz}, this implies $\left\langle \Ern-G,G\right\rangle_\infty=0$, and so \eqref{finetune} becomes
\begin{align*}
    \liminf_{R \to \infty}\left\langle \Ern,\Ern\right\rangle_R\ge \frac1{\widehat{\rho}(0)}\sum_{\gamma \in \Gamma}\left|\left\langle \Ern,e^{2\pi i \gamma r}\right\rangle\right|^2\textrm{,}
\end{align*}
where $\Gamma$ is a finite subset of $S_\textrm{rad}\cup -S_\textrm{rad}$. We have from Lemma \ref{besicovitchlemmathing} that
\begin{align*}
    \left|\left\langle \Ern,e^{2\pi i \gamma r}\right\rangle\right|^2=\frac{\widehat{\rho}(0)^2\left|\ell\left(|\gamma|\right)\right|^2}{4\pi^2|\gamma|^3}\textrm{,}
\end{align*}
so taking $\Gamma=S_0 \cup -S_0$ for any finite subset $S_0 \subset S_\textrm{rad}$, this implies \begin{align*}
    \liminf_{R \to \infty}\left\langle \Ern,\Ern\right\rangle_R=\liminf_{R \to \infty}\frac1{R}\int_0^R\rho\left(\frac{r}{R}\right)\left|\Ern(r)\right|^2\textup{d}r\ge \frac{\widehat{\rho}(0)}{2\pi^2}\sum_{\gamma \in S_0}\frac{\left|\ell(\gamma)\right|^2}{\gamma^3}\textrm{.}
\end{align*}
By considering an exhausting sequence of such subsets $S_0$, we get  \begin{align*}
    \liminf_{R \to \infty}\frac1{R}\int_0^R\rho\left(\frac{r}{R}\right)\left|\Ern(r)\right|^2\textup{d}r\ge \frac{\widehat{\rho}(0)}{2\pi^2}\sum_{\gamma \in S_\textrm{rad}}\frac{\left|\ell(\gamma)\right|^2}{\gamma^3}\textrm{.}
\end{align*} Combining this with \eqref{corybooker}, we get
\begin{align*}
    \liminf_{R \to \infty}\frac1{R}\int_1^R\left|\Ern(r)\right|^2\textup{d}r\ge \frac{\widehat{\rho}(0)}{2\pi^2}\sum_{\gamma \in S_\textrm{rad}}\frac{\left|\ell(\gamma)\right|^2}{\gamma^3}\textrm{.} 
\end{align*}
Letting $\rho(x) \to 1^{-}$, we see that $\widehat{\rho}(0)=\int_0^1\rho(x)\;\textup{d}x \to 1$, and so \begin{align*}
      \liminf_{R \to \infty}\frac1{R}\int_1^R\left|\Ern(r)\right|^2\textup{d}r\ge \frac{1}{2\pi^2}\sum_{\gamma \in S_\textrm{rad}}\frac{\left|\ell(\gamma)\right|^2}{\gamma^3}\textrm{.}
\end{align*}
The proof is complete.
\end{proof}
It now remains to prove Lemma \ref{besicovitchlemmathing}:
\begin{proof}[Proof of Lemma \ref{besicovitchlemmathing}]
    Define $I(R,\xi)$ by \begin{align}
        I(R,\xi)&\coloneqq\frac1{R}\int_0^R\rho\left(\frac{r}{R}\right)\Ern\left(r,\Lambda\right)e^{-2\pi i \xi r}\textup{d}r \label{IRxidef}\textrm{,}
    \end{align} so we want to show \begin{align*}
      \lim_{R \to \infty}  I(R,\xi)=\frac{\widehat{\rho}(0)\exp\left(-\frac{3\pi \sgn(\xi) i}{4}\right)}{2\pi}\cdot \frac{\ell\left(|\xi|\right)}{|\xi|^{\frac3{2}}}\textrm{.}
    \end{align*}
     By definition, $\Ern_t(r,\Lambda) = r^{-\frac1{2}}\Err_t(r,\Lambda)$, so from Proposition \ref{prop:pseudoperiod} we may write
    \begin{align}
        \Ern_t\left(r,\Lambda\right)=\sum_{\tau \in \{\pm 1\}}\sum_{\gamma \in S_{\textrm{rad}}}\frac{a_{\tau}\ell(\gamma)\Phi(t\gamma)e^{2\pi i \tau r \gamma}}{\gamma^{\frac3{2}}}+B(r,t)\textrm{,} \label{toref}
    \end{align}
    where $a_\tau={\exp\left(-\frac{3\pi \tau i}{4}\right)}/{2\pi}$ and $B(r,t)$ satisfies $\left|B(r,t)\right|\lesssim r^{-\frac1{2}}$ independently of $t$.
    We start by noting that $\lim_{t \to 0}\left(\mathbbm{1}_{B_r}*\varphi_t\right)(x)=\mathbbm{1}_{B_r}(x)$ for every $x$ such that $\lVert x\rVert \neq r$. Now, if $\left\{\lambda \in \Lambda\;\vert\;  \lVert \lambda\rVert=r\right\}=\varnothing$, then \begin{align*}
        \lim_{t \to 0}\left(\mathbbm{1}_{B_r}*\varphi_t\right)(\lambda)=\mathbbm{1}_{B_r}(\lambda)
    \end{align*} for each $\lambda \in \Lambda$. This implies $\lim_{t \to 0}\Ern_t\left(r,\Lambda\right)=\Ern\left(r,\Lambda\right)$ for almost every $r>0$, from the expression for $\Ern_t$ in \eqref{despacito}.
    From that same expression, we can see that
    \begin{align*}
        \left|\Ern_t\left(r,\Lambda\right)\right|\le \frac{K\left(1+r^2\right)}{r^{\frac1{2}}}
    \end{align*}
    for some constant $K > 0$, and
    \begin{align*}
        \frac1{R}\int_0^R\rho\left(\frac{r}{R}\right)\frac{1+r^2}{r^{\frac1{2}}}\;\textup{d}r<\infty\textrm{.}
    \end{align*}
    Then by the dominated convergence theorem,
    \begin{align}
   \frac1{R}\int_0^R\rho\left(\frac{r}{R}\right)\Ern\left(r,\Lambda\right)e^{-2\pi i \xi r}\textup{d}r= \lim_{t \to 0}\frac1{R}\int_0^R\rho\left(\frac{r}{R}\right)\Ern_t\left(r,\Lambda\right)e^{-2\pi i \xi r}\textup{d}r
        \label{domcont}
    \end{align} for each $\xi \in \mathbb{R}$. To that end, let \begin{align*}
        I_t(R,\xi)&=\frac1{R}\int_0^R\rho\left(\frac{r}{R}\right)\Ern_t\left(r,\Lambda\right)e^{-2\pi i \xi r}\textup{d}r\textrm{,} 
    \end{align*}
    so that $I\left(R,\xi\right)=\lim_{t \to 0}I_t\left(R,\xi\right)$. From the formula for $\Ern_t$ in \eqref{toref}, we have 
    \begin{align*}
        I_t\left(R,\xi\right)&=\sum_{\tau \in \{\pm 1\}}\sum_{\gamma \in S_{\textrm{rad}}} \frac1{R}\int_0^R\rho\left(\frac{r}{R}\right)\frac{a_{\tau}\ell(\gamma)\Phi(t\gamma)e^{2\pi i r\left(\tau \gamma-\xi\right) }}{\gamma^{\frac3{2}}}\;\textup{d}r+\frac1{R}\int_0^R\rho\left(\frac{r}{R}\right)B(r,t)e^{-2\pi i \xi r}\;\textup{d}r\textrm{,}
    \end{align*} 
    where the switch of the sum and the integral is justified since $t>0$ and $\Phi$ is of rapid decrease. Using the fact that $B(r,t)=O\left(r^{-\frac1{2}}\right)$ independently of $t$, we may simplify this to \begin{align*}
        I_t\left(R,\xi\right)&=\sum_{\tau \in \{\pm 1\}}\sum_{\gamma \in S_{\textrm{rad}}} \frac1{R}\int_0^R\rho\left(\frac{r}{R}\right)\frac{a_{\tau}\ell(\gamma)\Phi(t\gamma)e^{2\pi i r\left(\tau \gamma -\xi\right) }}{\gamma^{\frac3{2}}}\;\textup{d}r+O\left(R^{-\frac{1}{2}}\right) \textrm{.}
    \end{align*}
 We make the substitution $r=Rq$ to get \begin{align*}
        I_t\left(R,\xi\right)&=\sum_{\tau \in \{\pm 1\}}\sum_{\gamma \in S_{\textrm{rad}}}\frac{a_\tau \ell(\gamma)\Phi(t\gamma)}{\gamma^{\frac3{2}}}\int_0^1\rho(q)e^{2\pi i q R\left(\tau \gamma-\xi\right)}\textup{d}q+O\left(R^{-\frac{1}{2}}\right) \\
       &= \sum_{\tau \in \{\pm 1\}}\sum_{\gamma \in S_{\textrm{rad}}}\frac{a_\tau \ell(\gamma)\Phi(t\gamma)\widehat{\rho}\left(R\left(\xi-\tau \gamma \right)\right)}{\gamma^{\frac3{2}}}+O\left(R^{-\frac{1}{2}}\right)\textrm{,}
    \end{align*} since $\rho$ is supported on $(0,1)$. We extract the terms where $\xi=\tau \gamma$, so \begin{align*}
        I_t\left(R,\xi\right)&=\sum_{\tau \in \{\pm 1\}}\sum_{\substack{\gamma \in S_{\textrm{rad}} \\ \tau \gamma=\xi}}\frac{a_{\tau}\ell(\gamma)\Phi(t\gamma)\widehat{\rho}(0)}{\gamma^{\frac3{2}}}+\sum_{\tau \in \{\pm 1\}}\sum_{\substack{\gamma \in S_{\textrm{rad}} \\ \tau \gamma\neq \xi}}\frac{a_{\tau}\ell(\gamma)\Phi(t\gamma)\widehat{\rho}\left(R\left(\xi-\tau \gamma\right)\right)}{\gamma^{\frac3{2}}}+O\left(R^{-\frac1{2}}\right).
    \end{align*} The first sum above is finite, so we can take the limit termwise as $t \to 0$. Note $\rho$ is Schwartz because it is smooth and has compact support. But $\Phi$ is bounded and $\sum_{s \in S}\left|c_s\right|\delta_s$ is tempered, so \begin{align*}
        \sum_{\substack{\gamma \in S_{\textrm{rad}} \\ \tau \gamma \neq \xi}}\frac{\left|\ell(\gamma)\widehat{\rho}\left(R\left(\xi-\tau \gamma\right)\right)\right|}{\gamma^{\frac3{2}}}\le  \sum_{\substack{s \in S\setminus\{0\} \\ \tau \lVert s\rVert \neq \xi}}\frac{\left|c_s\widehat{\rho}\left(R\left(\xi-\tau \lVert s\rVert\right)\right)\right|}{\lVert s\rVert^{\frac3{2}}}<\infty\textrm{.}
    \end{align*}
    By the dominated convergence theorem, we may take the limit termwise as $t\to 0$ (recall $\Phi(0)=\widehat{\varphi}(0,0)=1$) and conclude \begin{align}
        I\left(R,\xi\right)&=\sum_{\tau \in \{\pm 1\}}\sum_{\substack{\gamma \in S_{\textrm{rad}} \\ \tau \gamma=\xi}}\frac{a_{\tau}\ell(\gamma)\widehat{\rho}(0)}{\gamma^{\frac3{2}}}+\sum_{\tau \in \{\pm 1\}}\sum_{\substack{\gamma \in S_{\textrm{rad}} \\ \tau \gamma\neq \xi}}\frac{a_{\tau}\ell(\gamma)\widehat{\rho}\left(R\left(\xi-\tau \gamma\right)\right)}{\gamma^{\frac3{2}}}+O\left(R^{-\frac1{2}}\right)\textrm{.}
        \label{qualude}
    \end{align}
    We now take the limit as $R \to \infty$. Since $\rho \in \mathcal{S}\left(\mathbb{R}\right)$, we have \begin{align*}
        \left|\frac{\ell(\gamma)\widehat{\rho}\left(R\left(\xi-\tau \gamma \right)\right)}{\gamma^{\frac3{2}}}\right|\lesssim_p \frac{\left|\ell(\gamma)\right|}{R^p\gamma ^{\frac3{2}}\left|\xi-\tau \gamma\right|^p},
    \end{align*} and by taking $p$ large enough, \begin{align*}
        \sum_{\substack{\gamma \in S_{\textrm{rad}} \\ \tau\gamma\neq \xi}}\frac{\left|\ell(\gamma)\right|}{\gamma ^{\frac3{2}}\left|\xi-\tau \gamma\right|^p}\le \sum_{\substack{s \in S\setminus\{0\} \\ \tau \lVert s\rVert\neq \xi}}\frac{\left|c_s\right|}{\lVert s\rVert ^{\frac3{2}}\bigl|\xi-\tau \lVert s\rVert\bigr|^p}<\infty.
    \end{align*}
    By the dominated convergence theorem, we can also take the limit termwise in $R$ in \eqref{qualude} and conclude \begin{align*}
        \lim_{R \to \infty}I\left(R,\xi\right)=\sum_{\tau \in \{\pm 1\}}\sum_{\substack{\gamma \in S_{\textrm{rad}} \\ \tau \gamma=\xi}}\frac{a_{\tau}\ell(\gamma)\widehat{\rho}(0)}{\gamma^{\frac3{2}}}=\frac{a_{\sgn(\xi)}\ell\left(|\xi|\right)\widehat{\rho}(0)}{|\xi|^{\frac3{2}}},
    \end{align*} where we take $a_{0}=0$. Recall $a_{\tau}={\exp\left(-\frac{3\pi i \tau}{4}\right)}/{2\pi}$, so from the definition of $I(R,\xi)$ in \eqref{IRxidef}, we can write this as \begin{align*}
     \lim_{R \to \infty}\frac1{R}\int_0^R\rho\left(\frac{r}{R}\right)\Ern\left(r,\Lambda\right)e^{-2\pi i \xi r}\textup{d}r=\frac{\widehat{\rho}(0)\exp\left(-\frac{3\pi \sgn(\xi) i}{4}\right)}{2\pi}\cdot \frac{\ell\left(|\xi|\right)}{|\xi|^{\frac3{2}}}\textrm{.} 
\end{align*} This proves the lemma.
\end{proof}
% We developed the proof of this lower bound in our attempt to prove that $\Ern$ is Besicovitch almost-periodic, similar to the theory in \cite{besicovitch1926generalized}. The space $B$ of Besicovitch almost-periodic functions  is the set of measurable functions $f\colon[0,\infty) \to \mathbb{C}$ for which there exists a sequence of trigonometric polynomials $p_1,p_2, \ldots$ such that 
% \begin{align*}
%    \lim_{n \to \infty} \lim_{R \to \infty}\frac1{R}\int_0^R\rho\left(\frac{r}{R}\right)\left|f(x)-p_n(x)\right|^2\,dr=0\textrm{.}
% \end{align*}
% A trigonometric polynomial $p(x)$ is a function of the form \begin{align*}
%     p(x)=\sum_{j=1}^Na_je^{i\lambda_jx}
% \end{align*} for real $\lambda_j$. When $\Lambda$ is a lattice, the results in \cite{bleher1993distribution} imply that the normalised error function $\Ern$ is Besicovitch almost-periodic (with $\rho=\mathbbm{1}_{[0,1]}$), and the inequality in Proposition \ref{lowerboundaverageerror} is an equality, which is to say if $\Lambda=\mathbb{Z}^2$, then the corresponding normalised error term $\Ern$ satisfies \begin{align*}
%     \lim_{R \to \infty}\frac1{R}\int_0^R\left|\Ern(r)\right|^2dr=\frac1{2\pi^2}\sum_{\gamma \in S_\textrm{rad}}\frac{\left|\ell(\gamma)\right|^2}{\gamma^3}\textrm{.}
% \end{align*}
% We believe that a similar thing could hold in our case, and as a corollary we would get \begin{align*}
%     \lim_{R \to \infty}\frac1{R}\int_0^R\rho\left(\frac{r}{R}\right)\left|\Ern(r)\right|^2dr=\frac1{2\pi^2}\sum_{\gamma \in S_\textrm{rad}}\frac{\left|\ell(\gamma)\right|^2}{\gamma^3}\textrm{.}
% \end{align*}

% REFERENCES
{ \small 
	\bibliographystyle{plain}
	\bibliography{biblio.bib} }

\appendix
% \section{Alternative proof of Proposition \ref{lowerboundaverageerror}}

%
\section{Other numerics}\label{appendixA}

Here we give more numerics on more nontrivial Fourier quasicrystals. Specifically, consider the functions
\begin{align*}
    f(x) &= \sin\left(\pi\left(\frac{1-\sqrt{2}}{2}\right)x\right) - 4\sin\left(\pi\left(\frac{1+\sqrt{2}}{2}\right)x\right)\textrm{,} \\
    g(y) &= 2\cos\left(\pi\sqrt{5}y\right)\sin\left(\pi\left(\frac{\sqrt{3}}{2}\right)y\right) + \cos\left(\pi\left(\frac{\sqrt{3}}{2}\right)y\right)\sin\left(\pi\sqrt{5}y\right)\textrm{,}
\end{align*}
and their zero sets
\begin{align*}
    \Lambda_f &= \{x \in \mathbb{R}\;\vert\;f(x) = 0\}\textrm{,} \\
    \Lambda_g &= \{y \in \mathbb{R}\;\vert\;g(y) = 0\}\textrm{.}
\end{align*}
Then the set
\begin{equation}\label{eqn:appendix-product}
    \Lambda = \Lambda_f \times \Lambda_g
\end{equation}
is a Fourier quasicrystal. We call this special case a \emph{product Fourier quasicrystal}; this can be seen from the construction of $\Lambda$, as it is a product of two zero sets of functions that only depend on one variable. Both $\Lambda_f$ and $\Lambda_g$ are Fourier quasicrystals from \cite{kurasov2020stable}, and their product is a Fourier quasicrystal as well.

Figure \ref{fig:product} shows that the error again grows like $O\left(r^{1/2}\right)$.

\begin{figure}[ht]
    \centering
    \includegraphics[width=0.9\linewidth]{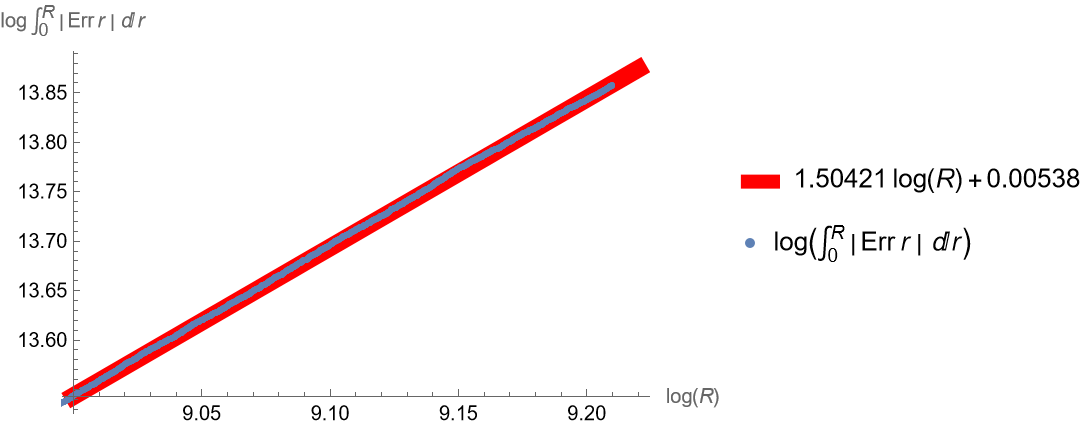}
    \caption{The $\log$-$\log$ plot of $\Err(r,\Lambda)$, where $\Lambda$ is a product Fourier quasicrystal as in \eqref{eqn:appendix-product}, for $R < 10000$}
    \label{fig:product}
\end{figure}

Figure \ref{fig:product-histogram} shows the normalized error $\Err(r,\Lambda)/r^{1/2}$ for $r \leq 10000$. The distribution appears like a normal distribution, similar to Bleher \cite{bleher1996distribution}.

\begin{figure}[ht]
    \centering
    \includegraphics[width=0.7\linewidth]{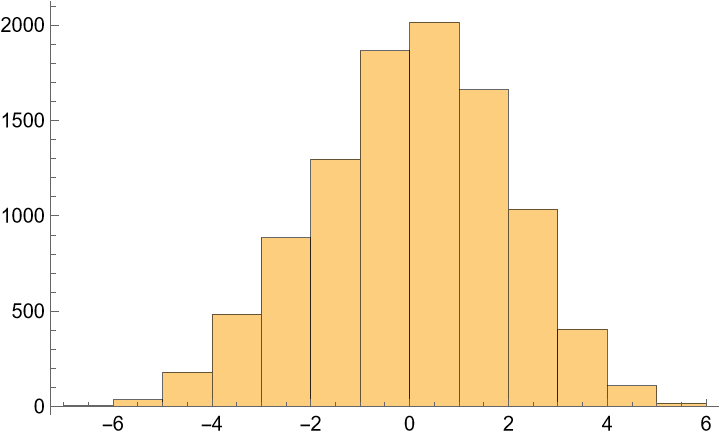}
    \caption{A histogram of the normalized error $\Err(r,\Lambda)/r^{1/2}$ for $\Lambda$ as in \eqref{eqn:appendix-product}}
    \label{fig:product-histogram}
\end{figure}

Next we consider
\begin{align}\label{eqn:appendix-nonproduct}
    \begin{split}
        \Lambda &= \left\{(x,y) \in \mathbb{R}^2 \;\vert\; f(x,y) = 0, \ g(x,y) = 0\right\}\textrm{,} \\
        f(x,y) &= \sin\left(\pi\left(\frac{x-y}{2}\right)\right) - 4\sin\left(\pi\left(\frac{x+y}{2}\right)\right)\textrm{,} \\
        g(x,y) &= 2\cos\left(\pi\left(-\sqrt{7}x+\sqrt{5}y\right)\right)\sin\left(\frac{\pi}{2}\left(-\sqrt{2}x+\sqrt{3}y\right)\right) \\
        &\quad + \cos\left(\frac{\pi}{2}\left(-\sqrt{2}x+\sqrt{3}y\right)\right)\sin\left(\pi\left(-\sqrt{7}x+\sqrt{5}y\right)\right)\textrm{.}
    \end{split}
\end{align}

This Fourier quasicrystal comes from the constructions of Alon, Kummer, Kurasov, and Vinzant \cite{alon2024higherdimensionalfourierquasicrystals}. Specifically, from their construction we take the matrix
\begin{equation*}
    L = \begin{pmatrix}
        1 & 0 & -\sqrt{2} & -\sqrt{5} \\
        0 & 1 & \sqrt{3}  &  \sqrt{7}
    \end{pmatrix}
\end{equation*}
and the algebraic variety $X \subseteq \mathbb{C}$ defined by
\begin{align*}
    p_1(z_1,z_2,z_3,z_4) &= -3 - z_1 + z_2 + 3z_1z_2 = 0\textrm{,} \\
    p_2(z_1,z_2,z_3,z_4) &= -3 - z_3 + z_4 + 3z_3z_4 = 0\textrm{,}
\end{align*}
so that $\Lambda$ in \eqref{eqn:appendix-nonproduct} is precisely defined by $\Lambda = \left\{x \in \mathbb{R}^2\;\vert\;\exp(2\pi iLx) \in X\right\}$.

Figure \ref{fig:nonproduct} again shows that the error grows like $O\left(r^{1/2}\right)$.

\begin{figure}[ht]
    \centering
    \includegraphics[width=0.9\linewidth]{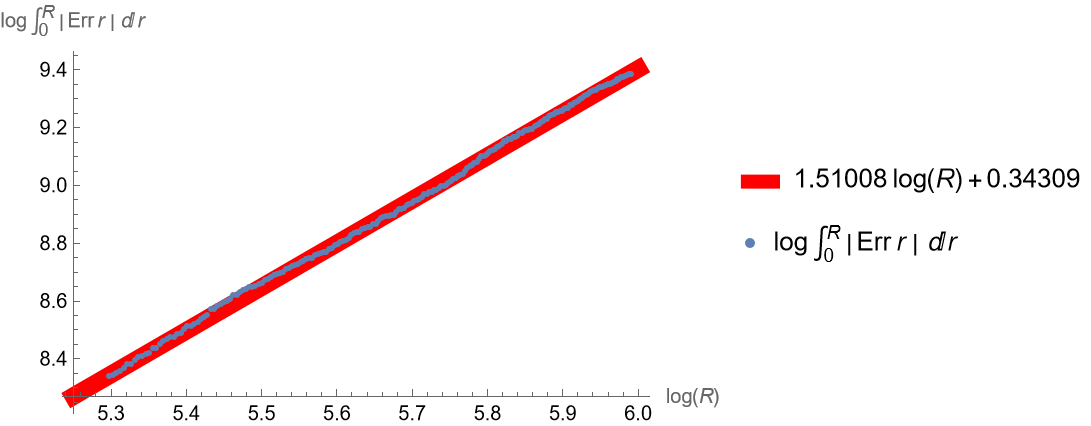}
    \caption{The $\log$-$\log$ plot of $\Err(r,\Lambda)$, where $\Lambda$ is a nonproduct Fourier quasicrystal as in \eqref{eqn:appendix-nonproduct}, for $R < 500$}
    \label{fig:nonproduct}
\end{figure}

Figure \ref{fig:product-non-histogram} shows the normalized error $\Err(r,\Lambda)/r^{1/2}$ for $r \leq 500$. Again like Bleher \cite{bleher1996distribution}, the distribution appears normal.

\begin{figure}
    \centering
    \includegraphics[width=0.7\linewidth]{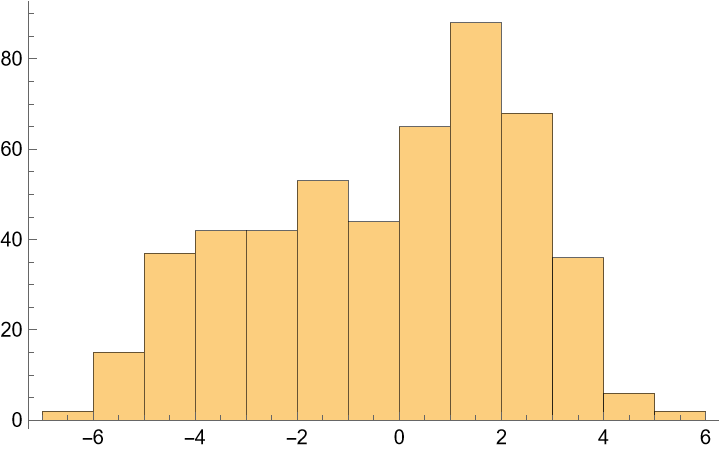}
    \caption{A histogram of the normalized error $\Err(r,\Lambda)/r^{1/2}$ for $\Lambda$ as in \eqref{eqn:appendix-nonproduct}}
    \label{fig:product-non-histogram}
\end{figure}
 
\end{document}